\documentclass[a4paper,onecolumn,11pt]{quantumarticle} 
\pdfoutput=1
\usepackage[utf8]{inputenc}
\usepackage[english]{babel}
\usepackage[T1]{fontenc}
\usepackage{amsmath,amssymb,amsthm,mathtools}
\usepackage{braket}
\usepackage{hyperref}

\usepackage{tikz}
\usepackage{lipsum}
\usepackage[numbers]{natbib}
\usepackage{ulem}

\usepackage{amsthm}
\usepackage{xcolor}
\usepackage{tcolorbox}
\usepackage{dsfont}
\usepackage{todonotes}
\usepackage{import}
\usepackage{soul}
\usepackage{float}
\usepackage{xspace}

\newcommand{\mancfoot}[1]{%
  \insert\footins{%
    \everycr{}%
    \font\smallfnt=cmr10 at 8pt\smallfnt%
    \vskip  2pt%
    \noindent\textsuperscript{#1}%
  }%
}

\newcommand{\eqdef}{\ensuremath{\stackrel{\mbox{\upshape\tiny def.}}{=}}}

\definecolor{main}{HTML}{5989cf}    
\definecolor{sub}{HTML}{cde4ff}     

\tcbset{
    sharp corners,
    colback = white,
    before skip = 0.2cm,    
    after skip = 0.5cm      
}                           

\newtcolorbox{blue_box}[1]{
    title = #1,
    fonttitle = \bf,
    colback = sub, 
    colframe = main, 
    boxrule = 0pt, 
    toprule = 3pt, 
    bottomrule = 3pt 
}

\newcounter{protocol}

\newtheorem{theorem}{Theorem}[section]
\newtheorem{lemma}[theorem]{Lemma}
\newtheorem{proposition}{Proposition}
\newtheorem{criterion}{Criterion}

\theoremstyle{remark}
\newtheorem{remark}[theorem]{Remark}
\newtheorem{example}[theorem]{Example}

\theoremstyle{definition}
\newtheorem{definition}[theorem]{Definition}

\newcommand{\Pos}{\mathrm{Pos}}
\newcommand{\CP}{\mathrm{CP}}

\newcommand{\Tr}{\text{Tr}}
\newcommand{\Hil}{\mathcal{H}}
\newcommand{\id}{\mathds{1}}
\newcommand{\SOS}{\text{SOS}\xspace}
\newcommand{\CSOS}{\text{CSOS}\xspace}
\newcommand{\RSOS}{\text{RSOS}\xspace}
\newcommand{\nonSOS}{\text{non-SOS}\xspace}

\newcommand{\ie}{\textit{i.e.}\xspace}

\newcommand{\RR}{\mathbb{R}}
\newcommand{\NN}{\mathbb{N}}

\begin{document}

\title{Detecting bipartite entanglement with PnCP maps and non-negative polynomials}

\author{Gaël Massé \textsuperscript{*}}
\affiliation{LIP6, Sorbonne Université,75005 Paris, France}
\affiliation{ICFO-Institut de Ciencies Fotoniques, Mediterranean Technology Park, 08860 Castelldefels (Barcelona), Spain}

\author{Mounir Rezig \textsuperscript{*}}
\affiliation{LIP6, Sorbonne Université,75005 Paris, France}
\affiliation{Thales SIX GTS France, 92230 Gennevilliers, France}
\email{mounir.rezig@lip6.fr}

\mancfoot{* These two authors contributed equally}

\author{Paul Catala}
\affiliation{CRAN,Université de Lorraine, CNRS, Vandoeuvre-lès-Nancy,France}
\author{Santiago Scheiner}
\affiliation{LIP6, Sorbonne Université,75005 Paris, France}
\author{Laia Serradesanferm Córdoba}
\affiliation{ICFO-Institut de Ciencies Fotoniques, Mediterranean Technology Park, 08860 Castelldefels (Barcelona), Spain}
\author{Enky Oudot}
\affiliation{LIP6, Sorbonne Université,75005 Paris, France}
\affiliation{ICFO-Institut de Ciencies Fotoniques, Mediterranean Technology Park, 08860 Castelldefels (Barcelona), Spain}
\author{Damian Markham}
\affiliation{LIP6, Sorbonne Université,75005 Paris, France}

\orcid{0009-0008-4834-0641}

\maketitle

\begin{abstract}
  Positive non-Completely Positive (PnCP) maps are an essential tool to detect entanglement since their characterization is a dual aspect of the separability problem. A recent algorithm proposed in~\cite{Klep2017ManyMorePositiveMaps} explains how to generate PnCP maps based on the construction of certain positive non-Sum-of-Squares polynomials. We implement this algorithm in a numerically robust way  and propose a working version on GitHub. We theoretically demonstrate that the maps produced by the algorithm are indecomposable, localized on the boundary of the positive cone and show that they are inequivalent with most other known PnCP maps. We numerically investigate their entanglement power, demonstrating notably that they are capable of detecting PPT entangled states that most criteria fail to detect. 
\end{abstract}

\section*{Introduction}
\label{sec:introduction}

A particularly striking aspect of quantum theory is the way the state of a composite system relates to the description of its individual parts. This peculiarity, already highlighted in the foundational works of the theory, originates in the fact that the Hilbert space of a composite system is given by the tensor product of the Hilbert spaces of its subsystems. This tensor-product structure makes possible the existence of genuinely non-classical correlations between different parts of a composite quantum system. In this work, we focus on the manifestation of this non-classical aspect called entanglement in the case of bipartite systems, that is quantum states whose description require the tensorial product space of exactly two subsystems. 

Entanglement is a central feature of quantum physics and a key resource for quantum information and computation~\cite{Horodecki2009,Chitambar2019}. Accordingly, the characterization of the set of entangled states, also called the separability problem, has been a central issue for quantum scientists, both from a theoretical and a practical perspective.  Remarkably, an important work proposed a family of separability criteria that can be cast in the form of an hierarchy of Semi Definite Programs (SDP) that is complete, \ie that any entangled state must be detected at a finite level of the hierarchy~\cite{Doherty2004CompleteFamily}. Nonetheless, this criterion is not necessarily suited for applications because of its exploding computational cost. Indeed, it has been shown that not only the separability problem was NP-hard~\cite{Gharibian2009Strong_NP_hardness}, but also that it could not be solved with a SDP in polynomial time for bipartite quantum states acting on Hilbert spaces of dimension larger than $3\times 2$~\cite{fawzi2019setseparablestatesfinite}.

Thus, a large number of different techniques have been introduced to address the problem from different perspectives. There are several notable criteria that can be quoted, including the famous Positive Partial Transpose (PPT) criterion~\cite{Peres1996SeparabilityCriterion} that completely solves the problem in low dimensions, the reduction criterion~\cite{Horodecki1999ReductrionCriterion}, the realignment criterion~\cite{chen2003MatrixRealignement} and many others, such as those listed in~\cite{Guhne2009EntanglementDetection}. Different benchmarks can coexist to evaluate the usefulness of a criterion, but in order to potentially extend the set of detected states, a criterion should be able to detect so called PPT entangled states, those for which the handy Partial Transpose fails, before anything else. In this work, we address the separability problem by implementing and analysing a new entanglement detection criterion. 

The usage of PnCP maps to detect entanglement has been acknowledged for a certain time now, several useful maps having been successfully derived notably to detect PPT Entangled states~\cite{Kossakowski2003ClassLinear,Ha2003ClassAtomic,Breuer2006OptimalEntanglement,Hall2006NewCriterion}. In order to construct new PnCP maps, a powerful link between PnCP maps and polynomials that are positive but can not be written as sum-of-squares (SoS) has recently been recognized. It has started to be used by the quantum information community~\cite{Ohst2025RevealingHidden,gulati2025positivemapsextendibilityhierarchies}. Multiple connections had however been demonstrated previously between moment-SoS representations and the separability problems~\cite{Gribling2022SeparableRank,Bohnet2017EntanglementTMP}. The existence of such polynomials has been known since 1888, when David Hilbert demonstrated their existence~\cite{Hilbert1888}, and has since been an important field of study, which gained new momentum in the last decades following important applications in the field of global optimization~\cite{Lasserre2001GlobalOptimization, blekherman2012semidefinite, Laurent2008sumssquares}. In particular, these works were translated for non-commutative algebras with specific purpose for quantum information~\cite{Navascues2007BoundingtheSet, Navascues2008Acconvergenthierachy,SupicBowles2020}. More broadly, moment-SOS techniques provide a common SDP language for polynomial optimization and quantum information\cite{Lasserre2001GlobalOptimization,Parrilo2003Semidefiniteprogramming,Laurent2008sumssquares}. In the separability problem, the DPS hierarchy can be viewed through its dual SOS formulation~\cite{Doherty2004CompleteFamily,Fang2020SumSquares}. A parallel while the NPA hierarchy translates the same moment-matrix idea to non-commutative polynomial algebras generated by quantum measurement operators \cite{Navascues2007BoundingtheSet,Navascues2008Acconvergenthierachy,Pironio2010NCPO}. In this sense, polynomial hierarchies connect separability relaxations, quantum correlation constraints, and SOS certificates \cite{Helton2002PositiveNC,Pironio2010NCPO,Fang2020SumSquares}. The present work uses this link in a complementary direction: rather than only approximating positivity by SOS relaxations, we start from explicit positive non-SOS biquadratic forms and convert them into PnCP maps \cite{Klep2017ManyMorePositiveMaps}.

Building on an important work establishing the respective volumes of positive and sum of square polynomials~\cite{blekherman2006thereare}, a recent work showed that there are many more Positive maps than Completely Positive~\cite{Klep2017ManyMorePositiveMaps}. Decisively, this last article presents an algorithm (KSMZ) to construct positive but non sum-of-square polynomials along with the mathematical transformation to turn them into PnCP maps. An inspiring implementation of this algorithm has been presented in a recent article~\cite{Bhardwaj2023PracticalApproach} but it suffers from numerical imprecisions leading to some false identification of entanglement. 

In this article, we present in section~\ref{sec:algoimpl} our implementation of the KSMZ article, that is accessible on the platform Github~\cite{GithubCode}. It is able to generate reliable positive non \SOS polynomials and the associated PnCP maps thanks to a new numerical test that guarantees positivity. Furthermore, we analyze in depth the properties of the algorithm in its original form, notably exploiting the duality~\cite{Fang2020SumSquares} between the Doherty-Parrilo-Spedalieri (DPS) hierarchy~\cite{Doherty2004CompleteFamily} and the Sum-of-Squares (SOS) hierarchy~\cite{Lasserre2001GlobalOptimization} in section~\ref{sec:Maps}. The distinction it establishes between so-called real sum-of-squares (\RSOS) and complex sum-of-squares polynomials (\CSOS) enables us to classify the KSMZ maps. In particular, it allows us to demonstrate that they are capable of detecting PPT entangled states.

For some mathematical protocols, notably for Semi-Definite-Programming (SDP) optimization problems used for entanglement detection, it is necessary to work with the observables called Entanglement Witnesses (EW)~\cite{Chruscinski2014EntanglementWitnesses} rather than directly with the PnCP maps. Indeed, using the PnCP maps directly gives rise to non-convex constraints that are incompatible with SDP formulations. EWs are deduced from the maps through the famous Choi-Jamiołkowski isomorphism~\cite{Choi1975CompletelyPositive,Jamilokowski1972Lineartransformations}, explained in detail in~\cite{Min2013ChannelState}, and their study is dual to that of the entangled states since for every entangled state there exists an EW that is capable of detecting it~\cite{Horodecki1996SeparabilityMixed}. Geometrically, a witness can be interpreted as a hyperplane that separates some entangled states from the convex set of separable states~\cite{Dirkse2020WitnessingEntanglementGeometry}. Thus, the study of entanglement detection can benefit from numerous results originating from the field of matrix algebras~\cite{Stormer2013PositiveLinear}.

In section~\ref{sec:EntanglementDetection}, we present the numerical results of the entanglement detection led with EW emanating from the KSMZ maps. We show that it can detect NPT and PPT entangled states, benchmark it against other criteria, and analyse its limits for practical purposes in terms of noise robustness. We finally conclude in section~\ref{sec:conclusion}. We start by introducing the elements of entanglement and polynomial theory needed in section~\ref{sec:entanglement_theory}.

\section{Entanglement Theory}
\label{sec:entanglement_theory}

\subsection{Entanglement, PnCP maps and Witnesses}

In this paper we restrict ourselves to bipartite systems, \ie quantum states, $\rho$, acting on a product of two Hilbert spaces $\mathcal{H}_A \otimes \mathcal{H}_B$.
A state is called separable when it can be factorized, that is when there exist $\{ \lambda_i \} \in [0,1]$, $ \{ \rho_i^A, \rho_i^B \} \in \mathcal L(\mathcal{H})$ such that 

\begin{definition}[Separability and entanglement]
A state is called separable when it can be factorized, that is when there exist $\{ \lambda_i \} \in [0,1]$, $ \{ \rho_i^A, \rho_i^B \} \in \mathcal L(\mathcal{H})$ such that 

\begin{equation}
    \rho = \sum_i \lambda_i \rho_i^A \otimes \rho_i^B 
\end{equation} with $\sum_i \lambda_i = 1$. A state is entangled when it is not separable. 
\end{definition}
For the following we denote the set of separable states as $\mathcal{SEP}$
The task we pursue consists in developing a new entanglement detection tool. It is based on the theory of Positive-non-Completely-Positive maps. 

\begin{definition}[Positivity of a map]
     \begin{align}
    \Lambda:= \left\{
\begin{array}{lll}
        \mathcal L(\mathcal{H}_A) & \rightarrow & \mathcal L(\mathcal{H}_B) \\
                     \rho & \mapsto  & \Lambda(\rho) \end{array}
                     \right. \nonumber
    \label{eq:posmap}
\end{align}
$\Lambda$ is positive if for any $\rho \geq 0$, $\Lambda(\rho) \geq 0$.
\end{definition}

A map $\Lambda$ is called Completely Positive if, in any dimension $d$, $\mathds{1}_d \otimes \Lambda$ is a positive map. In the framework of quantum theory, \textit{completely positive, trace-preserving} (CPTP) maps correspond to physical evolutions. \textit{Positive but not Completely Positive} (PnCP) maps, on the other hand, do not represent natural transformations, but are solely mathematical tools that, in this context, are aimed at detecting entanglement. This is due to the following criterion.

\begin{criterion}[PnCP Criterion]
If $\Lambda$ is a Positive non-Completely Positive map (PnCP) and if $\rho$ is a quantum state,
\begin{equation}
    (\mathds{1}_d \otimes \Lambda) (\rho) \ngeq 0 \Rightarrow \rho \text{ entangled } \nonumber
\end{equation}
\label{crit:PnCP}
\end{criterion}
The fact that for any entangled state, there exists a PnCP map that is able to detect it~\cite{Chruscinski2014EntanglementWitnesses}, establishes that the study of PnCP maps is dual to that of entangled states. Since PnCP maps are not physically implementable though, at first glance they require a full tomography, that is the complete knowledge of a quantum state $\rho$ to detect entanglement. Fortunately, there exist families of implementable observables whose detection capacity are equivalent to PnCP maps and do not require a full tomography~\cite{Chruscinski2014EntanglementWitnesses}. They are called Entanglement Witnesses. An Entanglement Witness (EW) is a block-positive but not positive operator. The formal definition is given below. 
\begin{definition}[Positivity of an operator]
A Hermitian operator $W \in \mathcal{L}(H)$ is \textit{positive} iff
\begin{equation}
\langle \Psi \vert W \vert \Psi \rangle \ge 0 ,
\end{equation}
for all vectors $\Psi \in \mathcal{H}$.
\label{def:posop}
\end{definition}

\begin{definition}[Block-positivity]
A Hermitian operator $W \in \mathcal{L}(H)$ is \textit{block-positive} iff
\begin{equation}
\langle \psi \otimes \phi \vert W \vert \psi \otimes \phi \rangle \ge 0 ,
\end{equation}
for all product vectors $\psi \otimes \phi \in \mathcal{H}$.
\label{def:blockposop}
\end{definition}

Geometrically speaking, an EW defines an hyperplane within the space of quantum states, separating detected entangled states to all others \cite{BengtssonZyczkowski2006GeometryQuantum}, as depicted in Figure~\ref{fig:EWgeometry}. 

\begin{figure}
  \centering
    \includegraphics[width=0.7\linewidth]{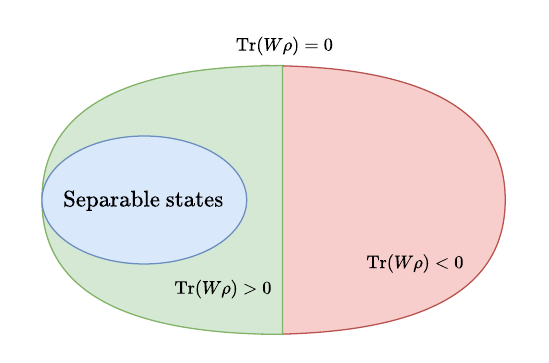}
    \caption{Pictorial view of the topology of separable and entangled states. The separable states form a convex set depicted in light blue, and the entangled states surround them in green and red.  The black line corresponds to the hyperplane for which $\Tr(W \rho) = 0$. The states are accordingly separated in two distinct areas, those for which the mean value of $W$ is positive and those for which it is negative. The entangled states that are detected by the witness are in red, those that are not are in green.}
    \label{fig:EWgeometry}
\end{figure}

An equivalent definition of an EW is that its expectation value is positive for all separable states, and negative for at least one entangled state. 

\begin{definition}[Entanglement Witness]
An operator $W \in \mathcal{L}(\Hil_A \otimes \Hil_B)$ is an \textit{entanglement witness (EW)} 
iff the following conditions hold:
\begin{itemize}
    \item $\Tr(W \rho_{\mathrm{sep}}) \ge 0$ for all separable states 
    $\rho_{\mathrm{sep}} \in \mathcal{S}_{\mathrm{sep}}$,
    \item there exists at least one entangled state 
    $\rho_{\mathrm{ent}} \notin \mathcal{S}_{\mathrm{sep}}$ such that
    \[
    \Tr(W \rho_{\mathrm{ent}}) < 0 .
    \]
\end{itemize}
\end{definition}
Interestingly, EWs can be robust to noise: although generally tailored to detect a targeted state, they can still detect a mixture of it with some noise, up to a certain threshold~\cite{Masse2020ImplementableHybrid,Bourennane2004ExperimentalDetection}.

The link between positive maps and operators is established by the celebrated Choi-Jamiołkowski isomorphism~\cite{Choi1975CompletelyPositive,Jamilokowski1972Lineartransformations}, for whose definition we first introduce some preliminary concepts, as follows.
Let $\Phi:~\mathcal{L}(\Hil_A)~\to~\mathcal{L}(\Hil_B)$ be a linear map and let $\{ \ket{i} \}_{i=1}^n$ be an orthonormal basis of $\Hil_A$, with $n = \dim \Hil_A$, and the maximally entangled vector
\[
\ket{\Omega} = \frac{1}{\sqrt{d}} \sum_{i=1}^n \ket{i} \otimes \ket{i} \in \Hil_A \otimes \Hil_A .
\]

\begin{definition}[Choi-Jamiołkowski isomorphism]

The \textit{Choi operator} (or Choi matrix) of $\Phi$ is defined by
\begin{equation}
C_\Phi = d( \id \otimes \Phi)(\ket{\Omega}\bra{\Omega})
= \sum_{i,j=1}^d \ket{i}\bra{j} \otimes \Phi(\ket{i}\bra{j}) .
\end{equation}
where $d = \dim(\Hil_A)$.
The correspondence $\Phi \mapsto C_\Phi$ is called the 
\textit{Choi-Jamiołkowski isomorphism}. 
\label{def:CHOI}
\end{definition}
This isomorphism maps Positive maps to block-positive operators, and Completely Positive (CP) maps to positive operators. Accordingly, the Choi operator of a map $\Phi$ is an EW if and only if $\Phi$ is a PnCP map. This is illustrated in Figure~\ref{fig:ChoiIsomorphism}.

\begin{figure*}[t!]
    \centering
    \includegraphics[width=1\linewidth]{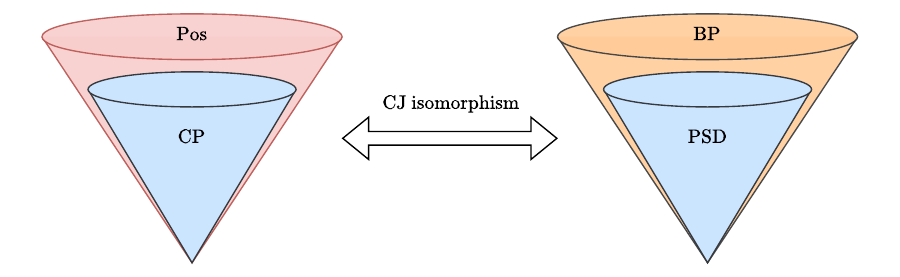}
    \caption{Choi-Jamiołkowski isomorphism between cones of maps and cones of operators. Positive maps form a cone $\Pos$ containing the completely positive cone $\CP$; under the Choi isomorphism, these correspond respectively to the cone of block-positive operators $\text{BP}$ and its positive semidefinite subcone $\text{PSD}$. Thus, completely positive maps give positive semidefinite Choi operators, while positive but not completely positive maps correspond to block-positive operators that are not positive semidefinite, \ie entanglement witnesses.}
\label{fig:ChoiIsomorphism}
\end{figure*}

\subsection{Positive polynomials and entanglement}

Central to this article is the link that exists between PnCP maps and some positive polynomials.

\subsubsection{Positive non \SOS Polynomials}

Certifying that a polynomial $p(x)\ge 0$ over $\mathbb{R}^n$ is a crucial problem in polynomial optimization. A sufficient condition is the existence of a sum-of-squares (SOS) decomposition, presented in Definition~\ref{def:sos}.

\begin{definition}
Let $n \in \mathbb N$ and $p \in \mathbb R[x_1, \dots, x_n]$ be a real polynomial . We say that $p$ is a $\SOS$ polynomial if there exist real polynomials $q_1, \dots, q_r \in \mathbb R[x_1, \dots, x_n]$, $r\in \mathbb N$, such that 
\begin{equation}
    p(x)=\sum_{i=1}^rq_i(x)^2
\end{equation}
    \label{def:sos}
\end{definition}

 Any such representation certifies nonnegativity. The converse is false in general: not every nonnegative polynomial admits an SOS decomposition, as the Hilbert’s 17th problem presents~\cite{Hilbert1888}. However, as originally suggested by Hilbert, a result by Artin shows that 
 any nonnegative polynomial is a quotient of sum-of-squares, \textit{i.e.} if $p$ is nonnegative, then 
%
there exists an SOS multiplier $s(x)$ such that $s(x)\,p(x)=\sum_i q_i(x)^2$.
Sharper SOS representations also exist when one considers nonnegative polynomials over (compact) semi-algebraic sets~\cite{putinar1993positive}. Such properties have a high numerical interest, since checking if a polynomial is SOS amounts to solving a semidefinite program, while checking nonnegativity remains NP-hard. 
SOS representations are thus used to compute 
tractable and arbitrarily tight lower bounds for global polynomial optimization problems~\cite{Lasserre2001GlobalOptimization, Parrilo2003Semidefiniteprogramming}, 
which has a broad range of applications ranging from quantum mechanics to mathematical finance~\cite{Laurent2008sumssquares, lasserre08, blekherman2012semidefinite}.


\subsubsection{Positive polynomials and PnCP maps}

For our purposes, positive non-\SOS polynomials are interesting insofar as they correspond to PnCP maps and EWs.

To this end,~\cite{Klep2017ManyMorePositiveMaps} proposes the following isomorphism between the vector space of linear maps from $S_n(\mathbb R)$ to $S_m(\mathbb R)$ and biforms of bidegree $(2,2)$ that are bihomogeneous polynomials of degree 2 in each set of variables :
\begin{align}
\Gamma : &\mathcal{L}(S_n(\mathbb{R}), S_m(\mathbb{R})) \to \mathbb{R}[x,y]_{2,2} \nonumber \\
&\Phi \mapsto p_\Phi(x,y) := \bra{y}\,\Phi\big(\ket{x}\bra{x}\big)\ket{y}. 
\label{eq:KSMZIsomorphism}
\end{align}
This is illustrated in Figure~\ref{fig:KSMZIsomorphism}.

   \begin{figure*}[t!]
        \centering  \includegraphics[width=1\linewidth]{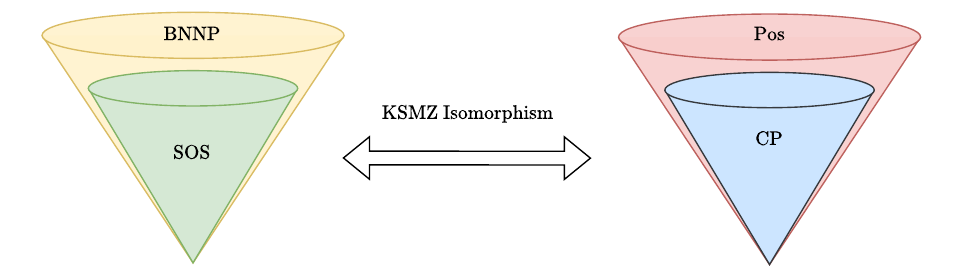}
        \caption{KSMZ isomorphism between nonnegative biquadratic forms and positive maps. The cone $\text{BNNP}$ of biquadratic nonnegative polynomials contains the subcone \SOS of sums of squares. Through the biquadratic-form/positive-map correspondence, $\text{BNNP}$ is identified with the cone $\text{Pos}$ of positive maps, while the \SOS subcone corresponds to completely positive maps $\text{CP}$. Hence, nonnegative biquadratic forms that are not sums of squares give rise to positive but not completely positive maps}
        \label{fig:KSMZIsomorphism}
    \end{figure*}


Accordingly, the positive non-\SOS polynomials that we consider here are
biforms of bidegree $(2,2)$ which are defined just below.

\begin{definition}[Bihomogeneous polynomial]
A polynomial $p(x,y)$ is bihomogeneous of bidegree $(d_1,d_2)$ if every
monomial of $p$ has total degree $d_1$ in the variables $x$
and total degree $d_2$ in the variables $y$.
\end{definition}

Positivity and complete positivity of $\Phi$ translate into nonnegativity and \SOS decompositions of the polynomial $p_\Phi$, which is formally proven in~\cite{Klep2017ManyMorePositiveMaps}. The resulting map is then extended to the whole space of real matrices, by putting the skew-symmetric part to zeros~\cite{Klep2017ManyMorePositiveMaps}. This decomposition is unique since 
\begin{equation}
M_n(\mathbb{R}) = S_n(\mathbb{R}) \oplus K_n(\mathbb{R}) \label{eq:direct_sum}
\end{equation}
where $K_n(\mathbb{R})$ is the skew-symmetric real matrices vector space. 

In order to define the complexification of this map, we need the following definition:
\begin{definition}
A linear map {$\Phi : \mathcal L(\mathcal H_A) \to \mathcal L(\mathcal H_B)$} between matrix spaces is said to be
$\dagger$-linear if
\[
\Phi(A^\dagger) = \Phi(A)^\dagger
\quad \text{for all } A \in \mathcal H_A.
\]
\end{definition}

The maps acting on the real matrices are extended to a $\dagger$-linear map $\Phi_c$ belonging to $\mathcal{L}\big(M_n(\mathbb{C}), M_m(\mathbb{C})\big)$ through the following construction
\[
\Phi_c (A+iB) = \Phi(A) + i \Phi(B)
\]
for all $A$ and $B$ are real matrices from $M_n(\mathbb{R})$. All of this is explained in~\cite{Klep2017ManyMorePositiveMaps} alongside with the proof that the extended maps are still PnCPs. This isomorphism between positive non \SOS polynomials and PnCP maps is summed-up in Figure~\ref{fig:KSMZIsomorphism}.

\subsubsection{Positive polynomials and Entanglement Witness}

A direct link between polynomial forms and witnesses has been acknowledged as early as 2004, as can be seen in~\cite{Doherty2004CompleteFamily}. From an entanglement witness $W$ it is possible to construct a corresponding polynomial, considering 
\begin{align}
E_W(x,y) &= \langle xy | W | xy \rangle
= \operatorname{Tr}[|x\rangle\langle x| \otimes |y\rangle\langle y|\, W] \nonumber \\
&= \sum_{ijkl} W_{ijkl}\, x_i^{*} y_j^{*} x_k y_l .
\label{eq:witness_polynomial}
\end{align}
where $|x\rangle = \sum_i x_i |i\rangle$, $|y\rangle = \sum_j y_j |j\rangle$, for some bases
$\{|i\rangle\}$ and $\{|j\rangle\}$ of $\mathcal{H}_A$ and $\mathcal{H}_B$. This mapping is also an isomorphism, that we named after the authors of~\cite{Doherty2004CompleteFamily}, the DPS isomorphism. This is illustrated in Figure~\ref{fig:DPS_iso}. Links between the three isomorphisms will be presented clearly in subsection~\ref{subsubsec:decomposabilitytriple}.

\begin{figure*}
    \centering
    \includegraphics[width=1\linewidth]{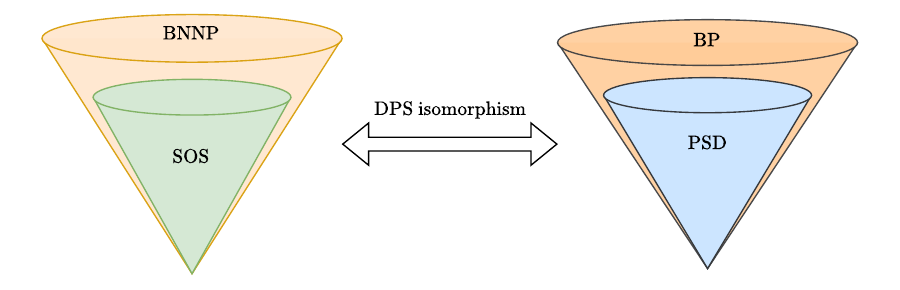}
    \caption{Schematic representation of the DPS correspondence between biquadratic forms and Hermitian operators. The cone $\text{BNNP}$ of biquadratic nonnegative polynomials is mapped to the cone $\text{BP}$ of block-positive operators, since nonnegativity is tested on product vectors. The subcone \SOS corresponds to positive semidefinite operators $\text{PSD}$. Consequently, a biquadratic form that is nonnegative but not \SOS corresponds to a block-positive but non-positive semidefinite operator, namely an entanglement witness.}
    \label{fig:DPS_iso}
\end{figure*}
\subsection{Decomposability and Hermitian polynomials}

\subsubsection{Maps decomposability}

A crucial question for our purposes, but absent from the study
of~\cite{Klep2017ManyMorePositiveMaps}, concerns the entanglement detection power of the generated maps.

Traditionally, entangled states are classified according to whether they can be detected by the partial transpose. More precisely, one distinguishes states that remain positive under the map $\id \otimes T$ (PPT states) from those that become non-positive (NPT states). It is the simplest separability criterion from a computational perspective. Even though its computational power decreases as the dimension of the system increases, it is a highly powerful criterion in low dimensions, and particularly in qubit-qubit or qubit-qutrit system where it is necessary and sufficient. To translate these properties at the level of maps, we need the following definition:

\begin{definition}[Decomposable maps]
        A positive map $\Lambda : \mathcal L(\mathcal H)\longrightarrow \mathcal L(\mathcal H)$ is said to be decomposable if it can be written as 
        \begin{equation*}
            \Lambda=CP_1 + CP_2 \circ \Gamma,
        \end{equation*}
        where $CP_1, CP_2$ are completely positive maps and $\Gamma$ is the partial transpose map. \label{def:decomposable}
\end{definition}

Decomposable maps form the class of maps whose entanglement detection capacities are limited to that of the Partial Transpose - they cannot detect PPT entangled states. This is elegantly shown in~\cite{Kye2013FacialStructures} by noticing that the cone of PPT states is the dual of the cone of decomposable maps. 


Any map that is not decomposable is said to be indecomposable. A fundamental result on positive maps was proven by \cite{Woronowicz1976PositiveMapsLowDimensional}:

\begin{theorem}[Decomposition of positive maps in low dimension]
\label{thm:decomposition_positive_maps_low_dimension}
Let us consider a linear map $\Lambda$ acting on a product of Hilbert spaces $\mathcal{H}_A \otimes \mathcal{H}_B$ and denote $d = \dim(\mathcal{H}_A \otimes \mathcal{H}_B) = \dim(\mathcal{H}_A) \times \dim(\mathcal{H}_B)$. If $d \leq 6$, then $\Lambda$ is decomposable~\cite{Stormer2013PositiveLinear}.
\end{theorem}

 This is one way to understand why the PPT criterion is sufficient and necessary to detect entanglement for qubit-qubit and qubit-qutrit systems~\cite{Horodecki1996SeparabilityMixed}. 

\subsubsection{Classification of \SOS Hermitian Polynomials}

It is useful to see how these concepts are transcribed at the level of polynomials.  For real polynomials, there is only one notion of \SOS.
 
 Quantum mechanics changes this picture. States, observables, and optimization problems in quantum information are naturally expressed in terms of polynomials in complex variables $z \in \mathbb C^n$ and their conjugates $\bar z$ \cite{Doherty2004CompleteFamily,Fang2020SumSquares}. These are Hermitian polynomials, real valued functions of both $z$ and $\bar z$. When one asks what a sum of squares certificate should mean for such objects, two genuinely different answers emerge, depending on which functions are allowed as summands: Real \SOS (RSOS) which is a square of Hermitian polynomials, \ie function of both $z$ and $\bar z$, and complex \SOS (\CSOS) which is a modulus squares of holomorphic polynomials in variable $z$ only \cite{Fang2020SumSquares}. Both reduce to the classical notion when restricted to real variables, yet they are not equivalent in the complex setting. 
 \begin{definition}[Bi-Hermitian polynomial]
     A bi-hermitian polynomial $p(z,\bar z,w,\bar w)$ is a polynomial in two complex vectors $z\in \mathbb C^n$, $w\in \mathbb C^m$ and their complex conjugates. The general form of such polynomial is 
     \begin{equation}
         p(z,\bar z,w,\bar w)=\sum_{(u,v,t,s)}p_{uv,st}z^u\bar z^v w^s\bar w^t
     \end{equation}
     where the coefficients satisfy $p_{uv,st}=\overline{p_{vu,ts}}$, the symmetry conditions that enforces real values. We say that $p$ is nonnegative if $p(z,\bar z,w,\bar w) \geq 0 $ $\forall(z,w)\in \mathbb  C^n \times \mathbb C^m$.
 \end{definition}

 The polynomial associated to an entanglement witness $W$, namely $E_W(z,w)=\bra{zw}W\ket{zw}$, is a canonical example of a bi-Hermitian polynomial: it is nonnegative for all product vectors if and only if W is block positive. Certifying the non-negativity of $E_W$ is therefore directly tied to the entanglement detection problem, and it is precisely here that the two notions of \SOS enter \cite{Doherty2004CompleteFamily,Fang2020SumSquares}.

 \begin{definition}(Real Sum of Squares)
     A bi-Hermitian polynomial $p(z,\bar z,w,\bar w)$ is \RSOS if there exist bi-Hermitian polynomials $g_1,...,g_k$ such that 
     \begin{equation}
         p(z,\bar z,w ,\bar w)=\sum_ig_i(z,\bar z,w,\bar w)^2. 
     \end{equation}
     \label{def:RSOS}
 \end{definition}

\begin{definition}(Complex Sum of Squares)
A bi-hermitian polynomial $p(z,\bar z,w,\bar w)$ is \CSOS if there exist holomorphic polynomials $q_1,...,q_k$, polynomial in $(z,w)$ alone, not in $(\bar z,w)$, such that 
\begin{equation}
    p(z,\bar z,w,\bar w)=\sum_i |q_i(z,w)|^2
\end{equation}
\label{def:CSOS}
\end{definition} 

The correspondence between positive maps and bi-Hermitian polynomials provides the natural framework in which the distinction between \CSOS, \RSOS, and non-\RSOS becomes meaningful \cite{Fang2020SumSquares}. In this setting, algebraic certificates on the polynomial side mirror exactly hierarchy on the map side. More precisely, if $\Gamma : M_n(\mathbb C) \to M_m(\mathbb C)$ is block positive and
\[
p_\Gamma(z,w)=\bra{w}\Gamma\big(\ket{z}\bra{z}\big)\ket{w}
\]
is the associated bi-Hermitian polynomial, then a \CSOS decomposition corresponds to the case $\Gamma$ completely positive. Allowing \RSOS enlarges this class precisely to decomposable map as we defined in \ref{def:decomposable}. Accordingly, failure of an \RSOS decomposition is the polynomial signature of indecomposability.

This correspondence is illustrated by three canonical examples. First, for the identity map $\id:M_n(\mathbb C)\to M_m(\mathbb C)$. Its polynomial is given by
\[
p_{\id}(z,w)=|\bra{w}z\rangle|^2
\]
where $\ket{z}=\sum_iz_i\ket{i}$ and $\ket{w}=\sum_jw_j\ket{j}$ with $z_i \in \mathbb C$, $w_j \in \mathbb C$ $\forall i \in \{1,...,n\}$, $\forall j\in \{1,...,m\}$. Upon this basis decomposition we get 
\[
p_{\id}(z,w)=\sum_i|w_iz_i|^2
\]
This is a \CSOS polynomial, since it is the modulus square of a holomorphic polynomial.

Second, consider the transposition map \(T(X)=X^T\).
\[
p_T(z,w)=\bra{w}T\big(\ket{z}\bra{z})\ket{w}
\]
And by expanding $\ket{z}$ and $\ket{w}$ in the computational basis we get 
\begin{equation*}
    p_{T}(z,\bar z,w,\bar w)=\sum_{i,j,k,l}\bar z_i\bar w_k w_lz_j\braket{i|l}\braket{k|j}
\end{equation*}
this impose $i=l$ and $k=j$, $p_T$ become 
\begin{equation*}
    p_{T}(z,\bar z,w,\bar w)=\sum_{i,j}\bar z_iw_iz_j\bar w_j
\end{equation*}
by factorizing it yields to 
\begin{equation*}
    p_{T}(z,\bar z,w,\bar w)=\bigg(\sum_i\bar z_iw_i\bigg)\bigg(\sum_jz_j\bar w_j\bigg)=\bigg|\sum_i \bar z_iw_i\bigg|^2
\end{equation*}
By defining $L(z,\bar z,w,\bar w)=\sum_i\bar z_iw_i$ and 
\begin{equation*}
    g_1=\frac{L+\bar L}{2}; \quad g_2=\frac{L-\bar L}{2i}
\end{equation*}
which are both bi-Hermitian polynomial in $(z,\bar z,w,\bar w)$ and $|L|^2=g_1^2+g_2^2$, thus $p_T$ is \RSOS, since it is expressed as the sum of a squared of Hermitian polynomials
\[
p_{T}(z,\bar z,w,\bar w)=g_1(z,\bar z, w, \bar w)^2+g_2(z,\bar z, w, \bar w)^2
\]
but it is not \CSOS, because the Transpose map is not CP. It therefore provides the basic example of an \RSOS polynomial that is not \CSOS, and corresponds exactly to a decomposable map.

Finally, for the Choi map on $M_3(\mathbb C)$,
\[
\Phi_{Ch}(X)=
\begin{pmatrix}
x_{11}+x_{33} & -x_{12} & -x_{13}\\
-x_{21} & x_{22}+x_{11} & -x_{23}\\
-x_{31} & -x_{32} & x_{33}+x_{22}
\end{pmatrix},
\]
Which is indecomposable, and the polynomial
\[
p_{\Phi_{Ch}}(z,\bar z, w,\bar w)=\bra{w}\Phi_{Ch}\big(\ket{z}\bra{z}\big)\ket{w}
\]
with $\ket{z}\bra{z}=(z_i\bar z_j)_{i,j}$ so 
\begin{equation*}
    \Phi_{\mathrm{Ch}}(\ket{z}\bra{z})=
\begin{pmatrix}
|z_1|^2+|z_3|^2 & -\bar z_1z_2 & -\bar z_1z_3\\
-\bar z_2z_1 & |z_2|^2+|z_1|^2 & -\bar z_{2}z_3\\
-\bar z_3z_1 & -\bar z_{3}z_2 & |z_3|^2+|z_2|^2
\end{pmatrix},
\end{equation*}
Then, the biquadratic form associated to the Choi map 

\begin{align}
    p_{\Phi_{Ch}}(z,\bar z,w,\bar w)&=(|z_1|^2+|z_3|^2)|w_1|^2   \label{eq:choi_polynomial}
    \\ 
    &+(|z_2|^2+|z_1|^2)|w_2|^2 \nonumber \\ 
    &+(|z_3|^2+|z_2|^2)|w_3|^2  \nonumber \\
    & -\bar z_2z_1\bar w_2w_1 -\bar z_1z_2\bar
    w_1w_2 \nonumber \\
    &-\bar z_3z_1\bar w_3 w_1-\bar z_1z_3\bar w_1 w_3 \nonumber \\ &-  \bar z_2z_3\bar w_2w_3-\bar z_3 z_2 \bar w_3 w_2 \nonumber
\end{align}

is nonnegative on product vectors but admits no \RSOS decomposition. This gives a canonical example of a non-\RSOS bi-Hermitian polynomial.

Taken together, these examples make clear that the passage from positive maps to bi-Hermitian polynomials preserves the relevant convex and algebraic structure as we can see in Figure~\ref{fig:TripleIsomorphism}. Complete positivity appears as a \CSOS certificate, decomposability as the broader \RSOS property, and indecomposability as the failure of any \RSOS representation despite global nonnegativity on product vectors. The identity map, the transpose map, and the three-dimensional Choi map therefore illustrate the three basic regimes of this correspondence.

\subsubsection{Decomposability in triple aspects}
\label{subsubsec:decomposabilitytriple}

Thus, bihomogeneous polynomials, maps and operators can be put into correspondence two by two through three different isomorphism:
\begin{itemize}
    \item The KSMZ isomorphism between polynomials and maps \cite{Klep2017ManyMorePositiveMaps}, following:
    \begin{equation}
 p_\Phi(x,y) = y^* \Phi(xx^*) y. \nonumber \label{}
\end{equation}
    \item The Choi isomorphism between maps and  operators\cite{Choi1975CompletelyPositive}, following:
    \begin{equation}
C_\Phi = \sum_{i,j=1}^n \ket{i}\bra{j} \otimes \Phi(\ket{i}\bra{j}) . \nonumber
\end{equation}
    \item The DPS isomorphism between operators and polynomials \cite{Doherty2004CompleteFamily},following:
   \begin{align}
       E_W(x,y) &= \langle xy | W | xy  \rangle
   \end{align}
A sketch that sums it up is given in Figure~\ref{fig:TripleIsomorphism}.
\end{itemize}

\begin{figure*}[t!]
    \centering
    \includegraphics[width=1\linewidth]{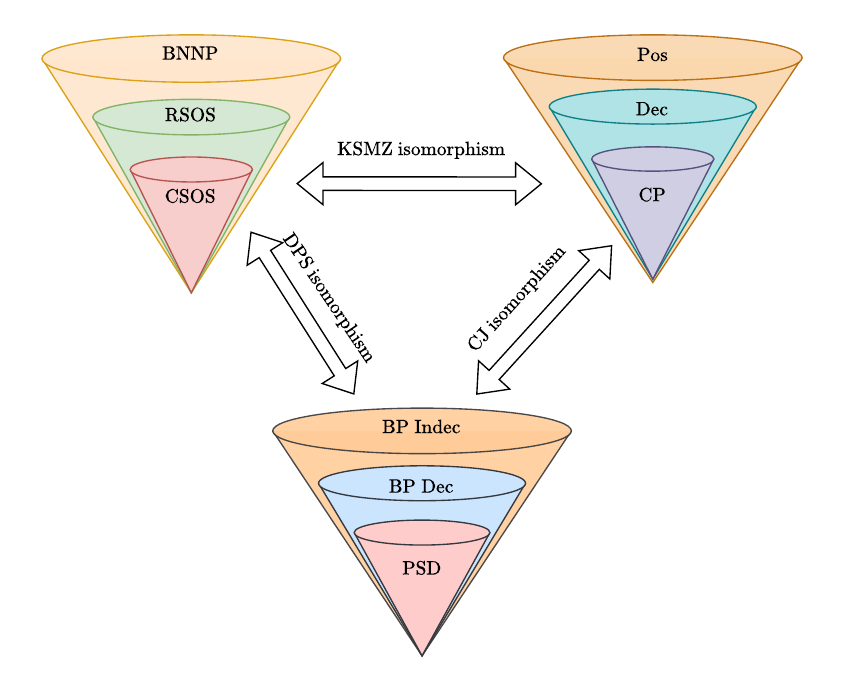}
    \caption{Triple correspondence between Hermitian biquadratic forms, positive maps, and block-positive operators. The cone $\text{BNNP}$ of biquadratic nonnegative polynomials corresponds to the cone $\text{Pos}$ of positive maps and, through the Choi isomorphism, to the cone $\text{BP}$ of block-positive operators. The subcone \CSOS corresponds to completely positive maps and positive semidefinite Choi operators, while the intermediate cone \RSOS corresponds to decomposable maps and decomposable block-positive witnesses. The \RSOS/\CSOS spearation is known on the operator polynomial side from \cite{Fang2020SumSquares}, where \CSOS corresponds to $\text{PSD}$ operator and \RSOS corresponds to decomposable block-positive operators. The central observation of this work is that this separation extends, through Choi-Jamiołkowski and KSMZ isomorphism, to a classification for complexified the positive maps generated by the KSMZ construction.}
    \label{fig:TripleIsomorphism}
\end{figure*}

We now show that the DPS algorithm is indeed equal to the composition of the KSMZ Isomorphism with the Choi-Jamiłkowski.
\begin{proposition}[Compatibility of the three correspondences] 
\label{prop:compatiblity_3iso}
Let 
\begin{equation*}
    \Phi:\mathcal L(\mathcal H_A)\longrightarrow \mathcal L(\mathcal H_B)
\end{equation*}
be a linear map, and let $C_\Phi\in\mathcal L(\mathcal H_A\otimes\mathcal H_B)$ be its Choi operator,
\begin{equation*}
    C_\Phi
    =
    \sum_{i,j}
    |i\rangle\langle j|\otimes \Phi(|i\rangle\langle j|).
\end{equation*}
Then, the Hermitian polynomial associated with $C_\Phi$ through the DPS correspondence coincides with the biquadratic form associated with $\Phi$ through the KSMZ correspondence. Namely, for all $x\in\mathcal H_A$ and $y\in\mathcal H_B$,
\begin{equation}
    E_{C_\Phi}(x,y)
    :=
    \langle x\otimes y|C_\Phi|x\otimes y\rangle
    =
    \langle y|\Phi(|x\rangle\langle x|)|y\rangle
    =
    p_\Phi(x,y).
\end{equation}
Equivalently, the DPS correspondence is the composition of the Choi-Jamiołkowski correspondence with the KSMZ correspondence.
\end{proposition}
\begin{proof}
Let us consider $|x\rangle = \sum_i x_i |i\rangle$, $|y\rangle = \sum_j y_j |j\rangle$, 
for some orthonormal bases $\{|i\rangle\}$ and $\{|j\rangle\}$ of 
$\mathcal{H}_A$ and $\mathcal{H}_B$. We apply the DPS expression to the Choi operator
$C_\phi$ of a given PnCP map $\phi$:
\[
E_{C_\phi}(x,y) = \langle xy | C_\phi | xy \rangle .
\]

Using the definition of the Choi operator,
\[
C_\phi = \sum_{i,j} |i\rangle\langle j| \otimes \phi(|i\rangle\langle j|),
\]
we obtain
\[
\langle xy | C_\phi | xy \rangle
=
\sum_{i,j}
\langle x|\langle y|
\left(|i\rangle\langle j| \otimes \phi(|i\rangle\langle j|)\right)
|x\rangle|y\rangle .
\]

Using tensor product factorization,
\[
=
\sum_{i,j}
\langle x|i\rangle
\langle j|x\rangle
\,
\langle y| \phi(|i\rangle\langle j|) |y\rangle .
\]

Since $|x\rangle = \sum_k x_k |k\rangle$, we have
\[
\langle x|i\rangle = x_i^*, 
\qquad
\langle j|x\rangle = x_j,
\]
thus
\[
E_{C_\phi}(x,y) =
\sum_{i,j}
x_i^* x_j
\,
\langle y| \phi(|i\rangle\langle j|) |y\rangle .
\]

We can rewrite the sum as
\[
E_{C_\phi}(x,y) =
\left\langle y \middle|
\phi\!\left(
\sum_{i,j} x_i^* x_j |i\rangle\langle j|
\right)
\middle| y
\right\rangle .
\]

Noting that
\[
\sum_{i,j} x_i^* x_j |i\rangle\langle j|
=
|x\rangle\langle x|
=
xx^*,
\]
we obtain
\[
\langle xy | C_\phi | xy \rangle
=
\langle y | \phi(xx^*) | y \rangle.
\]

Finally,
\[
\langle y | \phi(xx^*) | y \rangle
=
y^* \, \phi(xx^*) \, y
=
p_\phi(x,y),
\]
which is exactly the KSMZ polynomial.

Therefore, the DPS correspondence coincides with the composition of the
KSMZ isomorphism with the Choi–Jamiołkowski isomorphism.
\end{proof}


\section{Algorithm Implementation}
\label{sec:algoimpl}

The present work builds on the algorithm presented in~\cite{Klep2017ManyMorePositiveMaps} that is able to generate positive \nonSOS polynomials, before turning them into PnCP maps through the isomorphism previously described in \ref{eq:KSMZIsomorphism}. We implemented this algorithm in MATLAB and made it available at the public repository \url{https://github.com/Paulcat/Entanglement-Moments}~\cite{GithubCode}. In this part, we explain theoretically how this construction works, we present the manner with which we implemented it and finally we present an original numerical test to guarantee positivity of the maps constructed.

\subsection{Construction of positive \nonSOS polynomials useful for entanglement detection}

The construction of positive polynomials that are \SOS goes back to Hilbert, who used special configurations of zeros together with the Cayley-Bacharach theorem to prove existence; this approach was later made explicit and constructive by Reznick~\cite{Reznick2000}. A second line of work provided concrete examples (such as the Motzkin, Robinson, and Choi-Lam polynomials), where nonnegativity is established via algebraic inequalities, while the sparse monomial structure prevents any \SOS decomposition~\cite{Motzkin1967,ChoiLam1977,Robinson1973}. More recently, convex-analytic methods interpret the \SOS cone as a cone strictly contained in the Positive cone, and construct separating linear functionals showing that a given polynomial is not \SOS~\cite{Blekherman2012,Laurent2008sumssquares}. Finally, modern approaches from convex algebraic geometry exploit the structure of extreme rays and the algebraic boundary of these cones (e.g.\ via symmetroids or Gram matrix rank obstructions) to systematically produce Positive \nonSOS polynomials~\cite{BlekhermanRanestadSturmfels2013}. 

 However, not all positive non-\SOS polynomials are physically relevant for bipartite entanglement detection, as this structure imposes strict algebraic constraints. For a system of dimensions $d_A \times d_B$, we must restrict our attention to \textit{bihomogeneous} polynomials of degree $(2,2)$, which are functions of two vector-valued variables $\mathbf{a} \in \mathbb{R}^{d_A}$ and $\mathbf{b} \in \mathbb{R}^{d_B}$ representing the local degrees of freedom. Furthermore, non-homogeneous terms violate the quadratic scaling required for expectation values of linear operators, $\langle a, b | W | a, b \rangle$, rendering them impossible to map to standard quantum observables. It is necessary to distinguish between general \textit{quartic forms} (homogeneous polynomials of total degree four) and the biquadratic forms required in this case. While a general quartic polynomial $P(\mathbf{x})$ may be globally nonnegative and not \SOS, imposing a bipartite structure $P(\mathbf{a}, \mathbf{b})$—that is, interpreting it as a biquadratic form with respect to a fixed tensor decomposition—can lead to a loss of positivity when restricted to product states. In other words, although global nonnegativity is preserved, the resulting form may fail to be block positive with respect to the chosen bipartition~\cite{SkowronekZyczkowski2009}.

Nevertheless, these researches have started to be used to address the separability problem. Recently,~\cite{Ohst2025RevealingHidden} used PnCP maps created out of positive \nonSOS polynomials to witness non-classicality in the field of optics. More directly related to our work, a construction of PnCP maps~\cite{Buckley2020}, building on results regarding the number of zeros of positive semi-definite biquadratic form~\cite{quarez2015} has been adapted to detect PPT entangled states not detected either by the Choi map~\cite{buckley2022newexamplesentangled}. The KSMZ algorithm is one of those techniques. We now explain how its polynomials are created

\subsection{Theoretical description of the KSMZ algorithm}
\label{sec:descriptionKSMZ}

The algorithm  introduced in~\cite{Klep2017ManyMorePositiveMaps} presents a construction of nonnegative \nonSOS biquadratic forms $p(x,y)$ in two steps: 
\begin{enumerate}
    \item Build a quadratic form $q \in \big(\mathbb R[z]/I_{n,m}\big)_2$ that is nonnegative on $S_{n-1,m-1}(\mathbb R)$ but \nonSOS in the quotient ring.
    \item  Apply a map (the Segre pullback) to transfer these properties to a biquadratic form that defines the PnCP map. 
 \end{enumerate}

Working in the quotient ring $\mathbb R[z]/I_{n,m}$ (see Definition~\ref{def:segre-ideal}) is the key algebraic idea: $I_{n,m}$ encodes exactly the rank-one constraint $z_{ij}=x_iy_j$, so two polynomials in the same equivalence class vanish identically on every point of the Segre variety. Since biquadratic forms are completely determined by their values on rank-one matrices, no information is lost, nonnegativity of $p$ on $\mathbb R[x,y]_{2,2}$ is equivalent to nonnegativity of $q$ on $S_{n-1,m-1}(\mathbb R)$ \cite{Blekherman_2021}, and the impossibility to be \SOS in the quotient ring translates exactly into the \nonSOS property of $p$. The full mathematical details are given in Appendix~\ref{app:algo}.

\subsubsection{Construction of a nonnegative \nonSOS polynomial}
\label{subsubsec:AlgorithmKSMZ}
The key insight is that biquadratic forms $p(x,y)$ are completely determined by their values on rank-one matrices, which correspond to points on the real Segre variety. The algorithm randomly samples $N$ points $s_1,...,s_N$ on this variety and forces prescribed zeros at these points. The algorithm constructs:

\begin{enumerate}
\item A \SOS base term $H(z) = \sum_i h_i^2(z) \in \big(\mathbb R[z]/I_{n,m}\big)_2$ that vanishes exactly at the prescribed points, where the $h_i$ are linear forms.
\item A perturbation $f(z) \in \big(\mathbb R[z]/I_{n,m}\big)_2$ that vanishes to second order at each prescribed point along the variety, yet cannot be expressed as a product of the $h_i$.
\item A one-parameter family \[q_\delta(z) = \delta f + H(z)
\label{eq:one_param_family}
\] in $\big(\mathbb R[z]/I_{n,m}\big)_2$ for which nonnegativity is preserved on $S_{n-1,m-1}(\mathbb R)$ for sufficiently small $\delta > 0$.
\end{enumerate}

The role of the perturbation \(f\) is twofold. First, at each prescribed point
of the Segre variety, \(f\) vanishes together with its first derivatives;
equivalently, it vanishes to second order. Geometrically, \(f\) does not merely
pass through zero, but has tangential contact with zero at these points. This
second-order vanishing is what allows the perturbation \(H+\delta f\) to
preserve nonnegativity for sufficiently small \(\delta>0\).

The \nonSOS condition  comes from a different property:
\(f\) is chosen outside the span of the products \(h_i h_j\), where the \(h_i\)
are the linear forms vanishing at the prescribed points. Indeed, suppose that
\(q_\delta=H+\delta f\) admitted an \SOS decomposition in the quotient ring. Since
\(q_\delta\) vanishes at each prescribed point, every square appearing in the
decomposition would also have to vanish there. Hence each linear summand would
belong to the space spanned by the \(h_i\), and the whole \SOS expression would
belong to the span of the products \(h_i h_j\). Since \(H=\sum_i h_i^2\) already
belongs to this span, this would force \(f\) itself to belong to
\(\operatorname{span}\{h_i h_j\}\), contradicting the way \(f\) was chosen.

Therefore, \(q_\delta\) cannot be \SOS in
\(\mathbb R[z]/I_{n,m}\) for any \(\delta>0\).

\subsubsection{Transformation into PnCP maps }

Since $q_\delta \geq 0$ on $S_{n-1,m-1}(\mathbb R)$ for small $\delta > 0$ and is \nonSOS in the quotient ring, $\mathbb R[z]/I_{n,m}$ by construction, applying the Segre pullback map
\begin{equation}
    \tilde \sigma^{\#}_{n,m} : \big(\mathbb R[z]/I_{n,m}\big)_2 \to  \mathbb R[x,y]_{2,2}
\end{equation}
yields a biquadratic form $p(x,y) = \tilde \sigma^{\#}_{n,m}(q_\delta)$ that is both nonnegative and \nonSOS. By the isomorphism between biquadratic forms and linear maps presented in~\eqref{eq:KSMZIsomorphism}, this translates directly into a PnCP map $\Phi$.

For a complete description of the algorithm, including the role of the Segre variety's compactness in controlling the perturbation, the choice of prescribed zeros, and the technical details of the \SOS obstruction, we refer to Appendix~\ref{app:algo}.

\subsection{Implementation}
\label{sec:implementation}

We now explain how the algorithm is implemented in practice. Constructing a PnCP map reduces to finding the largest $\delta > 0$ such that the polynomial $q_\delta$ defined in \eqref{eq:one_param_family} is nonnegative. Such a condition is NP-hard to check, but can be handled via moment-\SOS relaxations \cite{Lasserre2001GlobalOptimization, Parrilo2003Semidefiniteprogramming}. The authors in~\cite{Klep2017ManyMorePositiveMaps}, building on a result by P\'{o}lya \cite{Polya1928}, propose for instance to solve for some order $\ell\in\NN_*$
\begin{equation}
    \label{eq:relax-Polya}
    \max\ \delta \quad \text{s.t.}\quad p_\delta(x,y)  Q(x,y)^\ell \;\;\text{is \SOS}
\end{equation}
where $Q(x,y) \eqdef \sum_{i,j} (x_i y_j)^2$, and $p_\delta \eqdef \tilde{\sigma}^\sharp_{n,m}(q_\delta)$ as before. The sum-of-squares constraint is known to amount to matrix inequalities constraints, and can be handled in polynomial time by interior point solvers \cite{vanderberghe96}.

The factorization appearing in \eqref{eq:relax-Polya} is in fact a particular instance of a more general result by Artin, which states that any nonnegative polynomial can be written as a ratio of sum-of-squares, yielding the following general relaxation:
\begin{equation}
    \label{eq:relax-Artin}
    \max_{\substack{\delta \in \mathbb{R}_+ \\ \sigma \in \mathbb{R}[x,y]_\ell}}\ \delta \quad \text{s.t.}\;\;
        \left\{\begin{aligned}
        &p_\delta(x,y) \sigma(x,y) \;\;\text{is \SOS}\\
        &\sigma(x,y) \quad \text{is \SOS}
    \end{aligned}\right.
\end{equation}
which is non-convex in $(\delta,\sigma)$. As described in \cite{Bhardwaj2023PracticalApproach}, \eqref{eq:relax-Artin} may be solved with the bisection method: starting with $\delta=1$ and $\ell=1$, solve the (convex) feasibility problem over $\sigma \in \RR[x,y]_\ell$; if a solution is not found, set $\delta \leftarrow \delta/2$ and reiterate until $\delta$ gets below a chosen threshold. If a solution is still not found, we increase $\ell$ and reiterate the whole process. 
In our experiments, the SDPs are solved with Mosek, interfaced with the Yalmip library \cite{Lofberg2004Yalmip}. We stop at the second order of the hierarchy, and set the threshold for $\delta$ to be $10^{-2}$. In the case no solution is found at this point, we start the procedure from the beginning to produce new SOS terms $h_i$ and perturbation $f$ (see section~\ref{subsubsec:AlgorithmKSMZ}).

\paragraph{Comparison with \cite{Bhardwaj2023PracticalApproach}.} The code is very largely based on the extensive proposition of the author of \cite{Bhardwaj2023PracticalApproach}, which can be found at \url{https://github.com/Abhishek-B/PnCP}. We however preferred to carefully re-implement each step of the procedure, rather than directly use the existing code, after observing a (possibly numerical) flaw in the results of~\cite{Bhardwaj2023PracticalApproach}: with the map $\Phi$ given in example 5.2 of \cite{Bhardwaj2023PracticalApproach}, testing against random semidefinite matrices $A$ shows that $\Phi(A)$ is not always semidefinite, as made explicit in the example below.

\begin{example}
With
\begin{equation*}
    A \eqdef \begin{bmatrix}0.9093 & 0.4767 & 0.6162 \\0.4767 & 0.8005 & 0.8288 \\0.6162 & 0.8288 & 1.1633 \end{bmatrix} \succeq 0
\end{equation*}
and the map $\Phi$ of example 5.2 in \cite{Bhardwaj2023PracticalApproach}, we obtain
\begin{equation*}
    \Phi(A) = \begin{bmatrix} 30.5807 & -19.9351 & 18.8731 \\ -19.9351 & 135.3230 & -29.6926 \\ 18.8731 & -29.6926 & 12.8955\end{bmatrix}
\end{equation*}
whose eigenvalues read $\begin{bmatrix} -0.9514 & 32.7428 & 147.0078 \end{bmatrix}^\top$, possibly contradicting the positivity of $\Phi$.
\end{example}

In order to certify the positivity of the maps outputted by our code, we added the test described in the following subsection. The code is available at \url{https://github.com/Paulcat/Entanglement-Moments}. 

\subsection{Guaranty of positivity}

If the polynomials $p_\delta$ are by construction inherently positive, their closeness to the border of the cones may result in numerical imprecisions, sometimes yielding (falsely) negative values. 
Hence, to ensure unequivocally the positivity of these polynomials, we complete the KSMZ algorithm with an additional test: we compute via a mere Lasserre's hierarchy a global lower bound on the values of the polynomials. Since the polynomials are bihomogeneous, we know any finite lower bound in practice proves that the polynomial is positive. This is due to the following proposition:

\begin{proposition}
The infimum of a non constant bihomogeneous polynomial can only be $-\infty$ or $0$.
\end{proposition}

\begin{proof}
Let $p(x,y) \in \RR[x,y]_{k1,k2}$ be a bihomogeneous polynomial in the variable $(x,y) \in \RR^n \times \RR^m$ of bidegree $(k_1,k_2)$. Then, for any $\lambda,\mu \in \RR$ the homogeneity relation imposes 
\begin{equation*}
    p(\lambda x, \mu y) = \lambda^{k_1} \mu^{k_2} p(x,y).
\end{equation*}
Since $p(x,y)$ is non constant, $k_1+k_2>0$. Therefore, if there exists $(x_0,y_0) \in \RR^n \times \RR^m$ such that $p(x_0,y_0) < 0$, then choosing $\lambda=\mu=t$ leads to $p(t x_0, t y_0) = t^{k_1 + k_2}p(x_0,y_0) $, hence $ p(t x_0, t y_0) \xrightarrow[t \to \infty]{} -\infty$, and so, $\inf_{x,y} p(x,y) = -\infty$. Conversely, with $p$ nonnegative and non constant, its minimum is $0$, since $p(0,0)=0$.

\end{proof}

This test is decisive in the sense that the certification that the polynomials we keep are indeed positive does not depend on any numerical precision. 
Hence, in our code, only polynomials that pass this additional test are retained. 


\section{Characterization of the PnCP maps generated by the KSMZ algorithm}
\label{sec:Maps}
We now  characterize the relevance of the KSMZ maps for the separability problem. In ~\cite{Klep2017ManyMorePositiveMaps} the authors left open the question of the decomposability of the maps generated by their algorithm, whereas in~\cite{Bhardwaj2023PracticalApproach} some numerical flaws invalidated the significance of their reported PPT state detection. As a consequence, we start by theoretically characterizing the capacities of these maps. First, we analytically demonstrate that the maps are indecomposable, which means that each of them is capable to detect a PPT entangled state. Then, we show that they live on the border of the positive cone, which gives an intuitive explanation on the fragility of our maps in terms of robustness. In addition, we characterize the EW associated to the KSMZ maps in terms of tightness and extremality and provide a demonstration of how one can recover the full detection power of a map by considering a whole family of witnesses rather than the unique Choi witness. Finally, we list other PnCP maps from the literature and prove for some, notably the famous Choi map, that the KSMZ maps are inequivalent to them. It means that they are capable of detecting PPT entangled states that the other map can't.

\subsection{Indecomposability of the maps}
\label{subsec:indecomposability}

In this subsection, we analytically demonstrate that the KSMZ maps are indecomposable, using an original proposition, building from~\cite{Fang2020SumSquares}. The algorithm of Section \ref{sec:algoimpl} produces a real linear map $\Phi :  S_n (\mathbb R) \to S_m(\mathbb R)$. Since $\Phi$ is not a complex linear map on $M_n(\mathbb C)$, we connect it to quantum information by passing to the complexified trivial extension introduced in \cite{Klep2017ManyMorePositiveMaps}
\begin{equation}
    \Gamma_\mathbb C=\big (\Phi \oplus 0\big)_\mathbb C : M_n(\mathbb C) \to M_m(\mathbb C) \label{eq:complexification}
\end{equation}
defined by $\Gamma_\mathbb C(A+iB)= \Phi(A)$ for $A \in  S_n(\mathbb R)$ and $B \in  K_n(\mathbb R)$ where $ K_n(\mathbb R)$ denotes the real skew-symmetric matrices. 

To $\Gamma_\mathbb C$, one associates the Hermitian biquadratic polynomial 
\begin{equation}
    p_{\Gamma_\mathbb C}(z,w)=\bra{w}\Gamma_\mathbb C(\ket{z}\bra{z})\ket{w}
\end{equation}
with $\ket{z}\in \mathbb C^n, \ket{w} \in \mathbb C^m$. Writing $\ket{z}=\ket{a}+i\ket{b}$ and $\ket{w}=\ket{c}+i\ket{d}$, $\ket{a},\ket{b} \in \mathbb R^n$ and $\ket{c},\ket{d}\in \mathbb R^m$, one has the structural identity 

\begin{align}
    p_{\Gamma_\mathbb C}(a+ib,c+id)&=p_\Phi(a,c)+p_\Phi(b,c)\notag\\
    &+p_\Phi(a,d)+p_\Phi(b,d) \label{eq:pol-dec}
\end{align}

whose derivation is given in Appendix~\ref{app:derivation_Herm_pol}. This identity is the bridge between the real algebraic obstruction obtained from the Segre based construction and the hermitian polynomial $p_{\Gamma_\mathbb C}$ : it shows that $p_{\Gamma_\mathbb C}$ is entirely controlled by the values of the real biquadratic form $p_\Phi$ on real inputs. The notions of CSOS and RSOS Hermitians polynomials introduced by~\cite{Fang2020SumSquares} and defined in Definitions~\ref{def:RSOS} and~\ref{def:CSOS}, provide the appropriate framework to state the main result of this subsection:
\begin{theorem}[Non-RSOS]
    Let $p_\Phi(x,y) \in \mathbb R[x,y]_{2,2}$ be the nonnegative non-SOS biquadratic form produced by the algorithm~\ref{sec:algoimpl} from~\cite{Klep2017ManyMorePositiveMaps}, and let $p_{\Gamma_\mathbb C}$ admit the decomposition \eqref{eq:pol-dec}. Then 
    \begin{equation*}
        p_\Phi(x,y) \; \text{is not SOS} \implies p_{\Gamma_\mathbb C} \text{is not RSOS}
    \end{equation*}
\end{theorem}
\begin{proof}
    Assume by contradiction that $p_{\Gamma_\mathbb C}$ is RSOS. By definition, there exists Hermitian polynomials $g_1,...,g_r$ in the variables $(z,\bar z,w,\bar w)$ such that 
    \begin{equation}
        p_{\Gamma_\mathbb C}(z,w)= \sum_i g_i(z,\bar z, w, \bar w)^2
    \end{equation}
    Restricting to the real slices $b=0$ and $d=0$, the identity \eqref{eq:pol-dec} becomes 
    \[
    p_{\Gamma_\mathbb C}(a+0,c+0)=p_\Phi(a,c)+p_\Phi(a,0)
   +p_\Phi(0,c)+p_\Phi(0,0).
    \]
 But then, since $p_\Phi$ is bihomogeneous of bidegree $(2,2)$, it vanishes whenever
one of its two vector variables is zero, see in Appendix \ref{app:bihomogeneous_zero_argument}.
 $p_\Phi(a,0) = 0$ ($\Phi$ is linear), and $p_\Phi(0,c) = 0$, and similarly $p_\Phi(0,0). =0$.
   Hence, $p_{\Gamma_\mathbb C}(a,c)=p_\Phi(a,c)$ and one obtains
    \begin{equation}
        p_\Phi(a,c)=\sum_i g_i(a,a,c,c)^2
    \end{equation}
    which is a SOS decomposition of $p_\Phi$. Hence, we have proven that \begin{equation*}
      p_{\Gamma_\mathbb C} \text{RSOS} \implies  p_\Phi(x,y) \; \text{SOS} 
    \end{equation*} which is the contrapositive of the theorem claimed
\end{proof}
The connection to indecomposability now follows directly from the duality relation established in \cite{Fang2020SumSquares}. A map $\Lambda$ is decomposable, in the sense of \ref{def:decomposable}, if and only if its associated Hermitian polynomial $p_\Lambda$ is RSOS. Since we have established that $p_{\Gamma_\mathbb C}$ is not RSOS, it follows immediately that $\Gamma_\mathbb C$ is not decomposable. To sum-up, since by construction the KSMZ polynomials $ p_\Phi(x,y) \in \mathbb R[x,y]_{2,2}$ are non-SOS, the polynomials  $p_{\Gamma_\mathbb C}$ are non RSOS, and the KSMZ maps $\Gamma_\mathbb C\ $ are indecomposable. 

\par Although we only focus on the KSMZ maps in this manuscript, the method we presented here can be useful to demonstrate or redemonstrate the indecomposability of other PnCP maps. It suffices to consider the Choi of the map. As an example, we provide a demonstration of the indecomposability of the Choi map in Appendix~\ref{app:ChoiIndec}. While finishing this manuscript, we became aware that a similar method was presented and used to do so, with a different language, in~\cite{Ho2026applicationschoipolynomials}.

\subsection{Location of the maps on the cone}
\label{subsec:location_of_the_maps}

Algorithm \ref{sec:algoimpl}  prescribes
zeros of the associated nonnegative biquadratic form \cite{Klep2017ManyMorePositiveMaps} :
\begin{equation}
    p_\Phi(x_0,y_0)=0,
    \qquad x_0\neq 0,\quad y_0\neq 0.
\end{equation}
Therefore $p_\Phi$ cannot be an interior point of the cone of nonnegative biquadratic forms.

Indeed, for every $\varepsilon>0$, the perturbation
\begin{equation}
    p_\varepsilon(x,y)
    =
    p_\Phi(x,y)
    -
    \frac{\varepsilon}{\|x_0\|^2\|y_0\|^2}
    \|x\|^2\|y\|^2
\end{equation}
satisfies
\begin{equation}
    p_\varepsilon(x_0,y_0)=-\varepsilon<0.
\end{equation}
Thus $p_\varepsilon$ is not nonnegative, while $p_\varepsilon\to p_\Phi$ as
$\varepsilon\to 0$.

Moreover, under the correspondence:
\begin{equation}
    \Phi \longmapsto p_\Phi(x,y)
    =
    \bra{y}\Phi(\ket{x}\bra{x})\ket{y},
\end{equation}
positivity of $\Phi$ is equivalent to nonnegativity of the biquadratic form $p_\Phi$ on all product directions:
\begin{equation}
    \Phi \in \mathcal P_{n,m}
    \qquad \Longleftrightarrow \qquad
    p_\Phi(x,y)\geq 0
    \quad \forall x\in\mathbb R^n,\; y\in\mathbb R^m.
\end{equation}
This equivalence is due to the fact that every positive semidefinite matrix is a convex combination of rank-one projectors. More directly, the form
\[
r(x,y)=\|x\|^2\|y\|^2
\]
corresponds to the positive map
\begin{equation}
    R(X)=\operatorname{Tr}(X)I_m.
\end{equation}
Therefore $p_\varepsilon$ corresponds to
\begin{equation}
    \Phi_\varepsilon
    =
    \Phi
    -
    \frac{\varepsilon}{\|x_0\|^2\|y_0\|^2}R.
\end{equation}
Since
\begin{equation}
    \bra{y_0}\Phi_\varepsilon(\ket{x_0}\bra{x_0})\ket{y_0}
    =
    -\varepsilon<0,
\end{equation}
the map $\Phi_\varepsilon$ is not positive. Hence arbitrarily small perturbations of
$\Phi$ leave the cone of positive maps. Since $\Phi$ is positive, we conclude that
\begin{equation}
    \Phi\in\partial\mathcal P_{n,m}.
\end{equation}

Finally, the same boundary property extends to the complexified map $\Gamma_{\mathbb C}$ defined in \eqref{eq:complexification}, which is the physically relevant object for entanglement detection. Indeed, viewing  $(x_0,y_0) \in \mathbb R^n \times \mathbb R^m$ as a complex vector and applying \eqref{eq:pol-dec} together with the bihomogeneity of $p_\Phi$ (see Appendix \ref{app:bihomogeneous_zero_argument}), one obtains 
\begin{equation}
     p_{\Gamma_\mathbb{C}}(x_0, y_0) 
    = p_\Phi(x_0,y_0) + p_\Phi(0,y_0) + p_\Phi(x_0,0) + p_\Phi(0,0) = 0.
\end{equation}
Hence $(x_0,y_0)$ is also a zero of $p_{\Gamma_\mathbb{C}}$, and the same perturbation argument applies verbatim with $\tilde{R}: \mapsto 
\operatorname{Tr}(X)I_m$, yielding $\Gamma_\mathbb{C} \in \partial 
\mathcal{P}_{n,m}^{\mathbb{C}}$. The real zero thus survives complexification, and the boundary location is a genuine geometric property of the physically  relevant map.

In agreement with subsections~\ref{subsec:indecomposability} and ~\ref{subsec:location_of_the_maps}, the location of the KSMZ maps is visualised in Fig~\ref{fig:ConesMapsDecomposability}.

\begin{figure}
    \centering
    \includegraphics[width=0.7\linewidth]{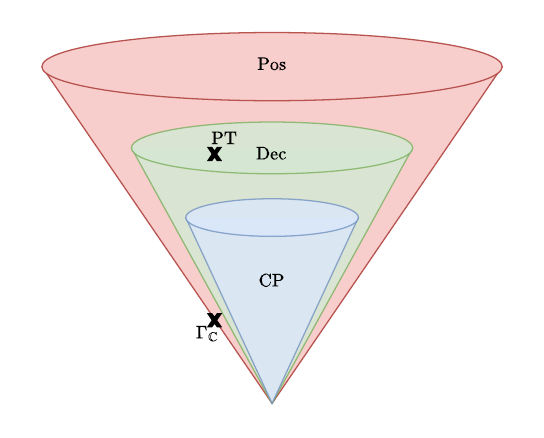}
    \caption{Cone of Positive maps (in red) that contain the cone of Decomposable map (in green), which contain itself the cone of Completely Positive map (in blue).The KSMZ maps lie in the border of Positive maps close to the Decomposable cone.}
    \label{fig:ConesMapsDecomposability}
\end{figure}

\subsection{Geometric characterization of the entanglement witness}
Having established that $W = C_{\Gamma_{\mathbb{C}}}$ is an entanglement witness (Definition \ref{def:CHOI}) lying on the boundary of the positive cone we now characterize it geometrically. The geometric properties of a witness, how and where its supporting hyperplane touches $\mathcal{SEP}$, directly inform its detection capacity. We introduce three levels of this hierarchy,\textit{tightness}, \textit{optimality}, and \textit{extremality}, and show where $W$ sits within it 

A witness $W$ is called \textit{tight} \cite{Chruscinski2014EntanglementWitnesses} if there is at least one product vectors saturating it, i.e some $\ket{x \otimes y}$ for which $\bra{x\otimes y}W\ket{x\otimes y}=0$. Geometrically the supporting hyperplane genuinely touches $\mathcal{SEP}$ in at least one point. A witness is called \textit{optimal} \cite{Lewenstein_2000,Chruscinski2014EntanglementWitnesses} if no positive semidefinite operator $P\succeq 0$ can be subtracted while preserving block positivity : the supporting hyperplane cannot be tilted any further inward toward $\mathcal{SEP}$. A sufficient condition for optimality is that the product vectors saturating $W$, i.e  those $\ket{x\otimes y}$ 
for which $\bra{x\otimes y}W\ket{x\otimes y}=0$, span the whole space 
$\mathcal{H}_A\otimes\mathcal{H}_B$~\cite{Lewenstein_2000}. Geometrically, the supported hyperplane then touches $\mathcal{SEP}$ along a face of maximal dimension leaving no margin for improvement. 

A witness is called \textit{extremal} \cite{hansen2015extremalentanglementwitnesses} if it is an extremal point of the convex cone of block positive operator cone $\text{BP}(\mathcal H_A ,\mathcal H_B)$, i.e it cannot be written as proper convex combination $W = \alpha W_1 + (1-\alpha)W_2, \; \alpha \in [0,1] $ of two distinct witnesses. Equivalently, \cite{hansen2015extremalentanglementwitnesses} show that $W$ is extremal if and only if no other witness has the same or larger set of zeros and \textit{Hessian zeros}, directions along which the biquadratic form $\bra{x \otimes y}W\ket{x \otimes y}$ vanishes to higher than second order, i.e. whose second derivative also vanishes at that point \cite{hansen2015extremalentanglementwitnesses}.

These three levels form a strict hierarchy illustrated in Figure~\ref{fig:hierarchy-ew} for which no reverse implication holds in general: an optimal witness need not be extremal~\cite{bera2023class}, and a tight witness needs not optimal \cite{Lewenstein_2000,Chruscinski2014EntanglementWitnesses}.

\begin{figure}[h!]
\centering
\begin{tikzpicture}[
    every node/.style={font=\normalsize}
]
\node (ext) at (0, 0)    {\textit{extremal}};
\node (opt) at (3, 0)    {\textit{optimal}};
\node (tgt) at (6, 0) {\textit{tight}};

\draw (ext) -- (opt)
    node[midway, above, sloped] {$\Rightarrow$};
\draw (opt) -- (ext)
    node[midway, below, sloped] {$\not\Leftarrow$};
\draw (opt) -- (tgt)
    node[midway, above, sloped] {$\Rightarrow$};
\draw (tgt) -- (opt)
    node[midway, below, sloped] {$\not\Leftarrow$};

\end{tikzpicture}
\caption{Hierarchy of geometric properties of entanglement witnesses.
Here \textit{tight} means that the witness is saturated by at least one product vector.
One has \textit{extremal} $\Rightarrow$ \textit{optimal} $\Rightarrow$ \textit{tight},
while the converses fail in general. }
\label{fig:hierarchy-ew}
\end{figure}

\subsubsection{The KSMZ maps are not extremal}
The KSMZ construction is perturbative: the parameter \(0<\delta<\delta_0\) is chosen 
precisely to preserve nonnegativity of
\[
q_\delta = H+\delta f 
\]
on the Segre variety. This strict choice also implies that the corresponding 
witness does not generate an extremal ray of the block-positive cone. Indeed, 
as long as one can choose a larger admissible parameter 
\(\delta<\delta'\), the witness \(W_\delta\) lies on a genuine line 
segment of block-positive operators.  The proof is given in 
Appendix~\ref{app:non_extremality_ksmz}.

Non-maximality of the retained $\delta$ may be checked numerically through the so-called collapse of the semidefinite hierarchy with which we compute $q_\delta$ (or more specifically its pull-back $p_\delta = \sigma^\sharp_{n,m}(q_\delta)$). The hierarchy is said to have collapsed at some order $\ell$ if the solution remains the same for every higher order ($\ell+k)$. Such a phenomenon may be detected numerically by verifying some rank condition on the dual variables associated to our optimization program, called flatness condition \cite{curto96,Laurent2008sumssquares}. As long as the hierarchy has not collapsed, we know that the corresponding $\delta$ is not maximal, since its value can be made larger simply by considering a higher order of the SOS-hierarchy. We leave this additional certification for a future version of this work; it should be noted, however, that when using the bisection method to determine $\delta$ (see section~\ref{sec:implementation}), which we did in our experiments, we can be almost certain that $\delta$ does not reach its maximum value, even if this means reducing it slightly in a final step of the procedure.


As a conclusion, the witnesses are tight but not extremal. Since by construction, the $\ker$ of the witnesses are smaller than the dimension of the Hilbert spaces in which the EWs live, it is not straightforward to show whether the witnesses are optimal or not (the spanning property~\cite{Lewenstein_2000} is not fulfilled) and is thus left for future work. See Appendix~\ref{app:optimal?} for more explanations.

\subsection{Recovering the detection power of a map with a family of witnesses}
\label{subsec:recover_detection_power_family_witnesses}

The construction above yields a single witness
$W=C_{\Gamma_{\mathbb{C}}}$. In 
practice, one may consider a \textit{family} of witnesses $\{W_t\}_{t\in T}$ parametrised by the choice of prescribed zeros or by a deformation of the underlying positive map. Each member $W_t$ defines its own saturating set $\mathcal{Z}(W_t)$, its own position in the optimal-extremal hierarchy, and its own supporting hyperplane. The union of detected states $\bigcup_t \mathcal{E}(W_t)$ can substantially enlarge the detected region compared to any single witness and recover the full detection power of the PnCP map. Understanding how these hyperplanes collectively approximate $\mathcal{SEP}$ is the subject of this subsection.

Let
\[
\Phi:\mathcal L(\mathcal H_B)\to \mathcal L(\mathcal H_C)
\]
be a positive but not completely positive map. The usual Choi construction associates with \(\Phi\) a single witness, obtained by applying \(\operatorname{Id}\otimes\Phi\) to a maximally entangled state. However, the positive-map criterion does not test only one fixed direction \cite{Horodecki1996}. It detects a state \(\rho_{AB}\) whenever the operator
\[
(\operatorname{Id}_A\otimes \Phi)(\rho_{AB})
\]
fails to be positive semidefinite.

Thus, the detection power of \(\Phi\) is naturally recovered by considering the negative directions of this operator. This viewpoint is consistent with the witnessed-entanglement approach, where
entanglement properties are extracted by optimizing over suitable classes of witness operators \cite{Lewenstein_2000,Brandao_2005}. Let \(\Phi^\dagger\) denote the Hilbert--Schmidt adjoint of \(\Phi\), defined by
\[
\operatorname{Tr}\!\left[\Phi(X)Y\right]
=
\operatorname{Tr}\!\left[X\Phi^\dagger(Y)\right].
\]
For every vector
\[
|\eta\rangle\in \mathcal H_A\otimes\mathcal H_C,
\]
we define
\[
W_{\Phi,\eta}
=
(\operatorname{Id}_A\otimes\Phi^\dagger)
\left(|\eta\rangle\langle\eta|\right).
\]
Since \(\Phi\) is positive, these operators are block-positive; hence the non-positive ones are entanglement witnesses.

\begin{proposition}
Let
\[
\Phi:\mathcal L(\mathcal H_B)\to \mathcal L(\mathcal H_C)
\]
be a positive linear map, and let $\dim\mathcal H_A=n$ and $\dim\mathcal H_C=m$
\[
r=\min\{n,m\}.
\]
Then the family
\[
\mathcal W_\Phi^{\mathrm{fsr}}
=
\left\{
(\operatorname{Id}_A\otimes\Phi^\dagger)
\left(\ket{\eta}\bra \eta\right)
:
\operatorname{Schmidt\,rank}(\eta)=r
\right\}
\]
recovers the full detection power of \(\Phi\). More precisely, for every state
\[
\rho_{AB}\in \mathcal L(\mathcal H_A\otimes\mathcal H_B),
\]
one has
\[
(\operatorname{Id}_A\otimes\Phi)(\rho_{AB})\not\succeq 0
\]
if and only if there exists a vector
\[
\ket{\eta} \in \mathcal H_A\otimes\mathcal H_C
\]
of maximal Schmidt rank such that
\[
\operatorname{Tr}\!\left[
(\operatorname{Id}_A\otimes\Phi^\dagger)
\left(\ket\eta\bra\eta\right)
\rho_{AB}
\right]<0.
\]
\end{proposition}

\begin{proof}
Let
\[
X_\rho
=
(\operatorname{Id}_A\otimes\Phi)(\rho_{AB}).
\]
Since $\rho_{AB}$ is Hermitian and $\Phi$ is positive, $X_\rho$ is Hermitian. Therefore,
\[
X_\rho\not\succeq 0
\]
if and only if there exists a vector
\[
\ket{\eta_0} \in \mathcal H_A\otimes\mathcal H_C
\]
such that
\[
\bra{\eta_0}X_\rho\ket{\eta_0}<0.
\]
Equivalently,
\[
\bra{\eta_0}
(\operatorname{Id}_A\otimes\Phi)(\rho_{AB})
\ket{\eta_0}<0.
\]

Using the definition of the adjoint, we have:
\[
\begin{aligned}
\bra{\eta_0}
(\operatorname{Id}_A\otimes\Phi)(\rho_{AB})
\ket{\eta_0}
&=
\operatorname{Tr}\!\left[
\ket{\eta_0}\bra{\eta_0}
(\operatorname{Id}_A\otimes\Phi)(\rho_{AB})
\right] \\
&=
\operatorname{Tr}\!\left[
(\operatorname{Id}_A\otimes\Phi^\dagger)
\left(\ket{\eta_0}\bra{\eta_0}\right)
\rho_{AB}
\right].
\end{aligned}
\]
Thus, if $\rho_{AB}$ is detected by the map, there exists a vector
$\ket{\eta_0}$ such that
\[
W_{\Phi,\eta_0}
=
(\operatorname{Id}_A\otimes\Phi^\dagger)
\left(\ket{\eta_0}\bra{\eta_0}\right)
\]
satisfies
\[
\operatorname{Tr}(W_{\Phi,\eta_0}\rho_{AB})<0.
\]
This operator is block-positive: for every product vector
$\ket{x}\otimes \ket{y}$, the adjoint relation gives
\[
\langle x\otimes y|W_{\Phi,\eta_0}|x\otimes y\rangle
=
\operatorname{Tr}\!\left[
|\eta_0\rangle\langle\eta_0|
\bigl(|x\rangle\langle x|\otimes\Phi(|y\rangle\langle y|)\bigr)
\right]\geq 0,
\]
since $\Phi$ acts only the second subspace, and using the cyclicity of the trace. The conclusion comes from the fact that $\Phi(|y\rangle\langle y|)\succeq 0$ since $\Phi$ is positive, and that $|x\rangle\langle x|$ and $|\eta_0\rangle\langle\eta_0|$ since they are density matrices . Finally, since this operator has a negative expectation value on the state $\rho_{AB}$, it is not positive semidefinite.
Therefore $W_{\Phi,\eta_0}$ is an entanglement witness.

It remains to show that one may choose the detecting vector with maximal Schmidt rank The use of Schmidt-rank restrictions is standard in the theory of Schmidt-number witnesses and $k$-positive maps \cite{Terhal_2000}. The set of maximal Schmidt-rank vectors is dense in
\[
\mathcal H_A\otimes\mathcal H_C.
\]

Indeed, after choosing orthonormal bases $\{|i\rangle\}_{i \in \mathbb N}$ of $\mathcal{H}_A$ and $\{|j\rangle\}_{j \in \mathbb N}$ of $\mathcal{H}_C$, a vector $|\eta\rangle = \sum_{i,j} \eta_{ij} |i\rangle \otimes |j\rangle$ can be identified with its coefficient matrix $M_\eta = (\eta_{ij})_{i,j \in \mathbb N} \in M_{n}(\mathbb C)$, and its Schmidt rank is exactly the rank of $M_\eta$ via the singular value decomposition. Since matrices of maximal rank form a dense open subset of $M_{n}(\mathbb C)$, maximal Schmidt-rank vectors are dense in $\mathcal{H}_A \otimes \mathcal{H}_C$.

Therefore, there exists a sequence of maximal Schmidt-rank vectors
\[
|\eta_n\rangle\to|\eta_0\rangle.
\]
Now consider the function
\[
f(\eta)
=
\langle\eta|X_\rho|\eta\rangle.
\]
This function is continuous because it is polynomial. Since
\[
f(\eta_0)<0,
\]
we have, for \(n\) large enough,
\[
f(\eta_n)<0.
\]
Hence, for some maximal Schmidt-rank vector \(|\eta_n\rangle\),
\[
\langle\eta_n|
(\operatorname{Id}_A\otimes\Phi)(\rho_{AB})
|\eta_n\rangle<0.
\]

By the adjoint identity, this is equivalent to
\[
\operatorname{Tr}\!\left[
(\operatorname{Id}_A\otimes\Phi^\dagger)
\left(|\eta_n\rangle\langle\eta_n|\right)
\rho_{AB}
\right]<0.
\]
Thus, whenever \(\rho_{AB}\) is detected by \(\Phi\), it is detected by a member of the family \(\mathcal W_\Phi^{\mathrm{fsr}}\).

Conversely, if there exists a maximal Schmidt-rank vector \(|\eta\rangle\) such that
\[
\operatorname{Tr}\!\left[
(\operatorname{Id}_A\otimes\Phi^\dagger)
\left(|\eta\rangle\langle\eta|\right)
\rho_{AB}
\right]<0,
\]
then the adjoint identity gives
\[
\langle\eta|
(\operatorname{Id}_A\otimes\Phi)(\rho_{AB})
|\eta\rangle<0.
\]
Therefore,
\[
(\operatorname{Id}_A\otimes\Phi)(\rho_{AB})\not\succeq 0.
\]
This proves the equivalence.
\end{proof}

Geometrically, the map \(\Phi\) detects \(\rho_{AB}\) when the operator
\begin{equation}
X_\rho=(\operatorname{Id}_A\otimes\Phi)(\rho_{AB})
\end{equation}
fails to be positive semi definite. This means that $X_\rho$ has a negative direction $\ket{\eta}$. Each such direction defines a separating hyperplane in the original state space through the operator
\begin{equation}
W_{\Phi,\eta}
=
(\operatorname{Id}_A\otimes\Phi^\dagger)
\left(|\eta\rangle\langle\eta|\right).
\end{equation}
The key point is that maximal Schmidt-rank directions form a dense open set among pure directions. Since strict negativity is stable under small perturbations, any negative direction can be replaced by a nearby maximal Schmidt-rank direction. Hence the maximal Schmidt-rank family of witnesses recovers the full detection power of the map.

\begin{remark}
This result is already known by some members of the community but we chose to give a clear reference proof. 
\end{remark}



\begin{remark}
    To construct the family of witnesses associated with a given map, one thus needs to generate all states with maximal Schmidt rank. One only needs to multiply the maximally entangled state $\ket{\Omega}$ by $\id \otimes A$, where $A$ is an invertible matrix. Indeed, considering the action of $A$ in some canonical basis,
\[
A\ket{k} = \sum_{j}^{N} A_{jk}\ket{j},
\]
one has
\[
(\id \otimes A)\ket{\Omega}
= \sum_{i,j=1}^{N} \frac{A_{ji}}{\sqrt{N}} \ket{i}\ket{j}
\]
where $N$ is the dimension of the square matrix $A$.

This way, we can span the entire Hilbert space. With $M := \frac{A^T}{\sqrt{3}}$, the condition that the Schmidt rank is maximal is that the matrix of coefficients $M$ is invertible, which means that $A$ needs to be invertible.
\label{remark:fsr}
\end{remark}

\subsection{Translation at the polynomial level}
\label{subsec:familiespolynomial}

The family $\mathcal W_\Phi^{\mathrm{fsr}}$ also admits a natural interpretation at the polynomial level. Let
\begin{equation}
    \ket{\eta} \in \mathcal H_A\otimes\mathcal H_C .
\end{equation}
For every vector $\ket{x} \in\mathcal H_A$, define the contraction
\begin{equation}
    \ket{\eta_x}
    :=
    (\bra{x}\otimes I_C)\ket{\eta}
    \in \mathcal H_C .
\end{equation}
Equivalently, if
\begin{equation}
    \ket{\eta}
    =
    \sum_{i,\alpha}\eta_{i\alpha}\ket{i}_A\otimes \ket\alpha_C
\end{equation}
then
\begin{equation}
    \ket{\eta_x}
    =
    \sum_{\alpha}
    \left(
        \sum_i \overline{x_i}\eta_{i\alpha}
    \right)
    \ket\alpha_C .
\end{equation}
Thus $\ket\eta$ defines a contraction map
\begin{equation}
    |x\rangle\longmapsto |\eta_x\rangle,
\end{equation}
whose rank is exactly the Schmidt rank of $\ket\eta$.

Now consider
\begin{equation}
    W_{\Phi,\eta}
    =
    (\operatorname{Id}_A\otimes\Phi^\dagger)
    \left(\ket\eta\bra\eta\right).
    \label{eq:generalized_choi}
\end{equation}
The DPS polynomial associated with $W_{\Phi,\eta}$ is
\begin{equation}
    E_{W_{\Phi,\eta}}(x,y)
    =
    \bra{ x\otimes y}|W_{\Phi,\eta}\ket{x\otimes y} .
\end{equation}
Using the adjoint relation, we obtain
\begin{align}
E_{W_{\Phi,\eta}}(x,y)
&=
\operatorname{Tr}\!\left[
\ket{\eta}\bra\eta
\left(
\ket x\bra x\otimes \Phi(\ket{y}\bra{y})
\right)
\right] \nonumber \\
&=
\bra{\eta_x}
\Phi(\ket{y}\bra{y})
\ket{\eta_x} 
\end{align}
Therefore,
\begin{equation}
    E_{W_{\Phi,\eta}}(x,y)
    =
    p_\Phi(y,\eta_x),
\end{equation}
where
\begin{equation}
    p_\Phi(u,v)
    =
    \bra v|\Phi(\ket u\bra u)\ket v\
\end{equation}
is the KSMZ polynomial associated with $\Phi$.

Hence, the polynomial associated with $W_{\Phi,\eta}$ is not independent of $p_\Phi$: it is obtained from $p_\Phi$ by replacing its second variable by the contracted vector $\eta_x$. In other words, each witness in the family
$\mathcal W_\Phi^{\mathrm{fsr}}$ corresponds to a pullback of the original KSMZ polynomial along the contraction map
\begin{equation}
    \ket{x} \mapsto \ket{\eta_x}
\end{equation}

The maximal Schmidt-rank condition now has a direct polynomial meaning. It requires the contraction map $x\mapsto\eta_x$ to have maximal rank, so that the resulting polynomial explores as many directions of $p_\Phi$ as allowed
by the dimensions. In the square case $\dim\mathcal H_A=\dim\mathcal H_C$, maximal Schmidt rank means that this contraction map is invertible. Thus, the usual Choi witness corresponds to the canonical maximally entangled choice,
whereas the full-Schmidt-rank family corresponds to all full-rank linear reparametrizations of the same underlying polynomial.

\begin{figure}
    \centering
    \begin{tikzpicture}[>=stealth, node distance=3.5cm and 3.5cm]
    
    \node[draw, rectangle, minimum width=3cm, minimum height=1cm] (gamma) { $p_{\Gamma_{\mathbb C}}$ };
    \node[draw, rectangle, minimum width=3cm, minimum height=1cm, right=of gamma] (pgamma) {$\Gamma_{\mathbb C}$};
    \node[draw, rectangle, minimum width=3cm, minimum height=1cm, below=of pgamma] (weta) {$W_{\Phi,\eta}$};
    \node[draw, rectangle, minimum width=3cm, minimum height=1cm, left=of weta] (peta) {$E_{W_{\Phi,\eta}}$};
    
    \draw[<->] (pgamma) -- node[above] {KSMZ isomorphism}(gamma);
    \draw[<->] (pgamma) -- node[sloped,above, align=center] {Generalized \\ Choï-Jamio$\l$kowski \\ Isomorphism}(weta);
    \draw[<->] (weta) -- node[below] {DPS}(peta);
    \draw[<->] (peta) -- node[sloped, above] {Contraction}(gamma);
    \draw[<->, dotted] (gamma) -- node[sloped, above]{$\text{DPS} \circ \text{Contraction}$} (weta);

    \end{tikzpicture}
    
    \caption{Compatibility diagram for the polynomial interpretation of the witness family $\mathcal W_{\Phi,\eta}$. The top horizontal arrow represents the KSMZ correspondence between the complexified map $\Gamma_{\mathbb C}$ and its associated biquadratic polynomial $p_{\Gamma_{\mathbb C}}$. The right vertical arrow is the Choi--Jamio{\l}kowski correspondence, which maps $\Gamma_{\mathbb C}$ to the witness $W_{\Phi,\eta}$ after applying $\ket{\eta}$ \eqref{eq:generalized_choi}. The bottom horizontal arrow is the DPS correspondence, associating to $W_{\Phi,\eta}$ the product-state polynomial $E_{W_{\Phi,\eta}}$. The left vertical arrow represents the contraction induced by $\ket{\eta}$, through which $E_{W_{\Phi,\eta}}$ is obtained from the original KSMZ polynomial. The dotted diagonal arrow summarizes the resulting compatibility between these correspondences.}
    \label{fig:placeholder}
\end{figure}
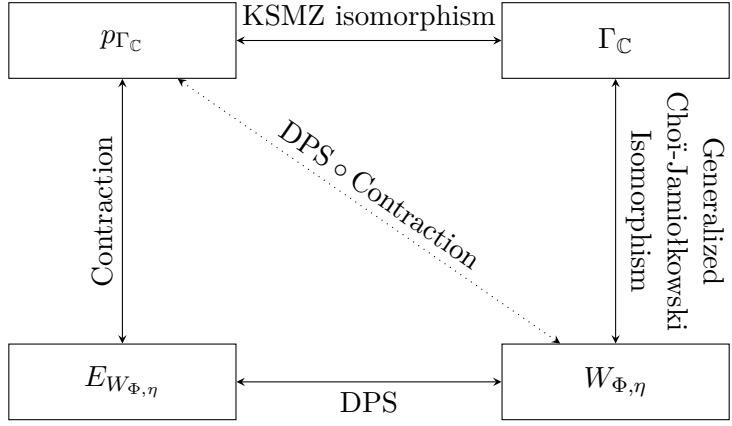

\subsection{The symmetric aspect of the KSMZ maps makes them inequivalent to the Choi map}
\label{subsec:ImpossibilityChoiBreuer}

The maps produced by the KSMZ algorithm are inequivalent to the Choi map, mostly due to the fact that they possess no skew-symmetric part. We now prove this fact. We also show that the KSMZ maps yield Partial Transpose invariant witnesses, which is the only indecomposable example we know. 
 
\subsubsection{Inequivalence with the generalized Choi maps}
\label{sec:choi-inequivalence}

The KSMZ maps considered here are obtained by extending a
real-linear map
\[
\Psi:S_3(\mathbb R)\longrightarrow S_3(\mathbb R)
\]
by zero on the skew-symmetric subspace, and then complexifying.
Writing $\Psi_{\mathbb C}$ for the complex-linear extension
of $\Psi$ to complex symmetric matrices, the resulting map is
\begin{equation}
\Gamma_0(X)
=
\Psi_{\mathbb C}\left(\frac{X+X^T}{2}\right),
\qquad X\in M_3(\mathbb C).
\label{eq:ksmz-zero-extension}
\end{equation}
Consequently, every complex skew-symmetric matrix belongs
to $\ker\Gamma_0$. Since this subspace has complex dimension
three,
\begin{equation}
\operatorname{rank}_{\mathbb C}(\Gamma_0)\leq 6.
\label{eq:ksmz-rank-bound}
\end{equation}
Here, rank refers to the map as a linear operator on $M_3(\mathbb C)$, not to its Choi matrix.

Consider the generalized Choi maps\cite{ha2011one}, in the convention
\begin{equation}
\Phi_{a,b,c}(X)=
\begin{pmatrix}
ax_{11}+bx_{22}+cx_{33} & -x_{12} & -x_{13}\\
-x_{21} & cx_{11}+ax_{22}+bx_{33} & -x_{23}\\
-x_{31} & -x_{32} & bx_{11}+cx_{22}+ax_{33}
\end{pmatrix}.
\label{eq:generalized-choi}
\end{equation}
For $a,b,c\geq0$, this map is positive if and only if
\begin{equation}
a+b+c\geq2,
\qquad
a\leq1\ \Longrightarrow\ bc\geq(1-a)^2.
\end{equation}
Within this positive family, it is not completely positive
precisely when $a<2$.

\paragraph{Proposition.}
Let $\Gamma_0$ be defined by
\eqref{eq:ksmz-zero-extension}.
There are no $U,V\in\operatorname{GL}_3(\mathbb C)$,
$\lambda>0$, and positive generalized Choi map
$\Phi_{a,b,c}$ such that
\begin{equation}
\Gamma_0
=
\lambda\,\operatorname{Ad}_U
\circ\Phi_{a,b,c}
\circ\operatorname{Ad}_V,
\qquad
\operatorname{Ad}_U(X)=UXU^\dagger.
\label{eq:choi-equivalence}
\end{equation}
In particular, $\Gamma_0$ is not equivalent to a positive
generalized Choi map under unitary input/output conjugations
and positive rescaling.

\begin{proof}
The map $\Phi_{a,b,c}$ acts as minus the identity on the
six-dimensional subspace of matrices with zero diagonal.
Its image on the diagonal subspace is nonzero, since
\[
\Phi_{a,b,c}(I_3)=(a+b+c)I_3\neq0.
\]
These two image subspaces have zero intersection. Hence
\begin{equation}
\operatorname{rank}_{\mathbb C}(\Phi_{a,b,c})\geq7.
\end{equation}
For invertible $U,V$, the maps $\operatorname{Ad}_U$
and $\operatorname{Ad}_V$ are invertible linear operators.
Their composition with $\Phi_{a,b,c}$, as well as nonzero
rescaling, preserves rank. Equation
\eqref{eq:choi-equivalence} therefore contradicts
\eqref{eq:ksmz-rank-bound}.
\end{proof}

This argument depends on the zero action on the
skew-symmetric subspace. It does not automatically extend
to other choices of extension.
Appendix~\ref{app:choi-inequivalence} gives a separate
obstruction for arbitrary extensions: in the strict KSMZ
perturbative regime, the real biquadratic form is
non-extremal, which excludes the standard Choi map up to
real invertible input/output transformations.

\subsection{Inequivalence with the Breuer-Hall maps}
\label{subsec:inequivalence_breuer_hall}

We now compare the KSMZ construction with the Breuer--Hall maps
\cite{Breuer2006,Hall2006}. Let $d=2r\geq 4$, and let
$U\in M_d(\mathbb C)$ be an antisymmetric unitary, $U^T=-U$. The
Breuer-Hall map is, up to an irrelevant positive normalization factor,
\begin{equation}
    \Phi_U(X)
    =
    \operatorname{Tr}(X)I_d
    -
    X
    -
    U X^T U^\dagger .
    \label{eq:breuer_hall_map}
\end{equation}
The normalization factor plays no role for the present discussion, since we
only compare the algebraic structure of the corresponding cones.

Since the choice of antisymmetric unitary is unique up to a unitary change of basis, it is enough to prove the statement for the standard representative
\begin{equation}
    J
    =
    \bigoplus_{\alpha=1}^{r}
    \begin{pmatrix}
        0 & 1 \\
        -1 & 0
    \end{pmatrix},
    \qquad
    J^T=-J,
    \qquad
    J^T J = I_d .
\end{equation}
We denote the corresponding Breuer-Hall map by $\Phi_J$.

\begin{proposition}
\label{prop:ksmz_not_breuer_hall}
The Breuer-Hall maps are not equivalent with the KSMZ construction, even after allowing an arbitrary skew-symmetric extension.
\end{proposition}

\begin{proof}
The KSMZ construction starts from a real map
\[
    \Phi:S_d(\mathbb R)\longrightarrow S_d(\mathbb R)
\]
whose associated biquadratic form
\[
    p_\Phi(x,y)=y^T\Phi(xx^T)y
\]
is nonnegative but not a sum of squares. Any subsequent extension to
$M_d(\mathbb R)$ is obtained by choosing an arbitrary action on
$K_d(\mathbb R)$, but this choice cannot modify the restriction of the map
to $S_d(\mathbb R)$.

We therefore only have to inspect the symmetric real restriction of the
Breuer-Hall map. Let
\[
    \Psi_J
    :=
    \Phi_J|_{S_d(\mathbb R)} .
\]
Since $S^T=S$ for $S\in S_d(\mathbb R)$, one has
\begin{equation}
    \Psi_J(S)
    =
    \operatorname{Tr}(S)I_d
    -
    S
    -
    J S J^T .
    \label{eq:BH_symmetric_restriction}
\end{equation}
The real biquadratic form associated with this restriction is
\begin{equation}
    p_{\Psi_J}(x,y)
    =
    y^T\Psi_J(xx^T)y .
\end{equation}
By Proposition~\ref{prop:BH_symmetric_core_SOS}, proved in
Appendix~\ref{app:breuer_hall_sos_core}, this polynomial is a real SOS
biquadratic form.

Hence the symmetric real core of the Breuer-Hall map is SOS, whereas a genuine KSMZ core is non-SOS by construction. Since the skew-symmetric extension cannot alter the symmetric restriction, no KSMZ map can reproduce $\Phi_J$. Therefore the Breuer-Hall maps are inequivalent to the KSMZ maps in the sense of the real symmetric-core obstruction.
\end{proof}

The conclusion is conceptually useful: the indecomposability of the Breuer-Hall maps does not come from a positive non-SOS real symmetric core,
as in the KSMZ construction. It comes from the antiunitary structure
$X\mapsto U X^T U^\dagger$, which is invisible if one only keeps the real
symmetric core.

\subsection{Existence of indecomposable EWs that are Partial Transpose invariant}

We can decompose $\ket{i} \bra{j}$ in its Hermitian and anti-Hermitian part:
\begin{equation}
    \ket{i} \bra{j} = \frac{\ket{i} \bra{j} + \ket{j} \bra{i} }{2} + \frac{\ket{i} \bra{j} - \ket{j} \bra{i} }{2} \nonumber
\end{equation}

We consider the Choi operator of a KSMZ map $C_{\phi}$. Since its application on the anti-Hermitian part is trivial, we have:
\begin{align}
    C_{\phi} &= \sum \ket{i} \bra{j} \otimes \phi ( \ket{i} \bra{j} ) \nonumber \\
             &= \sum \frac{1}{2} \ket{i} \bra{j} \otimes \phi ( \ket{i} \bra{j} + \ket{j} \bra{i})
\end{align}

\begin{align}
    C_{\phi}^{\Gamma} &= \sum \frac{1}{2} \ket{i} \bra{j} \otimes \phi ( \ket{i} \bra{j} + \ket{j} \bra{i})^{T} \nonumber \\
    &= \sum \frac{1}{2} \ket{i} \bra{j} \otimes \phi ( \ket{i} \bra{j} + \ket{j} \bra{i})
\end{align}
Since $\phi ( \ket{i} \bra{j} + \ket{j} \bra{i}) \in S_n(\mathbb{R})$.

So the KSMZ construction gives indecomposable entanglement witnesses that are partial-transpose invariant. This property is not specific to the Choi witness. Indeed, the whole family of witnesses introduced Section~\ref{subsec:recover_detection_power_family_witnesses} satisfies the same invariance. For any
\begin{equation}
W_{\phi,\eta}
=
(\operatorname{Id}\otimes\phi^\dagger)(\ket{\eta}\bra{\eta}),
\end{equation}
the trivial skew-symmetric extension implies that the range of $\phi^\dagger$ is contained in $S_n(\mathbb R)$. Hence
\begin{equation}
T\circ\phi^\dagger=\phi^\dagger,
\end{equation}
and therefore
\begin{equation}
W_{\phi,\eta}^{T_B}
=
(\operatorname{Id}\otimes T\circ\phi^\dagger)(\ket{\eta}\bra{\eta})
=
W_{\phi,\eta}.
\end{equation}


\section{Entanglement Detection}
\label{sec:EntanglementDetection}

We now present the detection of entangled states thanks to the refined algorithm.

To produce one map that acts on a $3 \times 3$ subsystem, it takes on average 2 minutes with a 12th Gen Intel(R) Core(TM) i5-1235U with 10 cores, 12 threads, and 16 Gigabytes of RAM. Note that our algorithm rejects polynomials that do not pass the decisive test. 
In order to boost our performance, we used the LIP6 Cluster HPC, with 2 x Intel Xeon E5-2690v3 	24 cores / 48 threads and a RAM of 192 GB.  We created a library of 20,000 maps in $3 \times 3$, which took several days.

\subsection{Detection of entangled states through SDPs}

We first propose an SDP that detects the entanglement of a non-specific entangled state. It is not possible to directly use the PnCP maps inside an optimization program, since the detection region, 
\[
\mathds{1} \otimes \Lambda \ngtr 0
\]
is not convex, which renders the algorithm inefficient. Trying different approaches~\footnote{for instance, \begin{equation}
\begin{aligned}
\min_{\rho, t }\quad & t\\
\text{subject to}\quad & \rho \succeq 0,\\
& \rho = \rho^{\dagger} \\
& \text{Tr}(\rho)=1. \label{eq:SDP_NPT} \\
& \mathds{1} \otimes \Lambda \pm t \mathds{1} \geq 0
\end{aligned}
\end{equation}} to incorporate this constraint leads to unbounded problems or trivial solutions. Consequently, all the SDPs we implement use EWs associated to one of the PnCP map. To start with, we use only the usual Choi Isomorphism and define 
\[
W_{\text{KSMZ}} = d* \mathds{1} \otimes \Lambda \left( \ket{\phi^+} \bra{\phi^+} \right)
\]
where $\ket{\phi^+}$ is the maximally entangled state in dimension $d$. The maps and the density matrices we use are given, alongside with numerical precision issues and solver-dependent certification details, in ~\ref{app:numerics}.

Next, we check that our entanglement witness $W_{\text{KSMZ}}$ can detect entanglement. To do so, we diagonalise it and exhibit a negative eigenvalue. The eigenstate associated to it could also be certified to be entangled via a simple Partial Transposition. The eigenvalue is of the orders of magnitude of $10^{-1}$, indicating that the EW we produce are noise robust with regard to NPT states. 

\par However we want to check that our maps are capable of detecting PPT entangled states. To do so, we run the following SDP: 
\begin{equation}
\begin{aligned}
\min_{\rho}\quad & \text{Tr}(W_{\text{KSMZ}}\rho)\\
\text{subject to}\quad & \rho \succeq 0,\\
& \rho = \rho^{\dagger} \\
& \text{Tr}(\rho)=1.  \\
& \rho^{T_B}\succeq 0, \label{eq:SDP_PPT}
\end{aligned}
\end{equation}

Because the feasible set is exactly the set of bipartite density operators with positive partial transpose, any strictly negative optimal value certifies that $W$ detects a PPT entangled state. This SDP is a key ingredient of our numerical analysis: it produces an explicit optimizer $\rho_\star$ satisfying \eqref{eq:SDP_PPT} and, consequently, numerical evidence for the indecomposability of $W$.

It is also interesting to demonstrate that the maps we produced are \textit{non-equivalent} with the Partial Transposition. So far we proved that they could detect states that the PPT criterion cannot, but the converse is actually also true. For a given KSMZ PnCP map and its associated EW $W_{\text{KSMZ}}$, we generated random NPT entangled states, and verified that the witness could not detect them.

\subsection{Beyond Choi Isomorphism} However, it is possible to enhance our criteria further. Indeed, the Choi matrix of a given map is not the only witness one can construct from it, as we demonstrated in subsection~\ref{subsec:recover_detection_power_family_witnesses}. By replacing the maximally entangled state $\ket{\phi^+} \bra{\phi^+}$ inside

\[
W_{\text{KSMZ}} = d \  \mathds{1} \otimes \Lambda \left( \ket{\phi^+} \bra{\phi^+} \right)
\]
by another entangled state, which is such that the operator created has at least one negative eigenvalue, we have defined another EW stemming from the same PnCP map. This new EW corresponds to another hyperplane in the plane of quantum states and can potentially lead to a better violation than the Choi witness. According to our result, the necessary and sufficient set of states that needs to be considered is the set of full Schmidt rank pure states. 
As explained in Remark~\ref{remark:fsr}, we can consider families of witness of the form
\[
W_{\text{KSMZ}}(A) = \left( \mathds{1} \otimes A \right) \left[ \mathds{1} \otimes \Lambda \left( \ket{\phi^+} \bra{\phi^+} \right) \right] \left( \mathds{1} \otimes A \right)^{\dagger}
\]
and sample on the set 
\[
\{ A \mid A \in \text{GL}_n(\mathbb{C}) \}
\]
In practice, we tried different optimization routines, sometimes restricting to the set of diagonal invertible matrices to lower the computational complexity. 
For instance, the minimum value over the set of PPT states for the quantity
\[
\Tr( W_{\text{KSMZ}} \rho),
\]
evaluated on the first $20$ maps, evolved from $-3.18 \times 10^{-8}$ with the Choi Witness to $-7.20 \times 10^{-8}$ when we considered a family of witness with diagonal invertible matrices, but only to $-4.68 \times 10^{-8}$ when we performed an optimization over the set of all invertible functions in the same comparable amount of time (few minutes). 

However, the optimization over the set of invertible matrices led to a more notable improvement for the best violation among the first $200$ maps. It enhanced the violation from $-5.02 \times 10^{-8}$ to $-1.37 \times 10^{-7}$ for the $62^{\text{nd}}$ map. 

\subsection{Average numerical indecomposability}
In addition, we want to numerically characterize the indecomposability of the KSMZ maps. Recall that a decomposable witness can be written 
\begin{equation}
W_D = P + Q^{T_B}, \qquad P,Q \succeq 0,
\end{equation}
so that the decomposable cone is
\begin{equation}
Dec=\{\,P+Q^{T_B}: P,Q\succeq 0\,\}.
\end{equation}
We then consider the operator-norm distance
\begin{equation}
d_\infty(W,Dec)
=
\inf_{W_D\in Dec}\|W-W_D\|_\infty. \label{eq:distance_dec}
\end{equation}
Now, we evaluate the following SDP:
\begin{equation} 
    \begin{aligned}
    \text{minimize}\quad & t\\
    \text{subject to}\quad & P \succeq 0,\quad Q \succeq 0,\\
    & -tI \preceq W-(P+Q^{T_B}) \preceq tI. \label{SDP:distance_dec}
    \end{aligned}
\end{equation}
The optimal value \(t^\star\) is exactly the distance from \(W\) to the decomposable cone in operator norm. In particular, a small value of \(t^\star\) means that the witness lies close to the decomposable set, whereas a larger value indicates a more robust form of indecomposability.

In parallel, one may also evaluate the PPT margin
\begin{equation}
\mu(W):=
-\min_{\rho\in\mathrm{PPT}} \operatorname{Tr}(W\rho), \label{eq:PPT_margin}
\end{equation}
which is computed through the SDP
\begin{equation}
    \begin{aligned}
    \text{minimize}\quad & \operatorname{Tr}(W\rho)\\
    \text{subject to}\quad & \rho \succeq 0 \\
    & \rho = \rho^{\dagger} \\
     & \rho^{T_B}\succeq 0\\
     &\operatorname{Tr}(\rho)=1. \label{SDP:PPT_margin}
    \end{aligned}
\end{equation}
By duality between the PPT cone and the decomposable cone \cite{Fang2020SumSquares}, 
$\mu(W)$ provides a quantitative measure of indecomposability and, in the regime relevant for our numerics, is equal to the distance $d_\infty(W,\mathrm{Dec})$, as illustrated in Fig~\ref{fig:dual_dec_ppt}. 

\begin{figure*}[t!]
    \centering
    \includegraphics[width=1\textwidth]{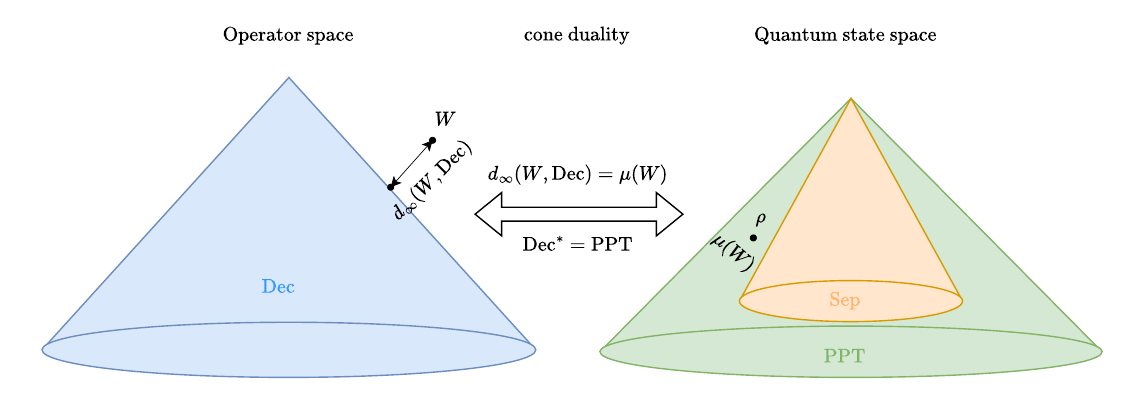}
    \caption{Geometric interpretation of decomposability and PPT detection. 
Left: in operator space, $d_\infty(W,\mathrm{Dec})$ denotes the operator-norm distance from the witness $W$ to the decomposable cone. 
Right: in state space, the PPT margin $\mu(W)=-\min_{\rho\in \mathrm{PPT}_1}\Tr(W\rho)$ measures the maximal violation of \(W\) on normalized PPT states. An optimal state $\rho$ defines the supporting functional associated with this violation. The relation between the two quantities is governed by the cone duality $\mathrm{Dec}^*=\mathrm{PPT}$.}
    \label{fig:dual_dec_ppt}
\end{figure*}

We use these SDPs map by map, and then average the resulting values over the generated family in order to obtain a numerical characterization of how close the KSMZ witnesses are to the decomposable boundary. 

In our dataset, the maximal PPT violation observed is only of the order of
\begin{equation}
\min_{\rho\in \mathrm{PPT}_1}\Tr(W\rho)\sim -3\times 10^{-6},
\end{equation}
so that the corresponding PPT margin remains small. This indicates that, although the witnesses are numerically indecomposable, they typically lie  close to the decomposable cone. Hence, the KSMZ witnesses appear to detect PPT entanglement only weakly on average, which is consistent with a near-boundary geometric location. Since these values are close to the numerical precision of the SDP solver, we performed additional numerical consistency checks, which are reported in Appendix~\ref{app:numerics}.

The previous SDPs were written for a fixed witness $W$. In the first numerical tests, we apply them to the Choi witness associated with a KSMZ map. However, as explained in Section~\ref{subsec:recover_detection_power_family_witnesses}, a single positive map $\Phi$ gives rise to the larger family
\begin{equation}
W_{\Phi,\eta}
=
(\operatorname{Id}\otimes\Phi^\dagger)(\ket{\eta}\bra{\eta}),
\end{equation}
where $\ket{\eta}$ ranges over full-Schmidt-rank vectors. Therefore, the same numerical indecomposability analysis naturally extends to the whole family: for each fixed $\ket{\eta}$, one simply replaces $W$ by $W_{\Phi,\eta}$ in
\eqref{SDP:distance_dec} and \eqref{SDP:PPT_margin}. This gives the quantities
\begin{equation}
d_\infty(W_{\Phi,\eta},\mathrm{Dec})
\qquad \text{and} \qquad
\mu(W_{\Phi,\eta}),
\end{equation}
which quantify the distance to the decomposable cone and the PPT-detection margin of the particular hyperplane selected by $\ket{\eta}$.

In the square case $n=m$, the full-Schmidt-rank condition also gives a structural interpretation of this family. Every full-Schmidt-rank vector can be written as
\begin{equation}
\ket{\eta}
=
(I\otimes A_\eta)\ket{\Omega},
\end{equation}
with $A_\eta$ invertible. Hence $W_{\Phi,\eta}$ is the Choi operator of the transformed map
\begin{equation}
\Phi^\dagger\circ \operatorname{Ad}_{A_\eta},
\qquad
\operatorname{Ad}_{A_\eta}(X)=A_\eta X A_\eta^\dagger.
\end{equation}
Since $\operatorname{Ad}_{A_\eta}$ is an invertible completely positive map,
decomposability is preserved under this transformation. Thus, independently of numerical approximations, the full-Schmidt-rank witnesses inherit indecomposability from the original KSMZ map. Nevertheless, the numerical quantities $d_\infty(W_{\Phi,\eta},\mathrm{Dec})$ and $\mu(W_{\Phi,\eta})$ are not invariant under the choice of $\eta$. The family may contain witnesses that are closer or farther from the decomposable cone than the Choi witness. Consequently, Section~\ref{SDP:distance_dec} and Section~\ref{SDP:PPT_margin} provide not only a way to quantify the average indecomposability of the Choi witnesses, but also a practical diagnostic for optimizing over the family of witnesses associated with the same KSMZ map.

\subsection{Benchmark against general criteria}

To complete the characterization of the KSMZ maps, we used the MATLAB package QETLAB (Quantum Entanglement Theory Laboratory)~\cite{qetlab}. To our knowledge, it is currently the most comprehensive and up-to-date package for entanglement detection, owing to its broad catalogue of criteria. In particular, the function IsSeparable can decide whether a given density matrix is separable or entangled by applying a hierarchy of 19 tests, including 7 entanglement criteria.

Since the criteria implemented in \textit{IsSeparable} are arranged in a hierarchy of increasing strength, we attribute each detection to the first criterion in the sequence that certifies entanglement. This provides a consistent, non-overlapping classification of the detection mechanisms according to their position in the hierarchy. This can however hide the fact that the most powerful criteria can potentially detect all the states previously detected. We provide the list of criteria of the \textit{IsSeparable} function in its original order in Appendix~\ref{app:IsSeparable Criterion List}.

Among the 2,000 states that maximally violate the KSMZ witness in our dataset, \textit{IsSeparable} detected the entanglement of only 33 states, corresponding to 1.7\% of the total. It therefore failed to detect entanglement in the remaining 98.3\% of cases. Among the detected states, 9 were identified by the realignment criterion, 11 by the stronger realignment criterion and 13 by the filter covariance matrix criterion, as displayed in Fig~\ref{fig:isseparableresults}.

\begin{figure}[h!]
    \centering
    \includegraphics[width=1.\linewidth]{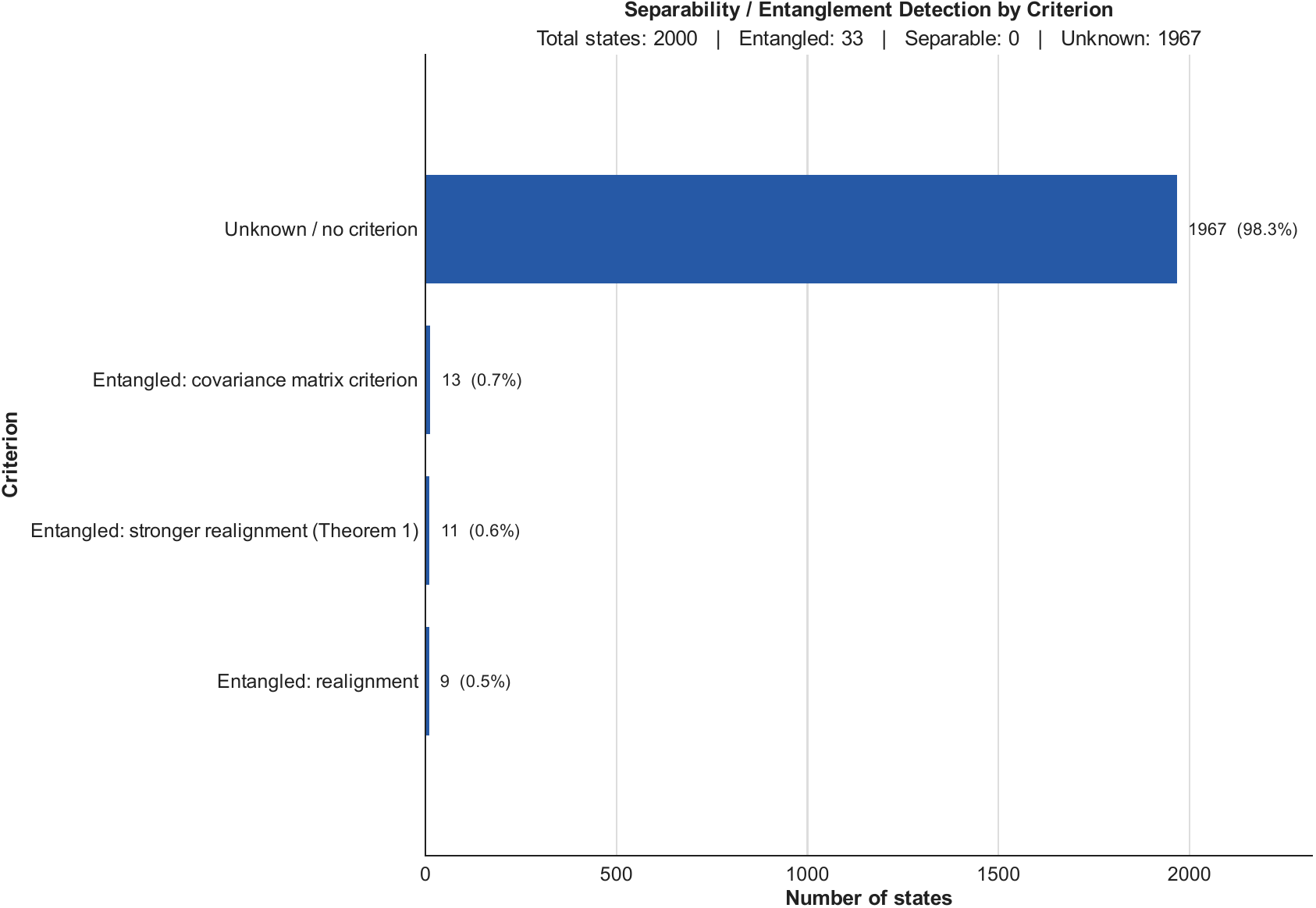}
    \caption{Results of the IsSeparable function for the entanglement detection of the 2,000 PPT entangled states that are maximally detected by the KSMZ maps. IsSeparable classifies loosely the criteria}
    \label{fig:isseparableresults}
\end{figure}

In order to complete this characterization, we searched the entanglement of these states with the DPS criterion only~\cite{Doherty2004CompleteFamily}, this time allowing for PPT symmetric extension up to level $4$, whereas it is by default set to $2$ in \textit{IsSeparable}. The DPS criterion is particularly interesting since it converges to the separable set as illustrated in Fig~\ref{fig:DPS_convergence}. However, its exploding computational cost prevents it from being able to detect all entangled states in practice. $15$ states out of the $2,000$ we considered were detected in this way. Interestingly, they were already detected at the level $2$ of the DPS hierarchy. 

PPT entanglement does not preclude Bell nonlocality, as demonstrated by Vértesi and Brunner \cite{vertesi2014disproving}. We therefore investigated whether Bell violations could also be detected in the 2,000 PPT-entangled states of our benchmark. We considered four binary-outcome measurement settings per party and tested the 174 nontrivial Bell inequalities of the corresponding classification \cite{Cruzeiro_2019}. For each state and inequality, we numerically minimized the Bell expression over local projective measurements using a see-saw algorithm with 30 random initializations, alternately optimizing each party’s observables while keeping the other party’s measurements fixed. No violation was found for any of the states, using a numerical violation threshold of $10^{-9}$. However, since this optimization is not guaranteed to reach the global optimum, the absence of detected violations does not establish Bell locality, even within the measurement scenario considered.

\begin{figure}[h!]
    \centering
    \includegraphics[width=0.7\linewidth]{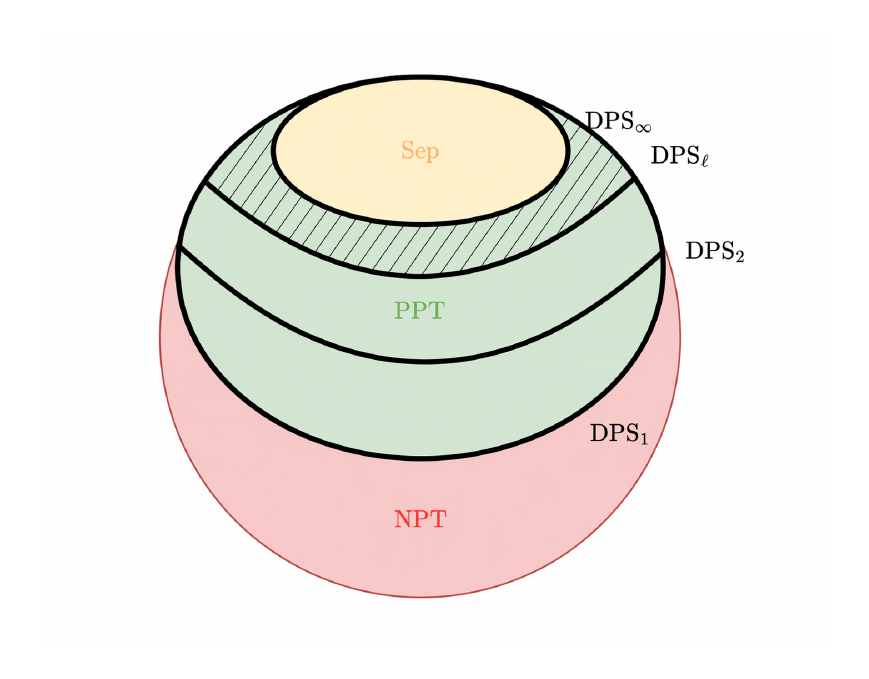}
    \caption{The DPS hierarchy converges to the set of separable states. At infinity, it can thus detect all entangled states. In practice, the computation cost explodes as the level of the hierarchy increases.}
    \label{fig:DPS_convergence}
\end{figure}


\subsection{Numerical characterization}

In the light of all these results, it is possible to propose a characterization of the KSMZ maps. These PnCP maps are capable of robustly detecting NPT states. These maps are inequivalent with either Choi, Breuer, or any PnCP maps we found so far. They are capable of detecting NPT states in a robust way, and,  in a non-robust way, PPT entangled states very close to the set of separable sets. These states are not detected by most-known criteria. In order to detect them with a DPS hierarchy, it would necessitate to reach levels of the hierarchy that are out of reach because of complexity issues. In this sense, these maps contribute significantly to the separability problem in dimensions $3 \times 3$ or $4 \times 4$. However, unlike what is said in~\cite{Bhardwaj2023PracticalApproach}, it is not possible to detect in polynomial time the entanglement of an unknown state by randomly producing and applying some of these maps. Still, the constitution of a large library of KSMZ maps could be incorporated with some benefit in a function like IsSeparable, as a test for entanglement inequivalent to any others. 

\par The authors were suggested at some stage to modify the algorithm, so that the polynomials and thus, the maps, that are created, would be adapted to a target unknown state. Indeed, since the generational part of the algorithm starts with a collection of random points, it would seem logical to implement a machine learning framework, that choses these points according to their capability to generate interesting maps with regard to the target. However, a careful analysis of the algorithm reveals that not only one, but two steps actually contain a choice of random objects. Because of that, we could verify that the polynomials generated do not depend continuously on the initial points, which prevents this line of development from becoming fruitful.




\section{Conclusion}
\label{sec:conclusion}

In this work, we presented a new criteria to detect entanglement based on PnCP maps. In practice, we implemented an algorithm proposed in~\cite{Klep2017ManyMorePositiveMaps} that generates positive non SoS polynomials and turns them into PnCP maps through a well chosen isomorphism. Then, we numerically produced these maps \textit{en masse}, after having added a test to ensure no numerical imprecision could hinder our results. We theoretically investigated the entanglement power of these maps, showed they were inequivalent with most known PnCPs and explored their location within the set of positive maps. We showed that they could detect NPT states in a robust way, and that they could detect PPT states that most criteria failed to detect. Doing so, we presented a clear framework of the relations between polynomials, maps and operators for entanglement detection. We hope that this will serve as an incentive for further explorations of the rich frameworks of polynomial optimization in the quantum information community. As far as this work is concerned, we plan to extend it in several aspects. First, we would like to explore the construction of other PnCP maps using the construction of positive non \SOS polynomials, \textit{à la}~\cite{buckley2022newexamplesentangled} for instance. Second, we would like to construct new PnCP maps from the KSMZ maps by adding a skew-symetric part. This problem is NP-hard in general because it necessitates checking the positivity of the map from scratch. However, we hope that a careful understanding of the location of the KSMZ maps in the set of positive maps could be fruitfully used. 

\section{Acknowledgments}
\label{sec:acknowledgments}
G.M and E. O acknowledge funding from ANR-22-PETQ-0009,ANR-24-CE97-0005-01, ANR-23-QUAC-0001 and ANR-22-PETQ-0011. M.R acknowledges fundings from ANR-22-PETQ-0007 and from Thales SIX. 
This project lasted many years and has led to a great number of interesting discussions. We would particularly like to emphasize that fruitful discussions occurred with Lucas Porto, Lucas Viera, Marco Tulio Quintino, Ulysse Chabaud and Seiseki Akibue.

\newpage


\onecolumn
\appendix
\newpage

\section{Notations}
\label{app:preliminaries}

\subsection{Operator notation}
\label{subsec:operator_notation_cones}

Throughout the paper, all Hilbert spaces are assumed to be finite-dimensional and complex. 
We denote by $\mathcal L(\mathcal H,\mathcal K)$ the vector space of linear operators from 
$\mathcal H$ to $\mathcal K$. When $\mathcal H=\mathcal K$, we simply write 
$\mathcal L(\mathcal H)$. After fixing an orthonormal basis of $\mathcal H$, with 
$n=\dim\mathcal H$, the space $\mathcal L(\mathcal H)$ is identified with the matrix 
algebra $M_n(\mathbb C)$.

For an operator $O:\mathcal H\to\mathcal K$, its operator norm is
\begin{equation}
    \|O\|_\infty
    =
    \sup_{\|x\|_{\mathcal H}\leq 1}
    \|O(x)\|_{\mathcal K}.
\end{equation}
Its trace norm is defined by
\begin{equation}
    \|O\|_1
    =
    \Tr(|O|)
    =
    \Tr\!\left(\sqrt{O^\dagger O}\right).
\end{equation}
Since we work in finite dimension, every linear operator is bounded and trace-class. 
Thus, no distinction between bounded, trace-class and linear operators is needed in what 
follows.

We write $O\succeq 0$ when $O$ is positive semidefinite. A density operator on 
$\mathcal H$ is an operator $\rho\in\mathcal L(\mathcal H)$ such that
\begin{equation}
    \rho=\rho^\dagger,
    \qquad
    \rho\succeq 0,
    \qquad
    \Tr(\rho)=1.
\end{equation}

Let $\operatorname{Herm}(\mathcal H)$ denote the real vector space of Hermitian operators 
on $\mathcal H$. For a bipartite system $\mathcal H_A\otimes\mathcal H_B$, we write
\begin{equation}
    \mathcal H_{AB}:=\mathcal H_A\otimes\mathcal H_B.
\end{equation}

\begin{definition}[Positive cone of operators]
The positive cone of operators on $\mathcal H_{AB}$ is
\begin{equation}
    \mathsf{PSD}(\mathcal H_{AB})
    :=
    \left\{
        W\in \operatorname{Herm}(\mathcal H_{AB})
        :
        W\succeq 0
    \right\}.
\end{equation}
\end{definition}

\begin{definition}[Block-positive cone]
The block-positive cone is

\begin{equation}
    \mathsf{BP}(\mathcal H_A,\mathcal H_B)
    :=
    \left\{
        W\in \operatorname{Herm}(\mathcal H_A\otimes\mathcal H_B)
        :
        \langle x\otimes y|W|x\otimes y\rangle\geq 0
        \ \forall x\in\mathcal H_A,\ \forall y\in\mathcal H_B
    \right\}.
\end{equation}

\end{definition}

\begin{definition}[Cone of decomposable operators]
Let $T_B$ denote the partial transposition on the second subsystem. The cone of decomposable operators is

\begin{equation}
    \mathsf{Dec}_{Op}(\mathcal H_A,\mathcal H_B)
    :=
    \left\{
        W=P+Q^{T_B}
        :
        P,Q\in \mathsf{PSD}(\mathcal H_A\otimes\mathcal H_B)
    \right\}.
\end{equation}

\end{definition}

These cones satisfy the inclusions
\begin{equation}
    \mathsf{PSD}(\mathcal H_{AB})
    \subseteq
    \mathsf{DecOp}(\mathcal H_A,\mathcal H_B)
    \subseteq
    \mathsf{BP}(\mathcal H_A,\mathcal H_B).
\end{equation}

\subsection{Projective notation}
\label{subsec:projective_notation}

Let $\mathbb K$ denote either $\mathbb R$ or $\mathbb C$. For an integer $N\geq 1$, we 
denote by $\mathbb P^{N-1}_{\mathbb K}$ the projective space associated with 
$\mathbb K^N$. It is defined as
\begin{equation}
    \mathbb P^{N-1}_{\mathbb K}
    =
    \big(\mathbb K^N\setminus\{0\}\big)/\sim,
\end{equation}
where the equivalence relation is given by
\begin{equation}
    u\sim v
    \quad\Longleftrightarrow\quad
    \exists \lambda\in\mathbb K^*
    \text{ such that }
    v=\lambda u.
\end{equation}
Here $\mathbb K^*$ denotes the set of nonzero elements of $\mathbb K$.

The equivalence class of a nonzero vector $u\in\mathbb K^N$ is denoted by
\begin{equation}
    [u]
    =
    \{\lambda u:\lambda\in\mathbb K^\times\}.
\end{equation}
Thus, a projective point $[u]\in\mathbb P^{N-1}_{\mathbb K}$ represents the one-dimensional 
linear subspace generated by $u$. In particular,
\begin{equation}
    [u]=[\lambda u],
    \qquad
    \forall \lambda\in\mathbb K^\times.
\end{equation}

\subsection{Polynomial cones}
\label{subsec:polynomial_cones}

Let $\mathbb R[x,y]_{2,2}$ denote the vector space of real bihomogeneous polynomials 
of bidegree $(2,2)$ in the variables $x\in\mathbb R^n$ and $y\in\mathbb R^m$.

\begin{definition}[Positive cone of biquadratic forms]
The positive cone of biquadratic forms is

\begin{equation}
    \mathrm{Pos}_{2,2}^{n,m}
    :=
    \left\{
        p\in \mathbb R[x,y]_{2,2}
        :
        p(x,y)\geq 0
        \ \forall x\in\mathbb R^n,\ 
        \forall y\in\mathbb R^m
    \right\}.
\end{equation}

\end{definition}

\begin{definition}[Boundary of the positive cone]
The boundary of $\mathrm{Pos}_{2,2}^{n,m}$ is
\begin{equation}
    \partial \mathrm{Pos}_{2,2}^{n,m}
    :=
    \overline{\mathrm{Pos}_{2,2}^{n,m}}
    \setminus
    \operatorname{int}\!\left(\mathrm{Pos}_{2,2}^{n,m}\right).
\end{equation}
\end{definition}

\begin{definition}[SOS cone]
The sum-of-squares cone is
\begin{equation}
    \mathrm{SOS}_{2,2}^{n,m}
    :=
    \left\{
        p\in \mathbb R[x,y]_{2,2}
        :
        p(x,y)=\sum_{r=1}^R h_r(x,y)^2
    \right\},
\end{equation}
where each $h_r$ is a real bilinear form in $(x,y)$. Through this paper we use the simple notation $\SOS$ which refer to $\mathrm{SOS}_{2,2}^{n,m}$
\end{definition}

\subsection{Cones of maps}
\label{subsec:map_cones}

Let $S_n(\mathbb R)$ denote the vector space of real symmetric $n\times n$ matrices. 
We consider real linear maps
\[
    \Phi:S_n(\mathbb R)\longrightarrow S_m(\mathbb R).
\]

\begin{definition}[Positive cone of maps]
The cone of positive maps from $S_n(\mathbb R)$ to $S_m(\mathbb R)$ is
\begin{equation}
    \mathcal P_{n,m}
    :=
    \left\{
        \Phi\in \mathcal L(S_n(\mathbb R),S_m(\mathbb R))
        :
        X\succeq 0 \Longrightarrow \Phi(X)\succeq 0
    \right\}.
\end{equation}
\end{definition}

\begin{definition}[Boundary of the positive-map cone]
The boundary of $\mathcal P_{n,m}$ is
\begin{equation}
    \partial \mathcal P_{n,m}
    :=
    \overline{\mathcal P_{n,m}}
    \setminus
    \operatorname{int}(\mathcal P_{n,m}).
\end{equation}
\end{definition}

\begin{definition}[Cone of completely positive maps]
The cone of completely positive maps is
\begin{equation}
    \mathcal{CP}_{n,m}
    :=
    \left\{
        \Phi\in \mathcal L(S_n(\mathbb R),S_m(\mathbb R))
        :
        \Phi(X)=\sum_{r=1}^R A_r X A_r^T
    \right\},
\end{equation}

where $A_r\in M_{m,n}(\mathbb R)$.
\end{definition}

\newpage

\section{Explanation of the KSMZ Algorithm}
\label{app:algo}

\subsection{Introduction to the Segre variety}
The Segre variety is an algebraic variety that realizes the Cartesian product of two projective spaces as a projective subvariety.
Fix $n,m\ge 2$ and consider the complex projective spaces $\mathbb P^{n-1}_{\mathbb C}$ and $\mathbb P^{m-1}_{\mathbb C}$.

\begin{definition}[Segre map]
\label{def:segre-map}
\begin{align*}
\sigma_{n,m}:\ &\mathbb P^{n-1}_{\mathbb C}\times \mathbb P^{m-1}_{\mathbb C}\longrightarrow \mathbb P^{nm-1}_{\mathbb C},\\
&([x],[y])\longmapsto [x\otimes y]
= [x_i y_j]_{1\le i\le n,\ 1\le j\le m}.
\end{align*}
\end{definition}

\begin{figure}[h]
    \centering
    \includegraphics[width=0.7\linewidth]{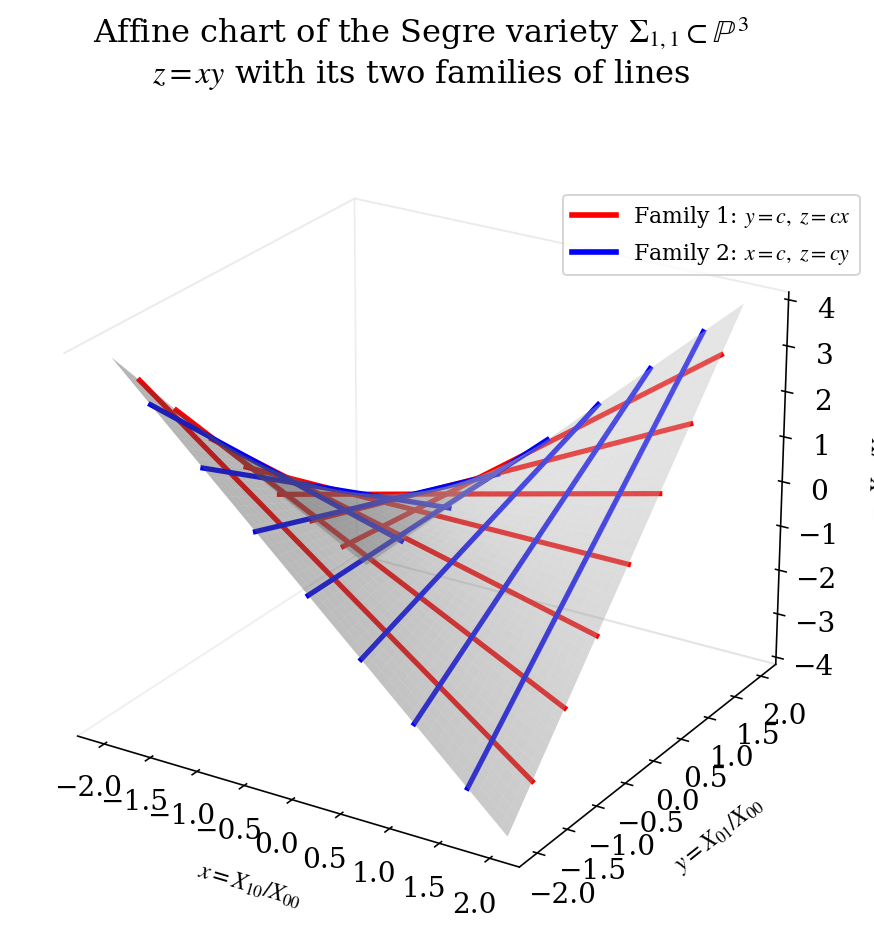}
    \caption{Affine chart of the Segre variety 
    $S_{1,1}\subset \mathbb P^3$, the image of the Segre embedding
    $\mathbb P^1\times\mathbb P^1\hookrightarrow\mathbb P^3$.
    In the chart $X_{00}\neq 0$, it is described by the equation
    $z=xy$. The red and blue lines show the two families of rulings,
    obtained by fixing one of the two projective coordinates.}
    \label{fig:segre-11}
\end{figure}

\begin{definition}[Segre variety]
\label{def:segre-variety}
The Segre variety (or Segre embedding) is the projective subvariety
\[
S_{n-1,m-1}\;:=\;\sigma_{n,m}\big(\mathbb P^{n-1}_{\mathbb R}\times \mathbb P^{m-1}_{\mathbb R}\big)\ \subset\ \mathbb P^{nm-1}_{\mathbb R}.
\]
Equivalently, $S_{n-1,m-1}$ is the set of projective classes of rank-one matrices in $\mathbb R^{n\times m}$:
\[
S_{n-1,m-1}=\big\{[Z]\in \mathbb P^{nm-1}_{\mathbb C}:\ \mathrm{rank}(Z)=1\big\}.
\]
\end{definition}

Let $z_{ij}$ ($1\le i\le n$, $1\le j\le m$) be homogeneous coordinates on $\mathbb P^{nm-1}_{\mathbb R}$, collected into an
$n\times m$ matrix of indeterminates $Z=(z_{ij})$.
The Segre variety is cut out by the $2\times 2$ minors of $Z$.

\begin{definition}[Defining equations of the Segre variety]
\label{def:segre-ideal}
Let $I_{n,m}\subset \mathbb R[z_{ij}]$ be the homogeneous ideal generated by all $2\times 2$ minors of $Z=(z_{ij})$, i.e.
\[
z_{ij}z_{k\ell}-z_{i\ell}z_{kj}\qquad (1\le i<k\le n,\ 1\le j<\ell\le m).
\]
Then
\[
S_{n-1,m-1}=V(I_{n,m})\subset \mathbb P^{nm-1}_{\mathbb R}.
\]
\end{definition}

The algorithm in \cite{Klep2017ManyMorePositiveMaps} is naturally expressed in the coordinate ring of the Segre variety. Consider the graded ring homomorphism induced by Definition \ref{def:segre-map} :
\begin{align*}
    \sigma^{\#}_{n,m} : &\mathbb R[z_{ij}] \to \mathbb R[x_1,...,x_n,y_1,...,y_m] \\
    &\sigma_{n,m}^{\#}(z_{ij}) \mapsto x_iy_j
\end{align*}
Since all $2\times 2$ minors vanish on rank-one matrices, we have $I_{n,m} \subset \ker(\sigma^\#_{n,m})$ and hence $\sigma^\#_{n,m}$
factors through the quotient:
\begin{equation}
\label{eq:sigma-sharp-quotient}
\widetilde{\sigma}^\#_{n,m}:\ \mathbb R[z_{ij}]/I_{n,m} \longrightarrow \mathbb R[x,y],
\qquad
\widetilde{\sigma}^\#_{n,m}([z_{ij}])=x_i y_j.
\end{equation}
In particular, quadratic forms in $\mathbb R[z_{ij}]/I_{n,m}$ pull back to biquadratic forms (bidegree $(2,2)$) in the variables $(x,y)$.

\subsection{Explanation of the KSMZ Algorithm}

Algorithm \cite{Klep2017ManyMorePositiveMaps} produces such a biquadratic form via a quadratic form on the Segre variety.

\subsection{ Choose generic Segre points}
Pick real points on the Segre variety,
\[
s_k \in S_{n-1,m-1}(\mathbb R)\subset \mathbb P^{nm-1}_{\mathbb R},
\qquad k=1,\dots,N,
\]
in general position (the algorithm is designed so that the required rank/dimension conditions hold with probability $1$ for random choices).

Let $J\subset \mathbb R[z]/I_{n,m}$ denote the homogeneous ideal of the chosen point set $\{s_1,\dots,s_N\}$ in the coordinate ring of the Segre variety.

\subsection{Build a linear system of vanishing linear forms}
Let $J_1$ be the degree-1 part of $J$. Choose a basis
\[
h_0,\dots,h_d \in J_1,
\qquad d=n+m-2,
\]
so that each $h_i$ vanishes on all $s_k$:
\[
h_i(s_k)=0\quad \forall\, i,\ k.
\]
The sum of squares
\begin{equation}\label{eq:sum-h2}
H(z):=\sum_{i=0}^{d} h_i(z)^2
\end{equation}
is automatically nonnegative on $S_{n-1,m-1}(\mathbb R)$ and satisfies $H(s_k)=0$ for all selected points.

\subsection{Pick a quadratic perturbation with \textit{second-order} vanishing}
Let $J_2$ be the degree-2 part of $J$, and let $J^{(2)}_2$ denote the subspace of degree-2 forms that vanish to \textit{second order}
at each $s_k$ (equivalently: they vanish and all first derivatives vanish at those points, in any affine chart).
Choose
\[
f \in J^{(2)}_2
\quad\text{such that}\quad
f \notin \mathrm{span}\{h_i h_j:\ 0\le i\le j\le d\}.
\]
The condition $f \notin \mathrm{span}\{h_i h_j\}$ is the key algebraic obstruction that will force the final form to be \textit{non-SOS}
in the quotient ring (hence ``non-CP'' on the map side).

\subsection{Form \texorpdfstring{$q=\delta f + \sum h_i^2$}{q=delta f + sum h_i^2} and choose $\delta>0$}
Define a one-parameter family of quadratic forms in $\mathbb R[z]/I_{n,m}$:
\begin{equation}\label{eq:qdelta}
q_\delta(z):=\delta\,f(z)+\sum_{i=0}^{d} h_i(z)^2.
\end{equation}


\begin{lemma}[Nonnegativity on the real Segre variety for small $\delta$]\label{lem:small-delta}
Assume that the common zero set of $J_1=\mathrm{span}\{h_0,\dots,h_d\}$ on the real Segre variety
is exactly the prescribed finite set of points:
\begin{equation}\label{eq:zero-set-assumption}
\{\,s\in S_{n-1,m-1}(\mathbb R):\ h_i(s)=0\ \forall i\,\}=\{s_1,\dots,s_N\}.
\end{equation}
Assume moreover that $f$ vanishes to second order at each $s_k$ along $S_{n-1,m-1}(\mathbb R)$.
Then there exists $\delta_0>0$ such that for all $0<\delta\le \delta_0$,
\[
q_\delta(z):=\delta f(z)+\sum_{i=0}^{d} h_i(z)^2 \ \ge\ 0
\quad \text{for all } [z]\in S_{n-1,m-1}(\mathbb R).
\]
\end{lemma}

\begin{proof}
Let $S:=S_{n-1,m-1}(\mathbb R)$ denote the real Segre variety. As a real projective algebraic variety, $S$ is compact.

Define
\[
H(z):=\sum_{i=0}^{d} h_i(z)^2.
\]
By construction, $H\ge 0$ on $S$ and $H(s_k)=0$ for all $k$.

\subsubsection*{Uniform positivity away from the common zero set.}
Let
\[
Z:=\{z\in S:\ H(z)=0\}.
\]
By \eqref{eq:zero-set-assumption}, $Z=\{s_1,\dots,s_N\}$ is finite.

Fix $\varepsilon>0$ and consider the closed set
\[
S_\varepsilon:=S\setminus \bigcup_{k=1}^N B_\varepsilon(s_k),
\]
where $B_\varepsilon(s_k)$ denotes a small neighborhood of $s_k$ in $S$ (in any fixed metric induced by an embedding in Euclidean space).
Since $S_\varepsilon$ is compact and does not intersect $Z$, the continuous function $H$ attains a strictly positive minimum on $S_\varepsilon$:
\[
m_\varepsilon:=\min_{z\in S_\varepsilon} H(z) > 0.
\]
Likewise, $f$ is continuous on $S$, so it is bounded on $S_\varepsilon$:
\[
M_\varepsilon:=\max_{z\in S_\varepsilon} |f(z)| <\infty.
\]
Choose
\[
\delta_\varepsilon:=\frac{m_\varepsilon}{2M_\varepsilon}\quad (\text{with the convention }\delta_\varepsilon=1\text{ if }M_\varepsilon=0).
\]
Then for any $0<\delta\le \delta_\varepsilon$ and any $z\in S_\varepsilon$,
\[
q_\delta(z)=H(z)+\delta f(z)\ \ge\ H(z)-\delta |f(z)|
\ \ge\ m_\varepsilon - \delta M_\varepsilon
\ \ge\ m_\varepsilon - \frac{m_\varepsilon}{2}
\ =\ \frac{m_\varepsilon}{2}>0.
\]
Hence $q_\delta\ge 0$ on $S_\varepsilon$ for all $0<\delta\le \delta_\varepsilon$.

\medskip
\subsubsection*{Local nonnegativity near each prescribed zero using second-order vanishing.}
Fix $k\in\{1,\dots,N\}$. Work in an affine chart of $\mathbb P^{nm-1}$ containing $s_k$ and choose local coordinates
$u\in\mathbb R^{\dim S}$ on $S$ near $s_k$, with $u=0$ corresponding to $s_k$.
Let $\widehat{H}(u)=H(s_k(u))$ and $\widehat{f}(u)=f(s_k(u))$ denote the pullbacks to these local coordinates.

Because each $h_i$ is linear and $h_i(s_k)=0$, we have $\widehat{H}(u)=O(\|u\|^2)$ and $\widehat{H}(u)\ge 0$.
Moreover, by the genericity/rank conditions in Algorithm~4.1 (equivalently: the differentials of the $h_i$ span the cotangent space at $s_k$),
there exists a constant $c_k>0$ and a neighborhood $U_k$ of $0$ such that
\begin{equation}\label{eq:H-lower-local}
\widehat{H}(u)\ \ge\ c_k \|u\|^2\qquad\forall u\in U_k.
\end{equation}
(Geometrically, this is strict positive definiteness of the Hessian of $H$ along the tangent space of $S$ at $s_k$.)

Next, since $f$ vanishes to second order at $s_k$ along $S$, we have
\[
\widehat{f}(0)=0,\qquad \nabla\widehat{f}(0)=0,
\]
hence Taylor's theorem yields
\[
|\widehat{f}(u)| \le C_k \|u\|^2 \qquad\forall u\in U_k'
\]
for some constant $C_k>0$ and a (possibly smaller) neighborhood $U_k'\subset U_k$.

Set
\[
\delta_k:=\frac{c_k}{2C_k}\quad (\text{with }\delta_k=1\text{ if }C_k=0).
\]
Then for all $0<\delta\le \delta_k$ and all $u\in U_k'$,
\[
\widehat{q}_\delta(u)=\widehat{H}(u)+\delta \widehat{f}(u)
\ \ge\ \widehat{H}(u)-\delta|\widehat{f}(u)|
\ \ge\ c_k\|u\|^2 - \delta C_k\|u\|^2
\ \ge\ \frac{c_k}{2}\|u\|^2\ \ge\ 0.
\]
Thus $q_\delta\ge 0$ on a neighborhood of $s_k$ in $S$ for all $0<\delta\le \delta_k$.

\medskip
\subsubsection*{Global choice of $\delta_0$.}
Choose $\varepsilon>0$ small enough so that the neighborhoods $B_\varepsilon(s_k)$ are contained in the coordinate neighborhoods
where Step 2 applies. Define
\[
\delta_0:=\min\{\delta_\varepsilon,\delta_1,\dots,\delta_N\}>0.
\]
Then for any $0<\delta\le \delta_0$, Step 1 gives $q_\delta\ge 0$ on $S_\varepsilon$, while Step 2 gives $q_\delta\ge 0$ on each $B_\varepsilon(s_k)$.
Since $S=S_\varepsilon\cup \bigcup_k B_\varepsilon(s_k)$, this proves $q_\delta\ge 0$ on all of $S$.
\end{proof}


\begin{lemma}[Non-SOS in the quotient ring]\label{lem:nonsos}
Let $J\subset \mathbb R[z]/I_{n,m}$ be the homogeneous ideal of the point set $\{s_1,\dots,s_N\}\subset S_{n-1,m-1}(\mathbb R)$,
and let $J_1$ be its degree-1 component. Assume $h_0,\dots,h_d$ form a basis of $J_1$, so that
\[
J_1=\mathrm{span}\{h_0,\dots,h_d\}.
\]
Let $f\in J_2^{(2)}$ be a quadratic form vanishing to second order at each $s_k$, and assume
\begin{equation}\label{eq:f-not-span}
f\ \notin\ \mathrm{span}\{h_i h_j:\ 0\le i\le j\le d\}.
\end{equation}
Then for every $\delta>0$, the quadratic form
\[
q_\delta=\delta f+\sum_{i=0}^d h_i^2 \in (\mathbb R[z]/I_{n,m})_2
\]
is \emph{not} a sum of squares in the quotient ring $\mathbb R[z]/I_{n,m}$.
\end{lemma}

\begin{proof}
Suppose, for contradiction, that $q_\delta$ is SOS in $\mathbb R[z]/I_{n,m}$ for some $\delta>0$.
Since $q_\delta$ has degree $2$, any SOS representation can be taken with linear summands:
there exist linear forms $\ell_1,\dots,\ell_r\in (\mathbb R[z]/I_{n,m})_1$ such that
\begin{equation}\label{eq:q-sos-linear}
q_\delta=\sum_{k=1}^r \ell_k^2.
\end{equation}

\medskip
\subsubsection*{SOS zeros force each summand to vanish.}
For every prescribed point $s_j$, we have $h_i(s_j)=0$ for all $i$, and $f(s_j)=0$ (since $f$ vanishes to second order), hence
\[
q_\delta(s_j)=\delta f(s_j)+\sum_{i=0}^d h_i(s_j)^2=0.
\]
Evaluating \eqref{eq:q-sos-linear} at $s_j$ gives
\[
0=q_\delta(s_j)=\sum_{k=1}^r \ell_k(s_j)^2.
\]
Over $\mathbb R$, a sum of squares is zero if and only if each term is zero, so
\begin{equation}\label{eq:ell-vanish}
\ell_k(s_j)=0\qquad \forall\,k,\ \forall\,j.
\end{equation}
Thus each $\ell_k$ vanishes on the entire set $\{s_1,\dots,s_N\}$, meaning $\ell_k\in J_1$ for all $k$.

\medskip
\subsubsection*{Conclude that $q_\delta$ lies in the quadratic span of $J_1$.}
Since $\ell_k\in J_1=\mathrm{span}\{h_0,\dots,h_d\}$, we can write
\[
\ell_k=\sum_{i=0}^d \alpha_{k,i} h_i
\]
for real coefficients $\alpha_{k,i}$. Therefore
\[
q_\delta=\sum_{k=1}^r \ell_k^2
= \sum_{k=1}^r \left(\sum_{i=0}^d \alpha_{k,i} h_i\right)^2
\in \mathrm{span}\{h_i h_j:\ 0\le i\le j\le d\}.
\]
In particular, $q_\delta$ belongs to the subspace $\mathrm{span}\{h_i h_j\}$.

\medskip
\subsubsection*{Contradiction with the choice of $f$.}
But by definition,
\[
q_\delta=\delta f+\sum_{i=0}^d h_i^2.
\]
The term $\sum_{i=0}^d h_i^2$ is already in $\mathrm{span}\{h_i h_j\}$. Hence if $q_\delta$ is also in $\mathrm{span}\{h_i h_j\}$,
then so is $\delta f = q_\delta-\sum_i h_i^2$, and therefore $f\in \mathrm{span}\{h_i h_j\}$, contradicting \eqref{eq:f-not-span}.
This proves that no such SOS representation \eqref{eq:q-sos-linear} can exist, i.e.\ $q_\delta$ is not SOS.
\end{proof}


\newpage
\section{Proving the indecomposability of maps thanks to the Kira polynomials: the case of the Choi map}
\label{app:ChoiIndec}

\begin{proposition}[Indecomposability of the Choi map from the real slice]
Let \(\Phi_{\mathrm{Ch}}:M_3(\mathbb C)\to M_3(\mathbb C)\) be the Choi map
\[
\Phi_{\mathrm{Ch}}(X)
=
\begin{pmatrix}
x_{11}+x_{33} & -x_{12} & -x_{13}\\
-x_{21} & x_{22}+x_{11} & -x_{23}\\
-x_{31} & -x_{32} & x_{33}+x_{22}
\end{pmatrix}.
\]
Then \(\Phi_{\mathrm{Ch}}\) is indecomposable.
\end{proposition}

\begin{proof}
Consider the Hermitian biquadratic polynomial associated with the Choi map defined in~\eqref{eq:choi_polynomial},
\[
p_{\mathrm{Ch}}(z,w)
=
\bra{w}\Phi_{\mathrm{Ch}}(\ket z\bra z)\ket w .
\]
Restricting to real variables \(z=x\in\mathbb R^3\), \(w=y\in\mathbb R^3\), one obtains
\begin{align}
p_{\mathrm{Ch}}(x,y)
&=
(x_1^2+x_3^2)y_1^2
+
(x_2^2+x_1^2)y_2^2
+
(x_3^2+x_2^2)y_3^2  \nonumber\\
&\quad
-2x_1x_2y_1y_2
-2x_1x_3y_1y_3
-2x_2x_3y_2y_3 .
\end{align}

We first show that this real biquadratic form is not SOS. Suppose, by contradiction, that
\[
p_{\mathrm{Ch}}(x,y)=\sum_r q_r(x,y)^2
\]
with real bilinear forms
\[
q_r(x,y)=\sum_{i,j=1}^3 a^{(r)}_{ij}x_i y_j .
\]
Since the monomials \(x_1^2y_3^2\), \(x_2^2y_1^2\), and \(x_3^2y_2^2\) do not appear in
\(p_{\mathrm{Ch}}\), their coefficients must vanish:
\[
\sum_r \bigl(a^{(r)}_{13}\bigr)^2
=
\sum_r \bigl(a^{(r)}_{21}\bigr)^2
=
\sum_r \bigl(a^{(r)}_{32}\bigr)^2
=
0.
\]
Hence
\[
a^{(r)}_{13}=a^{(r)}_{21}=a^{(r)}_{32}=0
\qquad \text{for all } r.
\]

Now compare the coefficients of the diagonal monomials
\(x_i^2y_i^2\). We get
\[
\sum_r \bigl(a^{(r)}_{11}\bigr)^2
=
\sum_r \bigl(a^{(r)}_{22}\bigr)^2
=
\sum_r \bigl(a^{(r)}_{33}\bigr)^2
=
1.
\]
Next, compare the coefficients of
\(x_1x_2y_1y_2\), \(x_1x_3y_1y_3\), and \(x_2x_3y_2y_3\). Since
\(a^{(r)}_{21}=a^{(r)}_{13}=a^{(r)}_{32}=0\), one obtains
\[
\sum_r a^{(r)}_{11}a^{(r)}_{22}=-1,
\qquad
\sum_r a^{(r)}_{11}a^{(r)}_{33}=-1,
\qquad
\sum_r a^{(r)}_{22}a^{(r)}_{33}=-1.
\]

Define the real vectors
\[
u=(a^{(r)}_{11})_r,\qquad
v=(a^{(r)}_{22})_r,\qquad
t=(a^{(r)}_{33})_r .
\]
The coefficient identities give
\[
\|u\|=\|v\|=\|t\|=1,
\]
and
\[
u\cdot v=-1,\qquad u\cdot t=-1,\qquad v\cdot t=-1.
\]
By equality in Cauchy--Schwarz, \(u\cdot v=-1\) implies \(v=-u\), and
\(u\cdot t=-1\) implies \(t=-u\). Hence \(v=t\), so \(v\cdot t=1\), contradicting
\(v\cdot t=-1\). Therefore \(p_{\mathrm{Ch}}(x,y)\) is not SOS.

Now suppose that \(p_{\mathrm{Ch}}(z,w)\) were RSOS. Then there would exist Hermitian
polynomials \(g_k(z,\bar z,w,\bar w)\) such that
\[
p_{\mathrm{Ch}}(z,w)=\sum_k g_k(z,\bar z,w,\bar w)^2 .
\]
Restricting this identity to real variables \(z=x\), \(w=y\) gives a real SOS
decomposition of \(p_{\mathrm{Ch}}(x,y)\), contradicting what we have just proved.
Therefore \(p_{\mathrm{Ch}}(z,w)\) is not RSOS.

By the correspondence between decomposable maps and RSOS Hermitian polynomials,
\(\Phi_{\mathrm{Ch}}\) is not decomposable. Since the Choi map is positive, it is
therefore an indecomposable positive map.
\end{proof}

\newpage
\section{Technical identities for the complexified KSMZ polynomial}
\label{app:complexified_KSMZ_polynomials}

\subsection{Real-variable expansion of the Hermitian polynomial $p_{\Gamma_{\mathbb C}}$}
\label{app:derivation_Herm_pol}

Let $\Gamma_{\mathbb C}=(\Phi \oplus 0)$ be the complexified map obtained from the algorithm \ref{sec:algoimpl} and defined in \eqref{eq:complexification}, and $p_{\Gamma_\mathbb C}$ be its associated Hermitian polynomial  expressed as 
\begin{equation}
    p_{\Gamma_\mathbb C}(z,w)=\bra{w}\Gamma_\mathbb C(\ket{z}\bra{z})\ket{w} \label{eq:pol}
\end{equation}
with $\ket{z} \in \mathbb C^n$ and $\ket{w} \in \mathbb C^m$. 
Now fix $\ket{z}=\ket{a}+i\ket{b}$ and $\ket{w}=\ket{c}+i\ket{d}$, where $\ket{a},\ket{b} \in \mathbb R^n$ and  $\ket{c}, \ket{d} \in \mathbb R^m$. Then 
\begin{align}
    \ket{z}\bra{z}&=\big(\ket{a}+i\ket{b}\big)\big(\bra{a}-i\bra{b}\big) \nonumber \\
    &=\ket{a}\bra{a}+\ket{b}\bra{b}-i\big(\ket{a}\bra{b}-\ket{b}\bra{a}\big) \label{eq:realified_variable}
\end{align}
Then using \eqref{eq:complexification} and by injecting \eqref{eq:realified_variable} in \eqref{eq:pol}
\begin{align}
     p_{\Gamma_\mathbb C}(z,w)&=\big(\ket{c}-i\ket{d}\big)\Phi\big(\ket{a}\bra{a}+\ket{b}\bra{b}\big)\big(\ket{c}+i\ket{d}\big) \nonumber \\
     &=\bra{c}\Phi\big(\ket{a}\bra{a}+\ket{b}\bra{b}\big)\ket{c}+\bra{d}\Phi\big(\ket{a}\bra{a}+\ket{b}\bra{b}\big)\ket{d} \nonumber \\ &+i\big(\bra{c}\Phi\big(\ket{a}\bra{a}+\ket{b}\bra{b}\big)\ket{d}-\bra{d}\Phi\big(\ket{a}\bra{a}+\ket{b}\bra{b}\big)\ket{c}\big)
\end{align}
Since $\Phi \in S_n(\mathbb R)$, $\bra{c}\Phi\big(\ket{a}\bra{a}+\ket{b}\bra{b}\big)\ket{d}-\bra{d}\Phi\big(\ket{a}\bra{a}+\ket{b}\bra{b}\big)\ket{c}=0$ and then 
\begin{align}
     p_{\Gamma_\mathbb C}(z,w)&=\bra{c}\Phi\big(\ket{a}\bra{a}\big)\ket{c}+\bra{c}\Phi\big(\ket{b}\bra{b}\big)\ket{c}+\bra{d}\Phi\big(\ket{a}\bra{a}\big)\ket{d}+\bra{d}\Phi\big(\ket{b}\bra{b}\big)\ket{d}\nonumber \\
     &=p_\Phi(a,c)+p_\Phi(b,c)+p_\Phi(a,d)+p_\Phi(b,d)
\end{align}
That recover the expression of the Hermitian polynomial stated in \eqref{eq:pol-dec}. 
\subsection{Vanishing of bihomogeneous forms on zero slices}
\label{app:bihomogeneous_zero_argument}
\begin{lemma}

Let $p\in\mathbb R[x,y]_{2,2}$ be bihomogeneous of bidegree $(2,2)$.
Then, for every $x\in\mathbb R^n$ and every $y\in\mathbb R^m$,
\[
    p(x,0)=0,
    \qquad
    p(0,y)=0,
    \qquad
    p(0,0)=0.
\]
\end{lemma}

\begin{proof}
Since $p$ is bihomogeneous of bidegree $(2,2)$, every monomial appearing in
$p$ has total degree $2$ in the variables $x=(x_1,\dots,x_n)$ and total degree
$2$ in the variables $y=(y_1,\dots,y_m)$. Thus $p$ can be written as
\[
    p(x,y)
    =
    \sum_{i,j=1}^n
    \sum_{k,l=1}^m
    c_{ij,kl}\,x_i x_j y_k y_l .
\]
If $y=0$, then every factor $y_k y_l$ vanishes, hence
\[
    p(x,0)=0.
\]
Similarly, if $x=0$, then every factor $x_i x_j$ vanishes, hence
\[
    p(0,y)=0.
\]
Taking both variables equal to zero gives immediately
\[
    p(0,0)=0.
\]
This proves the claim.
\end{proof}

\newpage

\section{Inequivalence with Choi maps and the role of the extension}
\label{app:choi-inequivalence}

We distinguish two questions. For the zero extension used in this work, the rank obstruction of Section~\ref{sec:choi-inequivalence} excludes the entire
positive generalized Choi family under complex invertible input/output transformations. For arbitrary extensions, a different argument applies to the standard Choi map and real invertible transformations, provided that the
KSMZ parameter satisfies a strict positivity condition.

\subsection{Zero extensions}

Let $\Psi:S_3(\mathbb R)\to S_3(\mathbb R)$ be a real-linear positive map, and define
\begin{equation}
\Gamma_0(X)
=
\Psi_{\mathbb C}\left(\frac{X+X^T}{2}\right).
\end{equation}
Then
\begin{equation}
\{X\in M_3(\mathbb C):X^T=-X\}
\subseteq\ker\Gamma_0,
\end{equation}
so $\operatorname{rank}_{\mathbb C}(\Gamma_0)\leq6$.
Every positive generalized Choi map \eqref{eq:generalized-choi} has rank at least seven. Since rank is invariant under invertible input/output
congruences and nonzero rescaling, these maps cannot be equivalent.

This conclusion holds independently of whether the KSMZ perturbation parameter lies strictly inside the admissible interval or at an admissible endpoint. However, it relies on the zero extension and cannot
be inferred from the real symmetric restriction alone.

\subsection{Non-extremality in the strict perturbative regime}

Let $\mathcal P_{2,2}^{3,3}$ be the cone of real biquadratic forms nonnegative on $\mathbb R^3\times\mathbb R^3$. Write the real form produced by the KSMZ construction as
\begin{equation}
p_\delta(x,y)=H(x,y)+\delta F(x,y),
\qquad
H=\sum_j h_j(x,y)^2\neq0,
\label{eq:E-ksmz-form}
\end{equation}
where the $h_j$ are real bilinear forms.

Assume that $p_\delta$ is non-SOS and that there exists $\delta'>\delta>0$ such that
\begin{equation}
p_{\delta'}=H+\delta'F
\in\mathcal P_{2,2}^{3,3}.
\label{eq:E-strict-condition}
\end{equation}
This is the strict perturbative regime prescribed in Step~3 of Algorithm~4.1 in~\cite{Klep2017ManyMorePositiveMaps}.

\paragraph{Lemma.}
Under these assumptions, $p_\delta$ does not generate an extreme ray of $\mathcal P_{2,2}^{3,3}$.

\begin{proof}
Set $t=\delta/\delta'\in(0,1)$. Then
\begin{equation}
p_\delta=(1-t)H+t\,p_{\delta'}.
\label{eq:E-decomposition}
\end{equation}
Both summands belong to $\mathcal P_{2,2}^{3,3}$.

If $p_\delta$ generated an extreme ray, the nonzero form $H$ would be a positive multiple of $p_\delta$. Since $H$ is SOS, this would imply that $p_\delta$ is SOS, contrary to the assumption.
\end{proof}

\subsection{An obstruction for arbitrary extensions}

In the convention \eqref{eq:generalized-choi}, take the standard Choi map to be
\[
\Phi_{\mathrm{Ch}}=\Phi_{1,0,1}.
\]
Its associated real biquadratic form is
\begin{align}
p_{\mathrm{Ch}}(x,y)
&=y^T\Phi_{\mathrm{Ch}}(xx^T)y \notag\\
&=(x_1^2+x_3^2)y_1^2
 +(x_2^2+x_1^2)y_2^2
 +(x_3^2+x_2^2)y_3^2 \notag\\
&\quad
 -2x_1x_2y_1y_2
 -2x_1x_3y_1y_3
 -2x_2x_3y_2y_3.
\label{eq:E-choi-form}
\end{align}
This form generates an extreme ray of
$\mathcal P_{2,2}^{3,3}$
\cite{ChoiLam1977,Ha2013notesextremalitychoimap}.

Let $\Psi_\delta:S_3(\mathbb R)\to S_3(\mathbb R)$ be the map associated with $p_\delta$, so that
\begin{equation}
p_\delta(x,y)=y^T\Psi_\delta(xx^T)y.
\end{equation}
Consider any complex-linear extension
\[
\widehat\Psi_\delta:
M_3(\mathbb C)\longrightarrow M_3(\mathbb C)
\]
satisfying
\begin{equation}
\widehat\Psi_\delta(S)=\Psi_\delta(S)
\qquad\text{for every }S\in S_3(\mathbb R).
\label{eq:E-restriction}
\end{equation}
No assumption is imposed on its action on the skew-symmetric subspace. Such an extension is not necessarily positive; the following obstruction
applies regardless.

\paragraph{Proposition.}
Suppose that $p_\delta$ is non-SOS and satisfies
\eqref{eq:E-strict-condition}.
There are no
$A,B\in\operatorname{GL}_3(\mathbb R)$ and
$\lambda>0$ such that
\begin{equation}
\widehat\Psi_\delta(X)
=
\lambda A\,
\Phi_{\mathrm{Ch}}(BXB^T)\,A^T
\qquad\text{for every }X\in M_3(\mathbb C).
\label{eq:E-real-equivalence}
\end{equation}
In particular, no choice of skew-symmetric extension can reproduce the standard Choi map in the chosen real basis, or any of its real orthogonal
input/output transforms.

\begin{proof}
For real $x,y$, the matrix $xx^T$ belongs to
$S_3(\mathbb R)$. Thus \eqref{eq:E-restriction} gives
\begin{equation}
y^T\widehat\Psi_\delta(xx^T)y=p_\delta(x,y).
\end{equation}
If \eqref{eq:E-real-equivalence} held, then
\begin{equation}
p_\delta(x,y)
=
\lambda p_{\mathrm{Ch}}(Bx,A^Ty).
\label{eq:E-form-equivalence}
\end{equation}

The transformation
\[
p(x,y)\longmapsto\lambda p(Bx,A^Ty)
\]
is a linear automorphism of $\mathcal P_{2,2}^{3,3}$: it is invertible, and both it and its inverse preserve nonnegativity. It therefore preserves extreme rays.

The right-hand side of \eqref{eq:E-form-equivalence} generates an extreme ray, because $p_{\mathrm{Ch}}$ does. The left-hand side does not, by the preceding lemma. This is a contradiction.
\end{proof}

\subsection{Scope of the conclusions}

For the zero extension, inequivalence with the whole positive generalized Choi family holds under complex invertible input/output congruences.

For arbitrary extensions, the preceding proposition concerns the standard Choi map, real invertible input/output congruences, and parameters satisfying
\eqref{eq:E-strict-condition}. It does not establish the analogous statement for arbitrary complex unitary transformations or for every generalized Choi map. Indeed, complex changes of variables need not preserve the real rank-one
inputs used in the proof.

The strict condition must be justified for the retained parameter. It does not follow merely from failure of a numerical flatness test. One way to ensure it is to start from a certified admissible parameter $\bar\delta>0$ and choose
\[
\delta=\theta\bar\delta,
\qquad 0<\theta<1.
\]
Then $\delta'=\bar\delta$ satisfies
\eqref{eq:E-strict-condition}, since
\[
p_\delta=(1-\theta)H+\theta p_{\bar\delta}\geq0.
\]
The non-SOS assumption must also hold for the retained form, as ensured by the KSMZ construction. This prescription does not retroactively certify strictness for previously generated examples. No arbitrary-extension conclusion is asserted hereat a maximal admissible endpoint.

\newpage
\section{The Breuer-Hall symmetric core is SOS}
\label{app:breuer_hall_sos_core}

In this appendix we prove the algebraic fact used in
Proposition~\ref{prop:ksmz_not_breuer_hall}. Let \(d=2r\geq 4\), and let
\[
    J
    =
    \bigoplus_{\alpha=1}^{r}
    \begin{pmatrix}
        0 & 1 \\
        -1 & 0
    \end{pmatrix}.
\]
Thus \(J^T=-J\), \(J^TJ=I_d\), and \(J^2=-I_d\). Consider the canonical
Breuer-Hall map
\[
    \Phi_J(X)
    =
    \operatorname{Tr}(X)I_d
    -
    X
    -
    J X^T J^T .
\]
Its restriction to real symmetric matrices is
\[
    \Psi_J(S)
    =
    \operatorname{Tr}(S)I_d
    -
    S
    -
    J S J^T,
    \qquad
    S\in S_d(\mathbb R).
\]

\begin{proposition}
\label{prop:BH_symmetric_core_SOS}
The biquadratic form associated with the symmetric real restriction
\(\Psi_J\),
\[
    p_{\Psi_J}(x,y)
    =
    y^T\Psi_J(xx^T)y,
    \qquad
    x,y\in\mathbb R^d,
\]
is a sum of squares of real bilinear forms.
\end{proposition}

\begin{proof}
For a rank-one input \(S=xx^T\), we have
\[
    \operatorname{Tr}(xx^T)=\|x\|^2,
    \qquad
    Jxx^TJ^T = (Jx)(Jx)^T .
\]
Therefore
\begin{align}
    p_{\Psi_J}(x,y)
    &=
    y^T
    \left(
        \|x\|^2 I_d
        -
        xx^T
        -
        (Jx)(Jx)^T
    \right)
    y                                                        \\
    &=
    \|x\|^2\|y\|^2
    -
    (x^Ty)^2
    -
    (y^TJx)^2 .
    \label{eq:BH_real_core_polynomial}
\end{align}

We now identify \(\mathbb R^{2r}\) with \(\mathbb C^r\). Write
\[
    z_\alpha = x_{2\alpha-1}+ i x_{2\alpha},
    \qquad
    w_\alpha = y_{2\alpha-1}+ i y_{2\alpha},
    \qquad
    \alpha=1,\dots,r.
\]
Then
\[
    \|z\|_{\mathbb C^r}^2 = \|x\|_{\mathbb R^{2r}}^2,
    \qquad
    \|w\|_{\mathbb C^r}^2 = \|y\|_{\mathbb R^{2r}}^2,
\]
and
\[
    |\langle z,w\rangle_{\mathbb C}|^2
    =
    (x^Ty)^2+(y^TJx)^2 .
\]
Thus \eqref{eq:BH_real_core_polynomial} becomes
\begin{equation}
    p_{\Psi_J}(x,y)
    =
    \|z\|^2\|w\|^2
    -
    |\langle z,w\rangle|^2 .
\end{equation}
By the complex Lagrange identity,
\begin{equation}
    \|z\|^2\|w\|^2
    -
    |\langle z,w\rangle|^2
    =
    \sum_{1\leq \alpha<\beta\leq r}
    |z_\alpha w_\beta-z_\beta w_\alpha|^2 .
    \label{eq:complex_lagrange_identity}
\end{equation}
Each term in the right-hand side is the modulus square of a complex bilinear
form. Writing real and imaginary parts gives
\[
    |z_\alpha w_\beta-z_\beta w_\alpha|^2
    =
    \Re(z_\alpha w_\beta-z_\beta w_\alpha)^2
    +
    \Im(z_\alpha w_\beta-z_\beta w_\alpha)^2 ,
\]
which is a sum of two squares of real bilinear forms in \(x\) and \(y\).
Hence \(p_{\Psi_J}\) is a real SOS biquadratic form.
\end{proof}

\newpage
\section{Extremality, Optimality and Tightness of the EWs}
\label{app:optimality}

In this appendix we collect the geometric facts used in 
Section~\ref{subsec:location_of_the_maps} to discuss tightness, optimality, and 
extremality of the witnesses obtained from the KSMZ construction. We work with 
the cone of block-positive operators
\begin{equation}
    \mathrm{BP}(\mathcal H_A,\mathcal H_B)
    =
    \Bigl\{
    W=W^\dagger :
    \bra{x\otimes y}W\ket{x\otimes y}\geq 0
    \quad
    \forall\, x\in\mathcal H_A,\; y\in\mathcal H_B
    \Bigr\}.
\end{equation}
For a block-positive operator \(W\), we write
\begin{equation}
    q_W(x,y)
    =
    \bra{x\otimes y}W\ket{x\otimes y}
\end{equation}
for the associated biquadratic form, and
\begin{equation}
    \mathcal Z(W)
    =
    \Bigl\{
    \ket{x\otimes y} :
    q_W(x,y)=0
    \Bigr\}
\end{equation}
for its set of saturating product vectors. The set \(\mathcal Z(W)\) describes
the contact locus between the supporting hyperplane defined by \(W\) and the
separable cone.

\subsection{The hierarchy: extremal, optimal, tight}

We first justify the implication chain
\begin{equation}
    \text{extremal}
    \quad\Longrightarrow\quad
    \text{optimal}
    \quad\Longrightarrow\quad
    \text{tight}.
\end{equation}

\begin{lemma}[Optimal witnesses are tight]
Let \(W\) be an optimal entanglement witness. Then \(W\) is tight, i.e.
there exists at least one product vector \(\ket{x\otimes y}\) such that
\begin{equation}
    \bra{x\otimes y}W\ket{x\otimes y}=0.
\end{equation}
\end{lemma}

\begin{proof}
Assume by contradiction that \(W\) is not tight. Then
\begin{equation}
    \bra{x\otimes y}W\ket{x\otimes y}>0
\end{equation}
for every normalized product vector \(\ket{x\otimes y}\). Since the set of
normalized product vectors is compact, there exists \(m>0\) such that
\begin{equation}
    \bra{x\otimes y}W\ket{x\otimes y}\geq m
\end{equation}
for every normalized product vector. Hence, for any \(0<\varepsilon<m\),
\begin{equation}
    \bra{x\otimes y}(W-\varepsilon I)\ket{x\otimes y}
    =
    \bra{x\otimes y}W\ket{x\otimes y}
    -
    \varepsilon
    \geq 0.
\end{equation}
Thus \(W-\varepsilon I\) remains block-positive. But \(\varepsilon I\succeq 0\)
is a nonzero positive semidefinite operator, so one can subtract a nonzero
positive operator from \(W\) while preserving block positivity. This contradicts
optimality. Therefore \(W\) must be tight.
\end{proof}

\begin{lemma}[Extremal witnesses are optimal]
Let \(W\) be an extremal entanglement witness in the cone
\(\mathrm{BP}(\mathcal H_A,\mathcal H_B)\). Then \(W\) is optimal.
\end{lemma}

\begin{proof}
Recall that extremality in a cone means extremality of the ray generated by
\(W\): if
\begin{equation}
    W = W_1+W_2,
    \qquad
    W_1,W_2\in \mathrm{BP}(\mathcal H_A,\mathcal H_B),
\end{equation}
then \(W_1=\lambda_1 W\) and \(W_2=\lambda_2 W\) for some
\(\lambda_1,\lambda_2\geq 0\).

Assume by contradiction that \(W\) is not optimal. Then there exists a nonzero
positive semidefinite operator \(P\succeq 0\) such that
\begin{equation}
    W-P\in \mathrm{BP}(\mathcal H_A,\mathcal H_B).
\end{equation}
Since \(P\succeq 0\) is also block-positive, we have a decomposition
\begin{equation}
    W=(W-P)+P
\end{equation}
inside the block-positive cone. By extremality, \(P=\lambda W\) for some
\(\lambda\geq 0\). Since \(P\neq 0\), one has \(\lambda>0\), and therefore
\(W=\lambda^{-1}P\succeq 0\). This contradicts the fact that \(W\) is an
entanglement witness, since an entanglement witness is block-positive but not
positive semidefinite. Hence \(W\) is optimal.
\end{proof}

Combining the two lemmas gives
\begin{equation}
    \text{extremal}
    \quad\Longrightarrow\quad
    \text{optimal}
    \quad\Longrightarrow\quad
    \text{tight}.
\end{equation}
The converses fail in general: optimal witnesses need not be extremal, and
tight witnesses need not be optimal.

\subsection{Are the KSMZ witnesses optimal?}
\label{app:optimal?}
Optimality of entanglement witnesses is understood in the sense of~\cite{Lewenstein_2000}. See also the review~\cite{Chruscinski2014EntanglementWitnesses}.

\begin{proposition}[Necessary condition for improving a witness]
    Let $W$ be a block-positive operator on 
    $\mathcal{H}_A\otimes\mathcal{H}_B\simeq\mathbb{R}^{nm}$, i.e.
    \begin{equation*}
        \bra{x\otimes y}W\ket{x\otimes y}\geq 0 
        \quad \forall\, x\in\mathbb{R}^n,\, y\in\mathbb{R}^m.
    \end{equation*}
    Define the saturating product set
    \begin{equation*}
        \mathcal{Z}(W)
        =\Bigl\{\ket{x\otimes y} : \bra{x\otimes y}W\ket{x\otimes y}=0\Bigr\}.
    \end{equation*}
    Then any $P\succeq 0$ such that $W' = W - P$ remains block-positive must satisfy
    \begin{equation*}
        \operatorname{span}\{\mathcal{Z}(W)\}\subseteq\ker(P).
    \end{equation*}
    In particular, since Algorithm~\ref{sec:algoimpl} guarantees that
    $\{s_1,\dots,s_N\}\subseteq\mathcal{Z}(W)$, the operator $P$ must vanish on 
    $\operatorname{span}\{s_1,\dots,s_N\}$: the witness cannot be improved in those 
    directions.
    \label{prop:NC_Wit_imp}
\end{proposition}

\begin{proof}
Let $P\succeq 0$, $P\neq 0$, such that $W':=W-P$ is block-positive. 
Let $\ket{s}=\ket{x\otimes y}\in\mathcal{Z}(W)$, so that $\bra{s}W\ket{s}=0$.
Since $W'$ is block-positive, $\bra{s}W'\ket{s}\ge 0$, and therefore
\begin{align*}
0 \;\le\; \bra{s}W'\ket{s}
  &= \bra{s}(W-P)\ket{s}\\
  &= \underbrace{\bra{s}W\ket{s}}_{=\,0} - \bra{s}P\ket{s}
  = -\bra{s}P\ket{s}.
\end{align*}
Combined with $P\succeq 0$, which gives $\bra{s}P\ket{s}\ge 0$, we conclude 
$\bra{s}P\ket{s}=0$ for all $\ket{s}\in\mathcal{Z}(W)$.
Writing $P=R^\dagger R$, this gives $\|R\ket{s}\|^2=0$, hence $P\ket{s}=0$.
Since this holds for every element of $\mathcal{Z}(W)$, we obtain
\begin{equation*}
    \operatorname{span}\{\mathcal{Z}(W)\}\subseteq\ker(P),
\end{equation*}
which concludes the proof.
\end{proof}

The proposition shows that the prescribed zeros \cite{Klep2017ManyMorePositiveMaps}
$\{s_1,\dots,s_N\}$ act as fixed contact points between the witness
hyperplane and the separable cone. Any positive operator $P\succeq 0$
that could be subtracted from $W$ while preserving block-positivity must
annihilate all these directions. Consequently, if the prescribed product
vectors already span the full bipartite space,
\begin{equation}
    \operatorname{span}\{s_1,\dots,s_N\}
    =
    \mathcal H_A\otimes\mathcal H_B,
\end{equation}
then $\ker(P)=\mathcal H_A\otimes\mathcal H_B$, hence $P=0$, and the
witness is optimal by the spanning property \cite{Lewenstein_2000}.

Therefore, for the KSMZ witnesses, optimality can be certified by checking
the spanning property of the prescribed zero set. If this spanning condition
fails, the proposition does not imply non-optimality; it only shows that any
possible improvement must be supported on the orthogonal complement of
$\operatorname{span}\{s_1,\dots,s_N\}$. In that case, optimality remains
undecided and would require either an analysis of the full zero set
$\mathcal Z(W)$ or another independent argument.

\subsection{Non-extremality of the perturbative KSMZ witnesses}
\label{app:non_extremality_ksmz}

The previous subsection concerns optimality, namely whether a positive
semidefinite operator can be subtracted from $W$ while preserving
block-positivity. Extremality is a different and stronger geometric property:
it asks whether $W$ generates an extremal ray of the whole block-positive
cone \cite{Chruscinski2014EntanglementWitnesses}. We now show that, for the perturbative KSMZ construction, extremality
fails whenever the parameter $\delta$ is chosen strictly inside the admissible
positivity interval.

\begin{proposition}[Non-extremality away from the endpoint]
\label{prop:ksmz_non_extremal}
Let
\[
    q_\delta = H+\delta f
\]
be a KSMZ qudaratic form \cite{Klep2017ManyMorePositiveMaps}, where \(H=\sum_i h_i^2\) is the reference sum-of-squares term 
and \(f\) is the perturbation. Assume that \(q_\delta\) is nonnegative on the 
Segre variety and that there exists \(\delta'>\delta\) such that
\[
    q_{\delta'} = H+\delta' f
\]
is still nonnegative on the Segre variety. Then the associated entanglement 
witness \(W_\delta\) is not extremal in the cone of block-positive operators.
In particular, this applies to every strict choice \(0<\delta<\delta_0\) in the 
standard KSMZ perturbative construction.
\end{proposition}

\begin{proof}
Since \(q_{\delta'}=H+\delta'f\), we can write
\begin{equation}
    q_\delta
    =
    \left(1-\frac{\delta}{\delta'}\right)H
    +
    \frac{\delta}{\delta'}q_{\delta'}.
    \label{eq:ksmz_qdelta_segment}
\end{equation}
Indeed,
\begin{align}
    \left(1-\frac{\delta}{\delta'}\right)H
    +
    \frac{\delta}{\delta'}q_{\delta'}
    &=
    \left(1-\frac{\delta}{\delta'}\right)H
    +
    \frac{\delta}{\delta'}(H+\delta'f) \\
    &=
    H+\delta f \\
    &=
    q_\delta.
\end{align}
Since \(0<\delta<\delta'\), one has
\[
    0<\frac{\delta}{\delta'}<1.
\]
Thus \eqref{eq:ksmz_qdelta_segment} expresses \(q_\delta\) as a convex 
combination of two forms, \(H\) and \(q_{\delta'}\).

Both \(H\) and \(q_{\delta'}\) are nonnegative on the Segre variety. After 
pullback to product directions, this means that the corresponding biquadratic 
forms are nonnegative on all product vectors. Equivalently, the associated 
Choi operators \(W_H\) and \(W_{\delta'}\) are block-positive.

By linearity of the pullback and of the Choi correspondence, 
\eqref{eq:ksmz_qdelta_segment} gives
\begin{equation}
    W_\delta
    =
    \left(1-\frac{\delta}{\delta'}\right)W_H
    +
    \frac{\delta}{\delta'}W_{\delta'}.
    \label{eq:ksmz_Wdelta_segment}
\end{equation}
Hence \(W_\delta\) lies on the line segment joining \(W_H\) and 
\(W_{\delta'}\) inside the block-positive cone.

This segment is nontrivial. Indeed, \(W_H\), \(W_\delta\), and 
\(W_{\delta'}\) correspond respectively to the forms
\[
    H,
    \qquad
    H+\delta f,
    \qquad
    H+\delta'f.
\]
Since \(f\neq 0\) and \(\delta\neq\delta'\), these forms are distinct. Moreover,
in the KSMZ construction \(f\) is chosen outside the linear span generated by 
the products \(h_i h_j\), whereas \(H=\sum_i h_i^2\) belongs to that span. Thus 
\(H\) and \(f\) are linearly independent, so the above forms are not 
proportional.

Therefore \eqref{eq:ksmz_Wdelta_segment} is a genuine convex decomposition of 
\(W_\delta\) into two non-proportional block-positive operators. Consequently, 
\(W_\delta\) cannot generate an extremal ray of the block-positive cone. Hence 
\(W_\delta\) is not extremal.
\end{proof}

\begin{remark}
The argument relies on the existence of a larger admissible perturbation 
parameter \(\delta'>\delta\). Therefore it applies to the standard strict 
perturbative regime \(0<\delta<\delta_0\). It does not decide what happens at a 
maximal endpoint
\[
    \delta_*=\sup\{t>0:\ H+t f \geq 0 \text{ on the Segre variety}\}.
\]
At such an endpoint, extremality would require a separate zero and Hessian-zero in the sense of \cite{hansen2015extremalentanglementwitnesses}
analysis.
\end{remark}

\newpage
\section{Numerical precision and reliability of the computations}
\label{app:numerics}

\subsection{Precision of $\mathrm{Tr}(W_D\rho)$}

Let
\[
\widehat t := \operatorname{fl}\!\bigl(\operatorname{Tr}(W_D\rho)\bigr)
\]
denote the floating-point value returned by \textsc{MATLAB}, and
\[
t := \operatorname{Tr}(W_D\rho)
\]
the exact mathematical quantity. We write
\[
\widehat t = t + \delta_{\mathrm{Tr}},
\]
where $\delta_{\mathrm{Tr}}$ is the absolute forward rounding error.

All computations are performed in IEEE double precision, with machine precision
\[
\varepsilon_{\mathrm{mach}} = 2^{-52} \approx 2.22\times 10^{-16}.
\]

A standard result from backward error analysis of matrix multiplication~\cite{higham2002accuracy} states that the computed product satisfies
\[
\operatorname{fl}(W_D\rho) = W_D\rho + \Delta, 
\qquad
\|\Delta\|_F \le \gamma_n \|W_D\|_F \|\rho\|_F,
\]
where
\[
\gamma_n := \frac{n\,\varepsilon_{\mathrm{mach}}}{1 - n\,\varepsilon_{\mathrm{mach}}}.
\]
Taking the trace, which is a linear operator of norm $1$ with respect to the Frobenius inner product, yields
\[
|\delta_{\mathrm{Tr}}| \le \|\Delta\|_F \le \gamma_n \|W_D\|_F \|\rho\|_F.
\]

In the present case ($n=9$), using the matrices for the script \textit{A complete numerical demonstration of indecomposability}~\cite{GithubCode} :

\[
\rho = \texttt{new\_rho\_temp}, 
\qquad 
W_D = \texttt{W\_phi},
\]

\[
\rho =
\begin{pmatrix}
0.1448 & 0.0512 & 0.0028 & 0.1788 & 0.0517 & -0.0053 & -0.0070 & -0.0031 & 0.0033 \\
0.0512 & 0.0297 & 0.0023 & 0.0517 & 0.0431 & 0.0003 & -0.0031 & 0.0248 & 0.0065 \\
0.0028 & 0.0023 & 0.0097 & -0.0053 & 0.0003 & -0.0006 & 0.0033 & 0.0065 & 0.0127 \\
0.1788 & 0.0517 & -0.0053 & 0.4027 & 0.2021 & 0.0736 & 0.0422 & 0.0464 & 0.0262 \\
0.0517 & 0.0431 & 0.0003 & 0.2021 & 0.2296 & 0.0826 & 0.0464 & 0.1303 & 0.0438 \\
-0.0053 & 0.0003 & -0.0006 & 0.0736 & 0.0826 & 0.0412 & 0.0262 & 0.0438 & 0.0181 \\
-0.0070 & -0.0031 & 0.0033 & 0.0422 & 0.0464 & 0.0262 & 0.0188 & 0.0255 & 0.0159 \\
-0.0031 & 0.0248 & 0.0065 & 0.0464 & 0.1303 & 0.0438 & 0.0255 & 0.0976 & 0.0340 \\
0.0033 & 0.0065 & 0.0127 & 0.0262 & 0.0438 & 0.0181 & 0.0159 & 0.0340 & 0.0260
\end{pmatrix}
\]

\[
W_D =
\begin{pmatrix}
0.6414 & -0.8063 & -0.4999 & -0.3104 & 0.2598 & 0.1416 & 0.2522 & -0.1127 & -0.0358 \\
-0.8063 & 1.0277 & 0.6321 & 0.2598 & -0.1598 & -0.0768 & -0.1127 & -0.1332 & -0.1222 \\
-0.4999 & 0.6321 & 0.4903 & 0.1416 & -0.0768 & 0.0166 & -0.0358 & -0.1222 & -0.1891 \\
-0.3104 & 0.2598 & 0.1416 & 0.1566 & -0.0804 & -0.0174 & -0.1302 & 0.0276 & -0.0450 \\
0.2598 & -0.1598 & -0.0768 & -0.0804 & 0.0635 & -0.0070 & 0.0276 & -0.0260 & 0.0298 \\
0.1416 & -0.0768 & 0.0166 & -0.0174 & -0.0070 & 0.0519 & -0.0450 & 0.0298 & -0.0236 \\
0.2522 & -0.1127 & -0.0358 & -0.1302 & 0.0276 & -0.0450 & 0.1271 & 0.0069 & 0.0777 \\
-0.1127 & -0.1332 & -0.1222 & 0.0276 & -0.0260 & 0.0298 & 0.0069 & 0.0758 & 0.0023 \\
-0.0358 & -0.1222 & -0.1891 & -0.0450 & 0.0298 & -0.0236 & 0.0777 & 0.0023 & 0.0905
\end{pmatrix}
\]

one obtains
\[
\|W_D\|_F \approx 2.05,
\qquad
\|\rho\|_F \approx 0.55,
\]
and therefore
\[
|\delta_{\mathrm{Tr}}| \lesssim 
\gamma_9 \|W_D\|_F \|\rho\|_F
\approx 2.2\times 10^{-15}.
\]

Hence, any computed value of order $10^{-6}$ is several orders of magnitude larger than the maximal floating-point uncertainty; in particular, its sign is numerically certified.

\subsection{Hermitian matrix $\rho$ and the eigenvalue computation}

The matrix $\rho$ is Hermitian. MATLAB uses LAPACK in some linear algebra functions such as \texttt{eig}
(\href{https://www.mathworks.com/help/coder/ug/lapack-calls-in-generated-code.html}{MathWorks documentation}).
For Hermitian/symmetric eigenproblems, LAPACK’s User Guide~\cite{lapack_users_guide} states that the computed eigenvalues satisfy a bound of the form
\[
|\widehat\lambda_i-\lambda_i| \le \mathrm{EERRBD},
\qquad
\mathrm{EERRBD} = \varepsilon_{\mathrm{mach}}\,\mathrm{ANORM},
\]
with corresponding bounds for the eigenvectors as well
(\href{https://www.netlib.org/lapack/lug/node89.html}{LAPACK User Guide, ``Error Bounds for the Symmetric Eigenproblem''}).

More precisely, the LAPACK routine documentation for \texttt{DSYEVR} states~\cite{lapack_dsyevr} that an approximate eigenvalue is considered converged when it is located in an interval $[a,b]$ of width at most
\[
\mathrm{ABSTOL} + \mathrm{EPS}\max(|a|,|b|),
\]
where $a$ and $b$ are the bracketing endpoints produced by the iterative refinement, $\mathrm{EPS}$ is machine precision, and if $\mathrm{ABSTOL}\le 0$ then $\mathrm{EPS}\,|T|$ is used instead, with $|T|$ the $1$-norm of the tridiagonal matrix obtained after reduction to tridiagonal form
(\href{https://www.netlib.org/lapack/explore-html/d1/d56/group__heevr_gaa334ac0c11113576db0fc37b7565e8b5.html}{LAPACK \texttt{DSYEVR} documentation}).

For the matrix $\rho$, we have
\[
\|\rho\|_2 \approx 0.665,
\]
and the smallest computed eigenvalue is
\[
\lambda_{\min}(\rho) \approx 5.7 \times 10^{-9}.
\]

A conservative a priori bound on the absolute eigenvalue error is therefore
\[
\tau_\rho := 9\,\varepsilon_{\mathrm{mach}} \|\rho\|_2 
\approx 9 \times (2.22\times 10^{-16}) \times 0.665
\approx 1.3\times 10^{-15}.
\]

\paragraph{Conclusion.}
We obtain
\[
\lambda_{\min}(\rho) \approx 5.7\times 10^{-9}
\;\gg\;
\tau_\rho \approx 10^{-15}.
\]
Thus, even the smallest eigenvalue remains positive by a margin of approximately six orders of magnitude beyond the worst-case floating-point uncertainty.

In addition, $\rho$ was certified to be numerically positive by the backward-stable Cholesky decomposition algorithm of MATLAB.

\newpage

\section{List of the Criterion of Entanglement and Separability test operated by IsSeparable}
\label{app:IsSeparable Criterion List}

\textbf{[E]} denotes an entanglement criterion whereas \textbf{[S]} denotes a separability criteria.

\begin{enumerate}
\item \textbf{[E] PPT Criterion}: The most fundamental test; if the partial transpose of the matrix has any negative eigenvalues, the state is entangled \cite{Peres1996}.
\item \textbf{[S] Low-Dimensional PPT Sufficiency}: For systems of dimension $2 \otimes 2$ or $2 \otimes 3$, being PPT is not just necessary but also sufficient for separability \cite{Horodecki1996}.
\item \textbf{[S] Low-Rank PPT Sufficiency}: A PPT state is guaranteed to be separable if its rank is low (specifically $\leq 3$) or does not exceed the rank of its reduced density matrices \cite{Horodecki2000}.
\item \textbf{[E] Realignment (CCN) Criterion}: Also known as the Computable Cross Norm; if the trace norm of the realigned matrix is greater than 1, the state is entangled \cite{Chen2003}.
\item \textbf{[E] Stronger Realignment (Zhang et al.)}: A more sensitive version of the realignment test that uses the local properties of the subsystems to detect entanglement \cite{Zhang2008}.
\item \textbf{[S] Qubit-Qudit Spectrum Check}: A specific test for $2 \otimes n$ systems that determines separability based solely on the eigenvalues of the state \cite{Johnston2013}.
\item \textbf{[S] Block Hankel Structure}: In $2 \otimes n$ cases, if the blocks of the matrix satisfy the symmetry of a block Hankel matrix, the state is separable \cite{Hildebrand2008}.
\item \textbf{[S] Perturbed Block Hankel}: An extension for $2 \otimes n$ systems where the matrix is a rank-1 perturbation of a block Hankel matrix \cite{Hildebrand2005}.
\item \textbf{[S] Homothetic Images Approach}: A geometric check for $2 \otimes n$ systems determining if the state belongs to specific separable cones \cite{Hildebrand2016}.
\item \textbf{[S] Block Matrix Lemma (Lemma 1)}: A sufficient condition for $2 \otimes n$ states involving the relationship between the blocks of the bipartite matrix \cite{Johnston2013}.
\item \textbf{[S/E] Chow Form Check}: Specifically for $3 \otimes 3$ states of rank 4, this calculates the "Chow Form" determinant; a zero result proves separability, while a non-zero result proves entanglement \cite{Chen2013}.
\item \textbf{[S] In-Separable Ball}: Determines if the state is close enough to the maximally mixed state to fall within the guaranteed "separable ball" \cite{Gurvits2002}.
\item \textbf{[S] Identity Rank-1 Perturbation}: If the state is a small rank-1 perturbation of the identity and is PPT, it is proven separable \cite{Vidal1999}.
\item \textbf{[S] Operator Schmidt Rank}: If the operator Schmidt rank (the Schmidt rank of the matrix when viewed as a vector) is $\leq 2$, the state is separable \cite{Cariello2013}.
\item \textbf{[E] Exposed Qutrit Maps}: A test for $3 \otimes 3$ states using a specific family of positive linear maps that detect entanglement where others fail \cite{Ha2011}.
\item \textbf{[E] Breuer-Hall Maps}: Entanglement detection in even dimensions using positive maps constructed from antisymmetric unitary matrices \cite{Breuer2006, Hall2006}.
\item \textbf{[E] Filter Covariance Matrix Criterion (Filter CMC)}: A sophisticated test that applies the Covariance Matrix Criterion after bringing the state to its "Filter Normal Form" \cite{Gittsovich2008}.
\item \textbf{[E] PPT Symmetric Extension (Outer)}: Part of the Doherty-Parrilo-Spedalieri hierarchy; if a state lacks a $k$-copy PPT symmetric extension, it is entangled \cite{Doherty2004CompleteFamily}.
\item \textbf{[S] Symmetric Inner Extension}: A semidefinite programming check that looks for a symmetric extension within the separable cone to prove separability \cite{Navascues2009}.
\end{enumerate}

\newpage
\section{PPT margin, distance to the decomposable cone, and numerical interpretation}
\label{app:ppt_margin_distance_decomp}

In this appendix we explain in more detail the quantitative meaning of the two SDPs used in the numerics, namely the PPT margin \eqref{eq:PPT_margin} and the operator-norm distance from the witness $W$ to the cone of decomposable witnesses \eqref{eq:distance_dec}.
These two quantities are closely related by cone duality.

\subsection{The two cones}
Let
\begin{equation}
\mathrm{PPT}
=
\{\rho \succeq 0 : \rho^{T_B}\succeq 0,\ \text{Tr}(\rho)=1\}    
\end{equation}
denote the set of normalized PPT states, and let
\begin{equation}
\widetilde{\mathrm{PPT}}
=
\{\rho \succeq 0 : \rho^{T_B}\succeq 0\}
\end{equation}

be the corresponding unnormalized PPT cone. On the witness side, define the cone of decomposable witnesses by
\begin{equation}
\text{Dec}
=
\{\,P+Q^{T_B} : P,Q \succeq 0\,\}.
\end{equation}

A witness $W$ is decomposable if and only if $W\in Dec$.

\subsection{Why PPT states are the dual cone}
The key structural fact is that $\text{Dec}$ and $\widetilde{\mathrm{PPT}}$ are dual to one another with respect to the Hilbert--Schmidt pairing:
\begin{equation}
\text{Dec}^*
=
\widetilde{\mathrm{PPT}}.
\end{equation}
Indeed, if $\rho \succeq 0$ and $\rho^{T_B}\succeq 0$, then for every decomposable witness
\begin{equation}
W_D=P+Q^{T_B}, \qquad P,Q\succeq 0,
\end{equation}
one has
\begin{equation}
 \text{Tr}(W_D\rho)= \text{Tr}(P\rho)+ \text{Tr}(Q\rho^{T_B})\ge 0.
\end{equation}
Thus every PPT operator defines a nonnegative functional on the decomposable cone. Conversely, if
\begin{equation}
 \text{Tr}(W_D\rho)\ge 0 \qquad \forall W_D\in Dec,
\end{equation}
then testing separately against $P\succeq 0$ and $Q^{T_B}$ forces both $\rho\succeq 0$ and
$\rho^{T_B}\succeq 0$. Therefore the dual cone of decomposable witnesses is exactly the PPT cone.

\subsection{PPT margin as a quantitative certificate of indecomposability}
The quantity
\begin{equation}
\mu(W):=-\min_{\rho\in\mathrm{PPT}} \text{Tr}(W\rho)
\end{equation}
measures how far $W$ penetrates the dual cone separation. If $\mu(W)>0$, then there exists a PPT state $\rho$ such that
\begin{equation}
 \text{Tr}(W\rho)<0,
\end{equation}
which proves that $W\notin \text{Dec}$, hence that $W$ is indecomposable. In this sense,
$\mu(W)$ is a quantitative witness of indecomposability: the larger $\mu(W)$, the more strongly $W$ separates some PPT state from the decomposable cone.

The corresponding SDP is
\begin{equation}
\begin{aligned}
    \text{minimize}\quad &  \text{Tr}(W\rho)\\
    \text{subject to}\quad & \rho \succeq 0,\\
    & \rho^{T_B}\succeq 0,\\
    &  \text{Tr}(\rho)=1.
\end{aligned}
\end{equation}
This optimization is convex, numerically stable, and directly returns a PPT state minimizing the expectation value of $W$.

\subsection{Distance to the decomposable cone}
A complementary quantity is the operator-norm distance from $W$ to the decomposable cone:
\begin{equation}
d_\infty(W,\text{Dec})
=
\inf_{W_D\in \text{Dec}}\|W-W_D\|_\infty.
\end{equation}
Using the representation $W_D=P+Q^{T_B}$ with $P,Q\succeq 0$, this becomes
\begin{equation}
d_\infty(W,\text{Dec})
=
\min_{P,Q\succeq 0}\|W-(P+Q^{T_B})\|_\infty.
\end{equation}
In practice, this is implemented through the SDP
\begin{equation}
\begin{aligned}
\text{minimize}\quad & t\\
\text{subject to}\quad & P \succeq 0,\quad Q \succeq 0,\\
& -tI \preceq W-(P+Q^{T_B}) \preceq tI.
\end{aligned}
\end{equation}
The optimal value $t^\star$ is exactly the smallest operator-norm perturbation needed to move $W$ into the decomposable cone.

\subsection{Why $\mu(W)$ is a lower bound on the distance}
Let $W_D\in \text{Dec}$ be arbitrary and let $\rho\in \mathrm{PPT}$. Since $W_D$ is decomposable and $\rho$ is PPT,
\begin{equation}
 \text{Tr}(W_D\rho)\ge 0.
\end{equation}
Therefore

\begin{equation}
- \text{Tr}(W\rho)=\text{Tr}((W_D-W)\rho)- \text{Tr}(W_D\rho) \le |\text{Tr}((W-W_D)\rho)|.
\end{equation}

Because $\rho\succeq 0$ and $\Tr(\rho)=1$, one has $|\rho\|_1=1$, hence

\begin{equation}
| \text{Tr}((W-W_D)\rho)|\le \|W-W_D\|_\infty.
\end{equation}

Taking first the supremum over PPT states, and then the infimum over decomposable $W_D$, yields
\begin{equation}
\mu(W)\le d_\infty(W,Dec).
\end{equation}
Thus the PPT margin always gives a rigorous lower bound on the distance to the decomposable cone.

\subsection{When equality holds}
The previous inequality is in fact tight in the regime relevant to our numerics. Since the operator norm is dual to the trace norm on Hermitian matrices, standard conic duality yields
\begin{equation}
d_\infty(W,\text{Dec})
=
\max_{\substack{\rho\succeq 0\\ \rho^{T_B}\succeq 0\\  \text{Tr}(\rho)\le 1}}
\bigl(- \text{Tr}(W\rho)\bigr).
\end{equation}

Now, if the optimum is positive, then the maximizing $\rho$ necessarily saturates the trace constraint, i.e. $\text{Tr}(\rho)=1$, because scaling up a feasible $\rho$ with $\text{Tr}(\rho)<1$ would strictly improve the objective. Therefore, whenever $W$ actually detects PPT entanglement, one obtains

\begin{equation}
d_\infty(W,\text{Dec})
=
\mu(W).
\end{equation}

Equivalently,
\begin{equation}
d_\infty(W,\text{Dec})=\max\{0,\mu(W)\}.
\end{equation}
In particular, in the nontrivial regime $\mu(W)>0$, both SDPs return the same value.

\subsection{Numerical interpretation}
These two optimizations therefore provide complementary but ultimately equivalent information in the indecomposable regime. The PPT-margin SDP answers the question: does the witness attain a negative value on the PPT set, and by how much? The distance SDP answers the question: how far is the witness from becoming decomposable, measured in operator norm? When $\mu(W)>0$, the two viewpoints coincide quantitatively.

This is the reason why averaging these values over a family of KSMZ witnesses gives a meaningful numerical characterization of their average indecomposability: if the averaged values are small, then the witnesses typically lie close to the decomposable boundary; if they are larger, then the family penetrates more deeply into the indecomposable region.

\begin{figure}[h!]
    \centering
    \includegraphics[width=1\linewidth]{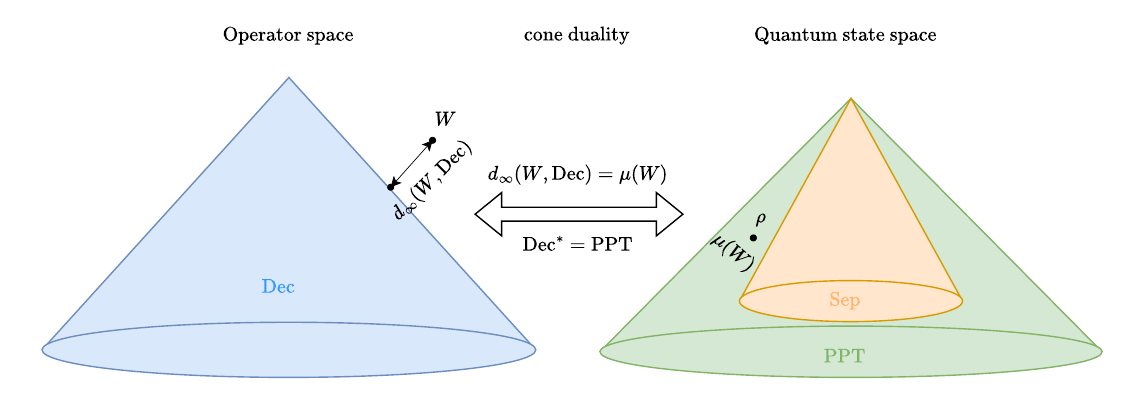}
    \caption{Geometric interpretation of decomposability and PPT detection. 
Left: in operator space, $d_\infty(W,\mathrm{Dec})$ denotes the operator-norm distance from the witness $W$ to the decomposable cone. 
Right: in state space, the PPT margin $\mu(W)=-\min_{\rho\in \mathrm{PPT}_1}\Tr(W\rho)$ measures the maximal violation of \(W\) on normalized PPT states. An optimal state $\rho$ defines the supporting functional associated with this violation. The relation between the two quantities is governed by the cone duality $\mathrm{Dec}^*=\mathrm{PPT}$.}
    \label{fig:placeholder}
\end{figure}

\end{document}